\documentclass[11pt,a4paper]{article}
\usepackage{amsmath}
\usepackage{amssymb}
\usepackage{amsthm}
\usepackage{thmtools}
\usepackage{hyperref}
\usepackage{cleveref}
\usepackage{fullpage}
\usepackage{microtype}
\usepackage{caption}
\usepackage{bbm}
\usepackage{hyperref, color}
\hypersetup{colorlinks=true,citecolor=blue, linkcolor=blue, urlcolor=blue}
\usepackage[linesnumbered,boxed,ruled,vlined]{algorithm2e}
\usepackage{bm}
\usepackage{bbm}
\usepackage[numbers]{natbib}
\usepackage{xcolor}
\usepackage{enumerate} 
\usepackage{enumitem}
\usepackage{tabularx}
\usepackage{array}
\usepackage{cleveref}
\usepackage{tcolorbox}
\usepackage[thinlines]{easytable}
\usepackage{tikz}
\usepackage{float}
\usetikzlibrary{decorations.pathmorphing, arrows.meta}
\newcolumntype{L}[1]{>{\raggedright\arraybackslash}p{#1}}
\newcolumntype{C}[1]{>{\centering\arraybackslash}m{#1}}
\newcolumntype{R}[1]{>{\raggedleft\arraybackslash}p{#1}}

\usepackage{makecell}
\usepackage{footnote}
\usepackage[ruled, vlined]{algorithm2e}
\makesavenoteenv{tabular}

\newcommand{\DTV}[2]{d_{\mathrm{TV}}\left({#1},{#2}\right)}

\newcommand{\e}{\mathrm{e}}

\renewcommand{\epsilon}{\varepsilon}

\newcommand{\mixtl}{T_{\textnormal{mix}}^{\textnormal{tilted}}}
\newcommand{\mixmax}{T_{\textnormal{mix-}\*1_V}^{\textnormal{GD}}}
\newcommand{\mixgd}{T_{\textnormal{mix}}^{\textnormal{GD}}}
\newcommand{\mixfd}{T_{\textnormal{mix}}^{\textnormal{FD}}}

\newcommand{\normone}[1]{\left\Vert #1 \right\Vert_1}

\newcommand{\mc}{\preceq_{\textnormal{mc}}}
\newcommand{\sd}{\preceq_{\textnormal{sd}}}

\newcommand{\psim}{P_{\textnormal{s-GD}}}

\newtheorem{theorem}{Theorem}[section]

\newtheorem{observation}[theorem]{Observation}

\newtheorem*{claim*}{Claim}

\newtheorem{lemma}[theorem]{Lemma}
\newtheorem{proposition}[theorem]{Proposition}
\newtheorem{corollary}[theorem]{Corollary}
\theoremstyle{definition}

\newtheorem{definition}[theorem]{Definition}
\newtheorem{remark}[theorem]{Remark}
\newtheorem*{remark*}{Remark}

\def\Pr{\mathop{\mathbf{Pr}}\nolimits}

\renewcommand{\emptyset}{\varnothing}

\newcommand{\abs}[1]{\left\vert#1\right\vert}

 \newcommand{\tuple}[1]{\left(#1\right)} \newcommand{\eps}{\varepsilon}
 
 \newcommand{\tp}{\tuple}

\newcommand{\defeq}{\triangleq}
\renewcommand{\d}{\,\-d}

\def\*#1{\mathbf{#1}} % Use \*A for \mathbf{A}
\def\+#1{\mathcal{#1}} % Use \+A for \mathcal{A}
\def\-#1{\mathrm{#1}} % Use \-A for \mathrm{A}

\usepackage{todonotes}
\usepackage{xifthen}

\renewcommand{\Pr}[2][]{ \ifthenelse{\isempty{#1}}
  {\mathbf{Pr}\left[#2\right]} {\mathbf{Pr}_{#1}\left[#2\right]} } % Use \Pr[a]{b} for \mathbf{Pr}_a[b], \Pr{b} for  \mathbf{Pr}[b]
\newcommand{\E}[2][]{ \ifthenelse{\isempty{#1}}
  {\mathbf{\mathbf{E}}\left[#2\right]}
  {\mathbf{\mathbf{E}}_{#1}\left[#2\right]} }
  \newcommand{\Var}[2][]{ \ifthenelse{\isempty{#1}}
  {\mathbf{\mathbf{Var}}\left[#2\right]}
  {\mathbf{\mathbf{Var}}_{#1}\left[#2\right]} }

\crefname{theorem}{Theorem}{Theorems}
\crefname{observation}{Observation}{Observations}
\crefname{claim}{Claim}{Claims}
\crefname{condition}{Condition}{Conditions}
\crefname{algorithm}{Algorithm}{Algorithms}
\crefname{property}{Property}{Properties}
\crefname{example}{Example}{Examples}
\crefname{fact}{Fact}{Facts}
\crefname{lemma}{Lemma}{Lemmas}
\crefname{corollary}{Corollary}{Corollaries}
\crefname{definition}{Definition}{Definitions}
\crefname{remark}{Remark}{Remarks}
\crefname{proposition}{Proposition}{Propositions}
\crefname{equation}{equation}{equations}
\crefname{enumi}{}{}
\crefname{enumii}{}{}
\crefname{enumiii}{}{}
\crefname{enumiv}{}{}
\creflabelformat{enumi}{#2#1#3}
\creflabelformat{enumii}{#2#1#3}
\creflabelformat{enumiii}{#2#1#3}
\creflabelformat{enumiv}{#2#1#3}

\crefformat{enumi}{#2#1#3}
\crefformat{enumii}{#2#1#3}
\crefformat{enumiii}{#2#1#3}
\crefformat{enumiv}{#2#1#3}
\crefrangeformat{enumi}{#3#1#4--#5#2#6}
\crefrangeformat{enumii}{#3#1#4--#5#2#6}
\crefrangeformat{enumiii}{#3#1#4--#5#2#6}
\crefrangeformat{enumiv}{#3#1#4--#5#2#6}

\title{Rapid Mixing of Glauber Dynamics for Monotone Systems\\ via Entropic Independence}

\newboolean{conf}
\setboolean{conf}{false}

\ifthenelse{\boolean{conf}}{\author{Author(s)}}{\author{Weiming Feng \footnote{School of Computing and Data Science, The University of Hong Kong. Emails: \texttt{wfeng@hku.hk} and \texttt{ymjessen02@connect.hku.hk}.} \and Minji Yang \footnotemark[1]}
}

\date{}

\begin{document}

\maketitle

\begin{abstract}
We study the mixing time of Glauber dynamics on monotone systems.
For monotone systems satisfying the entropic independence condition, we prove a new mixing time comparison result for Glauber dynamics.
For concrete applications, we obtain mixing time bounds from specified extremal configurations: $\tilde{O}(n)$ for the random cluster model induced by the ferromagnetic Ising model with consistently biased external fields, and $\tilde{O}(n^2)$ for the bipartite hardcore model under the one-sided uniqueness condition, where $n$ is the number of variables in the corresponding models, improving the best known results in [Chen and Zhang, SODA'23] and [Chen, Liu, and Yin, FOCS'23], respectively.

Our proof combines ideas from the stochastic dominance argument in the classical censoring inequality and the recently developed high-dimensional expanders. The key step in the proof is a novel comparison result between the Glauber dynamics and the field dynamics for monotone systems.
\end{abstract}
\section{Introduction}

Sampling from high-dimensional distributions is a central problem in computing and data science. 
A widely used sampling technique is the Markov chain Monte Carlo (MCMC) method. 
The \emph{Glauber dynamics} (also known as \emph{Gibbs sampling}) is one of the most fundamental sampling methods in MCMC. Given a distribution $\mu$ over $\{0,1\}^V$, the Glauber dynamics starts from an arbitrary feasible \emph{configuration} $X \in \{0,1\}^V$ and iteratively updates $X$ as follows:
\begin{itemize}
    \item select a variable $v \in V$ uniformly at random;
    \item resample the value of $v$ from $\mu_v^{X_{V \setminus \{v\}}}$, where $\mu_v^{X_{V \setminus \{v\}}}$ is the marginal distribution of $\mu$ on $v$ conditional on the current configuration $X_{V \setminus \{v\}}$ of the other variables.
\end{itemize}
One central problem for the Glauber dynamics is to understand its convergence rate, which is captured by the notion of the \emph{mixing time}. The mixing time is the number of iterations that the Glauber dynamics needs to make the distribution of the configuration $X$ be close to the target distribution $\mu$. 
For various kinds of distributions, the mixing time of the Glauber dynamics was of great interest and extensively studied~\cite{levin2017markov}.

In this paper, we focus on the case where the target distribution $\mu$ is a \emph{monotone system}. 
Monotone systems are an important class of high-dimensional distributions and are well studied in the literature on sampling algorithms~(e.g. see \cite{JerrumS93,ProppW96,MS04,peres2013can,MS13,RSVVY14,GJ18,HSZ19}).
For any $X,Y \in \{0,1\}^S$, where $S \subseteq V$, define the \emph{partial ordering}:
\begin{align}\label{eq:order1}
    X \preceq Y \text{ if and only if } X_v\leq Y_v \text{ for all } v\in S.   
\end{align}

\begin{definition}[Monotone systems]\label{def:monotone-system}
    A distribution $\mu$ over $\{0,1\}^V$ is said to be a monotone system if, for every $v \in V$ and every pair of feasible and comparable pinnings $\sigma,\tau \in \{0,1\}^{V \setminus \{v\}}$ with $\sigma \preceq \tau$, the distribution $\mu^\sigma_v$ is stochastically dominated by $\mu^\tau_v$.
\end{definition}
In the definition, $\mu^\sigma_v$ and $\mu^\tau_v$ are marginal distributions of $\mu$ on $v$ conditional on $\sigma$ and $\tau$ respectively.
The marginal distribution $\mu^\sigma_v$ is said to be stochastically dominated by $\mu^\tau_v$ if $\mu^\sigma_v(1) \leq \mu^\tau_v(1)$. Equivalently, there exists a coupling $(X,Y)$ of $\mu^\sigma_v$ and $\mu^\tau_v$, where $X \sim \mu^\sigma_v$ and $Y \sim \mu^\tau_v$, such that $X \leq Y$ with probability 1.

Let $\Omega=\operatorname{supp}(\mu)$ and let $P$ be the transition matrix of the Glauber dynamics. Its mixing time from $X_0\in\Omega$ is
\begin{align}\label{eq:mix-init}
    T^{\textnormal{GD}}_{\textnormal{mix-}X_0}(\mu,\epsilon) = \min\{t \in \mathbb{N} \mid \DTV{P^t(X_0,\cdot)}{\mu}\leq \epsilon\},
\end{align}    
where $\DTV{P^t(X_0,\cdot)}{\mu} = \frac{1}{2} \sum_{\sigma \in \Omega} |P^t(X_0,\sigma) - \mu(\sigma)|$ denotes the total variation distance between two distributions $P^t(X_0,\cdot)$ and $\mu$.
In particular, for monotone systems, one often considers the case where $X_0 = \*1_V$ is the maximum all-one configuration.%\footnote{In this paper, we always assume all-one configuration is feasible, i.e. $\*1_V \in \Omega$.}.
Whenever a chain is started from $\*1_V$, we assume $\mu(\*1_V)>0$.

The \emph{censoring inequality}~\cite{peres2013can,FillK13} is a powerful tool to study the mixing time of the Glauber dynamics on monotone systems, which states that extra updates do not delay mixing. Hence, to analyze the mixing time of Glauber dynamics, one can focus on updates in a local region of interest and ignore the updates outside the region. One typical application of the censoring inequality is to compare Glauber dynamics with \emph{block dynamics}. 
Given the current configuration $X$, the block dynamics samples a random subset $S\subseteq V$, where $S$ is \emph{independent} of $X$, and then updates $X_S$ conditional on $X_{V \setminus S}$. 
With the censoring inequality, proving the mixing of Glauber dynamics can be reduced to proving the mixing of block dynamics and the mixing of Glauber dynamics in local regions $S$.  
This framework and its generalizations have achieved great success in many monotone systems~\cite{ding2010mixing,MS13,BlancaCV20,BlancaZ23,BlancaG25}.

Recent advances in spectral independence~\cite{ALO21} and entropic independence~\cite{AJKPV22} have led to new techniques for analyzing the mixing time of Glauber dynamics. Many new mixing time results have been established through the analysis of \emph{field dynamics}~\cite{CFYZ21,CFYZ22,ChenE22,CZ23,Chen0Y23,chen2024rapid}. Given the current configuration $X$, field dynamics samples a random subset $S\subseteq V$ that is \emph{adaptive} to $X$ (i.e., the distribution of $S$ depends on the current state), and then updates the configuration $X_S$ according to a \emph{tilted conditional distribution}. The key challenge in applying the censoring inequality to field dynamics arises from the \emph{correlation} between the adaptive sampling of $S$ and the current configuration $X$, combined with the specific update rule for $X_S$. Consequently, for some important monotone systems, while the optimal mixing time of field dynamics is understood, the best known mixing time bounds for Glauber dynamics remain sub-optimal.

In this paper, we develop a new comparison result between the Glauber dynamics and the field dynamics for monotone systems, which leads to improved mixing time results from entropic independence.
We first exhibit some concrete applications of the new comparison result in \Cref{sec:applications} and then give our general and technical results in \Cref{sec:general-and-technical-results}.

\subsection{Applications}\label{sec:applications}
For graph applications, we assume $n\geq2$; smaller instances mix in at most one update.
\paragraph{Random cluster model from ferromagnetic Ising model}
The ferromagnetic Ising model is a fundamental spin system in statistical physics. 
Let $G=(V,E)$ be an undirected graph with $n$ vertices and $m$ edges. Let $\beta \in (1,+\infty)^E$ be the \emph{edge activities} and $\lambda\in [0,1]^V$ be the \emph{external fields}. The ferromagnetic Ising model $\mu_{\beta,\lambda}^{\textnormal{Ising}}$ is defined as follows: 
\begin{align*}
\forall \sigma \in \{0,1\}^V,\quad    \mu_{\beta,\lambda}^{\textnormal{Ising}}(\sigma) \propto w_{\beta,\lambda}^{\textnormal{Ising}}(\sigma) := \prod_{e\in m(\sigma)}\beta_e\prod_{v\in V: \sigma_v = 1}\lambda_v, 
\end{align*}
where $w_{\beta,\lambda}^{\textnormal{Ising}}(\sigma)$ denotes the weight of $\sigma$ and $m(\sigma):=\{e = \{u,v\}\in E \mid \sigma_u = \sigma_v\}$ is the set of monochromatic edges of $\sigma$.

In the ferromagnetic Ising model, neighboring vertices are more likely to have the same value. If the Glauber dynamics starts from the all-ones configuration, then it may take exponential time to reach a configuration dominated by 0s. Alternatively, to draw random samples from the ferromagnetic Ising model, consider the following \emph{random cluster model}. 

Let $p\in [0,1]^E$ and $\lambda\in [0,1]^V$ be parameters.
For any subset of edges $S \subseteq E$, let $\kappa(V,S)$ denote the set of connected components of the subgraph $(V,S)$. The random cluster model defines a distribution $\mu_{p,\lambda}^{\textnormal{RC}}$ such that for any subset of edges $S \subseteq E$,
\begin{align}\label{eq:weight-rc}
    \mu_{p,\lambda}^{\textnormal{RC}}(S) \propto w_{p,\lambda}^{\textnormal{RC}}(S) := \prod_{e\in S}\frac{p_e}{1-p_e} \prod_{C\in \kappa(V,S)}\left(1+\prod_{v\in C}\lambda_v\right).
\end{align}
Here, $C \in \kappa(V,S)$ denotes a connected component of the subgraph $(V,S)$ and $C$ denotes the set of vertices in this connected component.
At boundary values of $p$, the distribution is defined by continuous extension.
By setting $p_e = 1 - \frac{1}{\beta_e}$ for all $e \in E$,
one can transfer a random sample from the random cluster model $ \mu_{p,\lambda}^{\textnormal{RC}}$ to a random sample from the ferromagnetic Ising model  $\mu_{\beta,\lambda}^{\textnormal{Ising}}$ using a simple coupling in~\cite{JerrumS93,feng2023swendsen}.

The problem of sampling from random cluster models was extensively studied. 
Assume all $p_e \in (0,1)$ are constants.
Jerrum and Sinclair \cite{JerrumS93} gave the first polynomial time algorithm to sample from the random cluster model with $\lambda = \bm{1}$. Guo and Jerrum \cite{GJ18} further proved the $\tilde{O}(m^3 n^4)$ mixing time of the Glauber dynamics for this case. Recently, many works~\cite{liu2019ising,feng2023swendsen,CZ23,chen2024spectral} considered the random cluster model with consistently biased external fields, i.e. for all $v \in V$, $\lambda_v \leq 1 - \delta_\lambda$ for some constant $\delta_\lambda > 0$. Feng, Guo and Wang \cite{feng2023swendsen} proved that the mixing time of Glauber dynamics is $\Delta^{O(\Delta^2)} \cdot n \log n$, where $\Delta$ is the maximum degree of the graph. The mixing time may not be polynomial when the maximum degree is large. Chen and Zhang \cite{CZ23} gave a \emph{sampling algorithm} (not Glauber dynamics) in time $ m \cdot (\log n)^{C}$, where the exponent $C$ is a constant depending on $(p_e)_{e \in E}$ and $\delta_\lambda$.
By a standard comparison argument, the technique in~\cite{CZ23} also leads to $\tilde{O}(m^2)$ mixing time of the Glauber dynamics.

It is well-known that the random cluster model is a monotone system.
Our first application is a near-linear mixing time of the Glauber dynamics on the random cluster model with consistently biased external fields.

\begin{theorem}\label{thm:RC-model}
    Let $\delta_p,\delta_\lambda\in (0,1)$ be constants.
    For any graph $G=(V,E)$ with $|V|=n$ and $|E|=m$, any $p\in[\delta_p,1)^E$, and any $\lambda\in [0,1-\delta_\lambda]^V$,
    the mixing time of Glauber dynamics on $\mu_{p,\lambda}^{\textnormal{RC}}$ starting from the full edge set $E$ satisfies, for every $0<\epsilon\leq 1/2$,
    \[T^{\textnormal{GD}}_{\textnormal{mix-}E}(\mu_{p,\lambda}^{\textnormal{RC}},\epsilon)\leq C(\delta_p,\delta_\lambda) \cdot m \cdot (\log n)^{O(1)} (\log \frac{1}{\epsilon})^2,\]
    where $C(\delta_p,\delta_\lambda)$ is a constant depending on $\delta_p$ and $\delta_\lambda$.
\end{theorem}

Compared to the results in \cite{feng2023swendsen,CZ23}, our mixing time is $\tilde{O}(m)$ for all (possibly unbounded degree) graphs. In terms of a sampling algorithm, our result yields an algorithm with constant error in time $C(\delta_p,\delta_\lambda)\cdot m \cdot (\log n)^{O(1)}$, where the exponent of $\log n$ is a universal constant. The algorithm in \cite{CZ23} runs in time $m \cdot (\log n)^{f(\delta_p,\delta_\lambda)}$. The constant $f(\delta_p,\delta_\lambda) \to \infty$ as $\delta_\lambda \to 0$. Hence, our bound improves the exponent of $\log n$ when external fields $\lambda_v$ are near~1.

\paragraph{Bipartite hardcore model}

Let $G=(V,E)$ be a graph with maximum degree $\Delta\geq 3$ and $\lambda>0$ be the external field.
%We say $\sigma \in \{0,1\}^V$ is an \emph{independent set} of $G$ if all $v \in V$ with $\sigma_v = 1$ form an independent set.
%Let $\Vert \sigma \Vert_1$ denote the number of 1s in $\sigma$.
The hardcore model $\mu^{\textnormal{HC}}_{G,\lambda}$ defines a distribution over all independent sets in the graph $G$ such that
\begin{align*}
    \forall \text{ independent set } S \subseteq V, \quad \mu^{\textnormal{HC}}_{G,\lambda}(S)\propto \lambda^{|S|}.
\end{align*}
%where $\textnormal{Ind}_G$ is the set of all independent set of $G$.
The uniqueness threshold of the hardcore model is given by
\begin{align*}
\lambda_c(\Delta)\defeq \frac{(\Delta-1)^{\Delta-1}}{(\Delta-2)^\Delta}\approx \frac{e}{\Delta}.
\end{align*}
The famous computational phase transition result shows that if $\lambda \leq \lambda_c(\Delta)$, then the hardcore model can be sampled from in polynomial time~\cite{weitz2006counting,chen2024rapid}; if $\lambda > \lambda_c(\Delta)$, then unless $\textbf{NP} = \textbf{RP}$, a polynomial time sampler does not exist~\cite{Sly10,sly2012computational}.

If the graph $G=(V_L \uplus V_R,E)$ is bipartite, then the hardness results~\cite{Sly10,sly2012computational} do not apply directly.
The bipartite hardcore model is closely related to the \#BIS problem (counting the number of independent sets in a bipartite graph), which is one of the most important open problems in sampling and approximate counting~\cite{dyer2004relative}. 
The sampling problem of the bipartite hardcore model was extensively studied both in the high-temperature (small $\lambda$) regime \cite{liu2015fptas,Chen0Y23} and the low-temperature (large $\lambda$) regime~\cite{HPR20,cannon2020counting,liao2019counting,chen2022sampling,jenssen2023approximate}.

For the bipartite hardcore model, an important high-temperature regime is when the uniqueness condition is satisfied with respect to the degree on one side.
Let $\Delta_L$ denote the maximum degree of vertices on the left side $V_L$. Suppose $\lambda \leq (1-\delta)\lambda_c(\Delta_L)$ for some constant $\delta > 0$. Chen, Liu and Yin~\cite{Chen0Y23} proved that the mixing time of the Glauber dynamics for $\mu^{\textnormal{HC}}_{G,\lambda}$ is $C(\Delta_L, \lambda, \delta) \cdot n^3 (\log n)^{O(1/\delta)}$, where $C(\Delta_L, \lambda, \delta)$ depends only on $\Delta_L, \lambda$ and $\delta$, and $n = |V_L| + |V_R|$ is the number of vertices. 
Recently, Chen and Feng~\cite{chenfeng2024rapid} considered the \emph{balanced} bipartite hardcore model, which further requires the maximum degree $\Delta_R$ on the right side $V_R$ to satisfy $\Delta_R = O(\Delta_L)$. 
For Glauber dynamics starting from the independent set consisting of all vertices on one side, \cite{chenfeng2024rapid} proved the $C(\lambda, \delta) \cdot n (\log n)^{O(1/\delta)}$ mixing time.
Under the encoding in~\eqref{eq:bhc-mapping}, the bipartite hardcore model is a monotone system.
We give the following mixing time result.

\begin{theorem}\label{thm:bhc-mixing-time}
Let $\delta > 0$ be a constant. Let $\Delta_L \geq 3$ and $\lambda > 0$ be two parameters satisfying $\lambda < (1-\delta)\lambda_c(\Delta_L)$. For any bipartite graph $G=(V_L \uplus V_R,E)$ with $n$ vertices and left side degree at most $\Delta_L$, the mixing time of Glauber dynamics for $\mu^{\textnormal{HC}}_{G,\lambda}$ starting from the independent set consisting of all vertices in $V_R$ satisfies, for every $0<\epsilon\leq1/2$,
\[T^{\textnormal{GD}}_{\textnormal{mix-}V_R}(\mu_{G,\lambda}^{\textnormal{HC}},\epsilon)\leq n^2 \cdot \left(\frac{C}{\lambda}\right)^{O(1/\delta)}  (\Delta_L\log n)^{O(1)} \cdot \log^3(1/\epsilon),\]
where $C>4$ and all constants implicit in $O(\cdot)$ are absolute.
\end{theorem}

Compared to the result in \cite{Chen0Y23}, our result improves the mixing time from $\tilde{O}(n^3)$ to $\tilde{O}(n^2)$.
Compared to the result in \cite{chenfeng2024rapid}, our result works for general graphs. However, \cite{chenfeng2024rapid} gave a faster $\tilde{O}(n)$ mixing time in the balanced case.

\begin{remark}[Mixing time for marginal distribution and sampling algorithm]
Chen, Liu and Yin~\cite{Chen0Y23} also studied the problem of sampling from $\mu^{\textnormal{HC}}_{L,\lambda}$, which is the \emph{marginal distribution} of $\mu^{\textnormal{HC}}_{G,\lambda}$ projected onto the left side $V_L$. 
Formally, $S \cap V_L \sim \mu^{\textnormal{HC}}_{L,\lambda}$ if $S \sim \mu^{\textnormal{HC}}_{G,\lambda}$.
They gave a \emph{sampling algorithm} (not Glauber dynamics) in near-linear time $C(\Delta_L, \lambda, \delta) \cdot n \cdot (\log n)^{O(1/\delta)}$. 
This marginal distribution is also a monotone system.
Our technique also gives the near-linear mixing time of the Glauber dynamics on $\mu^{\textnormal{HC}}_{L,\lambda}$ starting from the empty set, for $0<\epsilon\leq1/2$:
\[T_{\textnormal{mix-}\emptyset}(\mu_{L,\lambda}^{\textnormal{HC}},\epsilon)\leq C(\Delta_L, \lambda, \delta) \cdot n \cdot (\log n)^{O(1)} \log^2\frac{1}{\epsilon} = \tilde{O}_{\Delta_L,\lambda,\delta}(n).\]
See \Cref{thm:bhc-mixing-time-marginal} for more details. 
Our mixing result for the marginal distribution immediately gives a sampling algorithm in near-linear time $C(\Delta_L, \lambda, \delta) \cdot n (\log n)^{O(1)}$. Compared to the sampling algorithm in \cite{Chen0Y23}, our algorithm improves the degree of the polylog$(n)$ factor from $O(1/\delta)$ to a universal constant $O(1)$.
\end{remark}

\subsection{General and technical results}\label{sec:general-and-technical-results}

\paragraph{Monotone spin systems with entropic independence}
Our general result is a comparison theorem, which works for monotone systems satisfying entropic independence. Let $\mu$ be a distribution with support $\Omega \subseteq \{0,1\}^V$. For any distribution $\nu$ over $\Omega$, define their KL-divergence as $\mathrm{D}_{\textnormal{KL}}(\nu\parallel \mu)= \sum_{\sigma\in \Omega}\nu(\sigma)\log \frac{\nu(\sigma)}{\mu(\sigma)}$. The following notion was introduced in~\cite{AJKPV22}.
\begin{definition}[Entropic independence \cite{AJKPV22}]
    Let $\alpha>0$. A distribution $\mu$ with support $\Omega \subseteq \{0,1\}^V$ is said to be $\alpha$-entropically independent if for any distribution $\nu$ over $\Omega$,
    \begin{align}\label{eq:entropic-independence}
    \sum_{i\in V}\mathrm{D}_{\textnormal{KL}}(\nu_i\parallel \mu_i)\leq \alpha \mathrm{D}_{\textnormal{KL}}(\nu\parallel \mu),
    \end{align}
    where $\mu_i$ and $\nu_i$ are the marginal distributions of $\mu$ and $\nu$ on variable $i \in V$ respectively. %Moreover, we say $\mu$ is $\alpha$-entropically independent under all pinnings if for any $\Lambda \subseteq V$ with $|\Lambda| \leq |V| - 2$ and any feasible pinning $\sigma \in \{0,1\}^{\Lambda}$, $\mu^{\sigma}$ is $\alpha$-entropically independent.
\end{definition}

Our general result compares the mixing time of Glauber dynamics for $\mu$ to the mixing time of the Glauber dynamics for a tilted distribution.

\begin{definition}[Tilted distributions]
    Let $\mu$ be a distribution with support $\Omega \subseteq \{0,1\}^V$ and $\theta > 0$ be a real number. The tilted distribution $(\theta * \mu)$ is defined as 
    \begin{align}
    \forall \sigma \in \Omega, \quad (\theta * \mu)(\sigma) \propto \mu(\sigma)\theta^{\normone{\sigma}},
    \end{align}
    where $\normone{\sigma}$ is the number of 1's in $\sigma$.
\end{definition}

The tilted distribution $(\theta * \mu)$ puts an external field on each vertex. For example, if $\mu$ is a hardcore model on graph $G=(V,E)$ with parameter $\lambda$, then $(\theta * \mu)$ is a hardcore model on $G$ with parameter $\theta \lambda$.
For many natural distributions, it is easy to analyze the mixing time of the Glauber dynamics on $(\theta * \mu)$ for a proper $\theta \in (0,1)$. 
We next define the mixing time of the Glauber dynamics for tilted distributions with pinnings. For any $\epsilon \in (0,1)$, define
\begin{align}\label{def:mixtl}
   \mixtl(\mu,\theta,\epsilon) \defeq \max\left\{\mixgd((\theta * \mu)^{\*1_{\Lambda}},\epsilon) \mid \Lambda \subseteq V \text{ and } \mu_\Lambda(\*1_\Lambda) > 0\right\},
\end{align}
where $(\theta * \mu)^{\*1_{\Lambda}}$ is the distribution of $\theta * \mu$ conditioned on all variables in $\Lambda$ being fixed to 1. For any distribution $\nu$, we use $\mixgd(\nu,\epsilon)$ to denote the mixing time of the Glauber dynamics $P$ on $\nu$ starting from an \emph{arbitrary} feasible configuration:
\begin{align}
    \mixgd(\nu,\epsilon) = \max_{X \in \{0,1\}^V: \nu(X) > 0} \min\{t \in \mathbb{N} \mid \DTV{P^t(X,\cdot)}{\nu}\leq \epsilon\}.\label{eq:mixgd}
\end{align}
We remark that the mixing time in~\eqref{eq:mixgd} is stronger than the mixing time in~\eqref{eq:mix-init} and we use the stronger mixing time to define the mixing time of tilted distributions in~\eqref{def:mixtl}. 

Our general result is a comparison theorem between $\mixtl(\mu,\theta,\epsilon)$ and $\mixmax(\mu,\epsilon)$ when $\mu$ is a monotone system satisfying entropic independence.

\begin{theorem}\label{thm:main-comparison}
Let $\mu$ be a monotone system with support $\Omega \subseteq \{0,1\}^V$. Let $\theta \in (0,1)$. Let $\alpha:[0,-\log \theta] \to \mathbb{R}_{> 0}$ be an integrable function.
Suppose for any $t \in [0,-\log \theta]$, for any $\Lambda \subseteq V$ such that $\mu_{\Lambda}(\*1_\Lambda) > 0$, the distribution $(e^{-t} * \mu)^{\*1_\Lambda}$ is $\alpha(t)$-entropically independent.
If $\mu$ is a point mass, the mixing time is zero. Otherwise, for $0<\epsilon\leq1/2$, the mixing time of the Glauber dynamics on $\mu$ starting from the maximum all-1 state $\*1_V$ satisfies
\begin{align*}
    \mixmax(\mu,\epsilon) \leq T \cdot \mixtl\tp{\mu,\theta,\frac{\epsilon}{2T}},
\end{align*}
where
\begin{align*}
    T = O\tp{\exp \tp{\int_{0}^{-\log\theta}4\alpha(t)dt } \cdot \tp{\log \log \frac{1}{\mu_{\min}} + \log \frac{1}{\epsilon}}} \text{ and } \mu_{\min} = \min_{\sigma \in \Omega} \mu(\sigma).
\end{align*}
\end{theorem}

In applications, by appropriately choosing the parameter $\theta$, we can guarantee that $(\theta * \mu)$ is in a non-critical region so that  $\mixtl(\mu,\theta,\epsilon)$ is easy to bound. Typically, $\mixtl(\mu,\theta,\epsilon) = n \cdot \mathrm{polylog}\frac{n}{\epsilon}$ and $\mu_{\min} = e^{-\Theta(n)}$. Hence, if we can verify the entropic independence condition and show that $\exp \tp{\int_{0}^{-\log\theta}4\alpha(t)dt } \leq \mathrm{polylog}(n)$, then \Cref{thm:main-comparison} gives a near-optimal $n \cdot \mathrm{polylog}\frac{n}{\epsilon}$ mixing time for the Glauber dynamics starting from the maximum all-1 state $\*1_V$.

Using field dynamics, previous works~\cite{CFYZ21,CFYZ22,ChenE22} compare the \emph{spectral gap} and \emph{modified log-Sobolev (mLS) constants} between the Glauber dynamics in the critical and non-critical regimes.
Our result \Cref{thm:main-comparison} directly compares the mixing time of Glauber dynamics on monotone systems. 
The spectral gap may lead to a sub-optimal mixing time bound and the mLS constants are sometimes hard to analyze even in the non-critical regime.
Compared to previous works, our comparison result is easier to apply and can yield the near-optimal mixing time bound for many distributions. However, our result only works for monotone systems.

In \cite{chenfeng2024rapid}, a mixing time comparison result is proved for Glauber dynamics in the critical and non-critical regimes on monotone systems. 
However, their non-critical regime is defined in terms of certain conditional distributions induced from $\mu$ while our non-critical regime is defined as the tilted distributions.
Due to this difference, our result works for many monotone systems for which it is unclear whether their result applies.
For example, our \Cref{thm:bhc-mixing-time} works for general bipartite graphs while their result only works for special balanced bipartite graphs.
From a technical perspective, the result in \cite{chenfeng2024rapid} is an application of the standard censoring inequality~\cite{peres2013can} while our result is obtained through a new comparison argument between Glauber dynamics and field dynamics.

\paragraph{Comparison through field dynamics}

The general comparison result in \Cref{thm:main-comparison} is proved through a Markov chain called \emph{field dynamics}. %This comparison is the main technical contribution of this paper.
Field dynamics was first introduced in~\cite{CFYZ21}. Later, it was used to prove the mixing time for many distributions whose parameters are in a critical regime~\cite{AJKPV22,CFYZ22,ChenE22,CZ23,Chen0Y23,chen2024rapid}.

\begin{definition}[Field dynamics \cite{CFYZ21}]\label{def:field-dynamics}
    Let $\theta \in (0,1)$ be a parameter. Let $\mu$ be a distribution with support $\Omega \subseteq \{0,1\}^V$. The field dynamics starts from an arbitrary feasible configuration $X \in \Omega$ and in every step, it updates the configuration $X$ as follows. 
    \begin{itemize}
        \item Sample a random subset $S \subseteq V$ by selecting each variable $v \in V$ independently with probability $p_v$ such that 
        \begin{align}\label{eq:field-dynamics-p}
             p_v = \begin{cases}
                \theta & \text{ if } X_v = 1;\\
                1 & \text{ if } X_v = 0.
            \end{cases}
        \end{align}
        \item Resample $X \sim (\theta * \mu)^{\*1_{V \setminus S}}$, where $(\theta * \mu)^{\*1_{V \setminus S}}$ is the distribution of $\sigma \sim (\theta * \mu)$ conditional on $\sigma_{V \setminus S}  = \*1_{V \setminus S}$.
    \end{itemize}
    \end{definition}
    
    Let $P_{\theta,\mu}$ denote the transition matrix of the field dynamics on $\mu$ with parameter $\theta$. Define the mixing time of the field dynamics starting from an arbitrary configuration\footnote{It suffices to define the mixing time of field dynamics from the maximum all-1 starting state $\*1_V$. The result in \Cref{thm:FD-comparison} remains valid under this weaker definition, with the zero-time case handled separately.} $X \in \Omega$ as 
    \begin{align*}
        \mixfd(\mu,\theta,\epsilon) = \max_{X \in \Omega} \min\left\{t \in \mathbb{N} \mid \DTV{P_{\theta,\mu}^t(X,\cdot)}{\mu}\leq \epsilon\right\}.
    \end{align*}

\begin{theorem}\label{thm:FD-comparison}
    Let $\mu$ be a monotone system with support $\Omega \subseteq \{0,1\}^V$ and $|V| = n$. Let $\theta \in (0,1)$. 
    If $\mu$ is a point mass, the mixing time is zero. Otherwise, for $0<\epsilon<1$, the mixing time of the Glauber dynamics for $\mu$ starting from the maximum all-1 state $\*1_V$ satisfies
    \begin{align*}
        \mixmax(\mu,\epsilon) \leq \mixfd\tp{\mu,\theta,\frac{\epsilon}{2}} \cdot \mixtl\tp{\mu,\theta,\delta}, \quad \text{where } \delta = \frac{\epsilon}{2 \mixfd(\mu,\theta,\frac{\epsilon}{2})}.
    \end{align*}
    %where $\delta = \frac{\epsilon}{2 \mixfd(\mu,\frac{\epsilon}{2})}$. %and $\mixtl(\cdot,\cdot,\cdot)$ defined in \eqref{def:mixtl} is the mixing time of the Glauber dynamics for tilted distribution.
    \end{theorem}

    A proof outline of \Cref{thm:FD-comparison} is given in \Cref{sec:proof-outline}.

    The general result in \Cref{thm:main-comparison} is a corollary of the above technical result.
    In \Cref{thm:main-comparison}, the entropic independence condition implies that the mixing time of the field dynamics satisfies $\mixfd\tp{\mu,\theta,\frac{\epsilon}{2}} \leq T$. 
    Then \Cref{thm:FD-comparison} immediately implies \Cref{thm:main-comparison}. 
    The formal proof of \Cref{thm:main-comparison} assuming \Cref{thm:FD-comparison} is deferred to \Cref{sec:proof-GR}.

\section{Preliminaries: A general comparison technique}\label{sec:preliminaries}

In the preliminaries,
we give and slightly reformulate the comparison technique developed by Fill and Kahn \cite{FillK13}, which will serve as a basic tool in our proof.
Let $\mu$ be a distribution with support $\Omega$. 
Let $\preceq$ be a partial order on $\Omega$. 
Define the following class of \emph{increasing} functions. Some literature also calls them \emph{non-decreasing} functions.
\begin{definition}[Increasing function]\label{def:increasing-function}
    A non-negative function $f:\Omega \to \mathbb{R}_{\geq 0}$ is said to be an increasing function if for any $X,Y \in \Omega$ with $X\preceq Y$, it holds that $f(X)\leq f(Y)$.
\end{definition}

The following stochastic dominance relation between distributions is well-known.
\begin{definition}[Stochastic dominance] \label{def:stochastic-dominance}
    Let $\pi$ and $\pi'$ be two distributions over $\Omega$. The following two definitions of stochastic dominance (denoted by $\pi \preceq_{\text{sd}} \pi'$) are equivalent.
    \begin{itemize}
        \item There exists a coupling $(X,Y)$ of $\pi$ and $\pi'$, where $X \sim \pi$ and $Y \sim \pi'$, such that $X \preceq Y$ with probability 1.
        \item For any increasing function $f:\Omega \to \mathbb{R}_{\geq 0}$, it holds that $\mathbb{E}_{X \sim \pi}[f(X)] \leq \mathbb{E}_{Y \sim \pi'}[f(Y)]$.
    \end{itemize}
\end{definition}
The equivalence of the two definitions in \Cref{def:stochastic-dominance} is proved in \cite[Theorem 22.6]{levin2017markov}.

Fill and Kahn \cite{FillK13} also introduced a relation between different Markov chains.
We first recall the definition of \emph{reversibility}.
A Markov chain with transition matrix $P$ is \emph{reversible} with respect to $\mu$ if for any $\sigma,\tau \in \Omega$, $\mu(\sigma)P(\sigma,\tau) = \mu(\tau)P(\tau,\sigma)$. 
Then $\mu$ is a stationary distribution of $P$ such that $\mu P = \mu$, where $\mu$ is viewed as a row vector. We also view $P$ as a linear operator acting on column vectors.
\begin{definition}[Stochastically monotone chain]\label{def:stochastically-monotone}
A Markov chain $P$ is said to be stochastically monotone if for any increasing function $f:\Omega \to \mathbb{R}_{\geq 0}$, $Pf$ is also an increasing function.
\end{definition}

In the above definition, a function $f:\Omega \to \mathbb{R}_{\geq 0}$ is viewed as a column vector. The definition says that $P$, as a linear operator, preserves the increasing property of the function.
The following proposition gives two additional equivalent definitions of the stochastically monotone chain.
\begin{proposition}[\text{\cite[Prop. 22.7]{levin2017markov}}]\label{prop:stochastically-monotone-equivalent}
    The following statements are equivalent.
\begin{itemize}
        \item $P$ is stochastically monotone.
        \item For any two distributions $\nu \sd \nu'$, $\nu P \sd \nu' P$.
        \item For any pair of comparable $\sigma,\tau \in \Omega$ with $\sigma \preceq \tau$, $P(\sigma,\cdot) \sd P(\tau,\cdot)$.
\end{itemize}
\end{proposition}

Define the following family of Markov chains for a distribution $\mu$:
\begin{align}\label{eq:mc-set}
    \+{MC}_\mu = \{ \text{Markov chains that are reversible w.r.t. $\mu$ and stochastically monotone}\}.
\end{align}
Define the following comparison inequality relation $\preceq_{\text{mc}}$ over the Markov chain space $\+{MC}_\mu$.
\begin{definition}[Comparison inequality relation~\cite{FillK13}]\label{def:comparison-inequality-relation}
   For any $P,Q \in \+{MC}_\mu$, we say that $P\mc Q$ iff for any distribution $\nu$ such that the function $\frac{\nu}{\mu}:\Omega \to \mathbb{R}_{\geq 0}$ is increasing, $\nu P \preceq_{\text{sd}} \nu Q$.
\end{definition}
The notation $\frac{\nu}{\mu}$ denotes a function from $\Omega$ to $\mathbb{R}_{\geq 0}$ defined by $\frac{\nu}{\mu}(\sigma) = \frac{\nu(\sigma)}{\mu(\sigma)}$ for any $\sigma \in \Omega$. %The value of $\frac{\nu}{\mu}$ is finite because $\mu(\sigma) > 0$ for all $\sigma \in \Omega$.
On $\+{MC}_\mu$, \Cref{def:comparison-inequality-relation} is equivalent to the Fill--Kahn relation \cite{FillK13}, expressed using increasing rather than decreasing test functions:
$P\mc Q$ if and only if $\langle Pg,f\rangle_\mu\leq\langle Qg,f\rangle_\mu$ for all nonnegative increasing functions $f,g$.
Here $\langle h,\ell\rangle_\mu=\sum_{\sigma\in\Omega}h(\sigma)\ell(\sigma)\mu(\sigma)$.
A proof is given in Appendix~\ref{sec:proofs-for-comparison-inequality-relation}.

It is straightforward to verify that the family $\+{MC}_{\mu}$ is closed under convex combinations, i.e., if $P,Q \in \+{MC}_{\mu}$, then $\lambda P + (1-\lambda)Q \in \+{MC}_{\mu}$ for any $0 \leq \lambda \leq 1$.
The matrix product $PQ$ represents the transition that first applies $P$ and then $Q$.
The following properties of the comparison inequality are proved in \cite{FillK13}.
\begin{lemma}[\text{\cite[Prop. 2.3, 2.4]{FillK13}}]\label{lem:comparison-inequality-properties}
The following properties hold for Markov chains in $\+{MC}_\mu$.
\begin{itemize}
 \item If $P_0 \mc Q_0$ and $P_1 \mc Q_1$, then for any $0 \leq \lambda \leq 1$, $\lambda P_0 + (1-\lambda)P_1 \mc \lambda Q_0 + (1-\lambda)Q_1$.
 \item If $P_i \mc Q_i$ for any $1\leq i \leq \ell$, then $\nu P_1P_2\ldots P_\ell \sd \nu Q_1Q_2\ldots Q_\ell$ for any distribution $\nu$ such that $\frac{\nu}{\mu}$ is increasing.
\end{itemize}
\end{lemma}

The following comparison results are proved in \cite{peres2013can,FillK13}.
\begin{lemma}[\text{\cite[Thm. 2.5]{peres2013can} and \cite[Cor. 3.3]{FillK13}}]\label{lem:gen-compare}
For any two distributions $\pi$ and $\pi'$, if $\frac{\pi}{\mu}$ is an increasing function and $\pi \sd \pi'$, then $\DTV{\pi}{\mu} \leq \DTV{\pi'}{\mu}$.     

As a consequence, for two Markov chains $P,Q \in \+{MC}_{\mu}$, if $P \mc Q$, then for any distribution $\pi_0$ such that $\frac{\pi_0}{\mu}$ is increasing, it holds that $\DTV{\pi_0 P}{\mu} \leq \DTV{\pi_0 Q}{\mu}$.
\end{lemma}

Proofs of \Cref{lem:comparison-inequality-properties,lem:gen-compare} are included in Appendix~\ref{sec:proofs-for-comparison-inequality-relation} for completeness.

\section{Proof outline of the technical result}\label{sec:proof-outline}

In this section, we give an outline of the proof of the comparison result in \Cref{thm:FD-comparison}.

\subsection{A simulation algorithm for the field dynamics}

%Our proof relies on the following Markov chain, called the \emph{field dynamics}.

To prove \Cref{thm:FD-comparison}, we compare the Glauber dynamics with an algorithm that simulates the field dynamics. In the simulation algorithm, each resampling step of the field dynamics is implemented by running the Glauber dynamics for the tilted distribution.
To use the comparison technique in \Cref{sec:preliminaries}, we need to ensure that the simulation algorithm forms a Markov chain.
To achieve this, we extend the configuration space from $\{0,1\}^V$ to $\{0,1,\star\}^V$.
%The simulation algorithm is defined as a stochastic process over the extended configuration space $\{0,1,\star\}^V$.
Define the following two basic operations.
\begin{itemize}
    \item \textbf{Lifting Operation}: $\textsf{lift}:\{0,1\}^V \to \{0,1,\star\}^V$ is a \emph{random} function that maps any $X \in \{0,1\}^V$ to a random configuration $Y = \textsf{lift}(X) \in \{0,1,\star\}^V$ using the following procedure. Independently for each $v \in V$, if $X_v = 0$, then $Y_v = 0$; if $X_v = 1$, set $Y_v = \star$ with probability $1-\theta$ and $Y_v = 1$ with probability $\theta$.
    \item \textbf{Contracting Operation}: $\textsf{contr}:\{0,1,\star\}^V \to \{0,1\}^V$ is a \emph{deterministic} function that maps any $Y \in \{0,1,\star\}^V$ to a configuration $X = \textsf{contr}(Y) \in \{0,1\}^V$ using the following procedure. For each $v \in V$, if $Y_v = 0$, then $X_v = 0$; if $Y_v \in \{1,\star\}$, then $X_v = 1$.
\end{itemize}
For $T_1,T_2\in\mathbb Z_{\geq1}$ and $\theta\in(0,1)$, define the following simulation algorithm for the field dynamics.

 \begin{tcolorbox}[colback=lightgray!20, colframe=lightgray!18, coltitle=black, title={\underline{Simulation Algorithm for Field Dynamics} $\+A_\mu(\theta,T_1,T_2)$}]
   \begin{enumerate}
        \item  Initialize with $X = \textsf{lift}(\*1_V)$, which is the lifted maximum all-1 state.
        \item \label{item:outer-loop} For each step $t$, from 1 to $T_1$, do the following:
       \begin{enumerate}
            \item \label{item:contracting} Contract $X \gets \textsf{contr}(X)$ to a configuration in $\{0,1\}^V$.
           \item \label{item:lifting} Lift $X \gets \textsf{lift}(X)$ to a configuration in $\{0,1,\star\}^V$. Let $S = \{v \in V \mid X_v \neq \star\}$.
            \item \label{item:inner-loop} Simulate the Glauber dynamics on the distribution $(\theta *\mu)^{\*1_{V \setminus S}}$ for $T_2$ steps to update $X$. Specifically, repeat the following steps $T_2$ times:
            \begin{enumerate}
                \item \label{item:random-variable} Select a random variable $v \in V$ uniformly at random;
                \item \label{item:resample} If $v \in S$, resample $X_v \sim (\theta *\mu)^{\sigma_{V \setminus \{v\}}}_v$; if $v \notin S$, then keep $X_v = \star$, where $\sigma = \textsf{contr}(X)\in \{0,1\}^V$ is the contraction of $X$. 
            \end{enumerate}
        \end{enumerate}
        \item Return the random sample $\textsf{contr}(X)$ as the output.
    \end{enumerate}
\end{tcolorbox}

In the outer loop, as indicated in Line~\ref{item:outer-loop}, the algorithm simulates the field dynamics starting from the all-1 configuration $\*1_V$ for $T_1$ steps.
Before Line~\ref{item:contracting}, $X$ is always a configuration in $\{0,1,\star\}^V$.
Then Line~\ref{item:contracting} contracts $X$ to a configuration in $\{0,1\}^V$.
Line \Cref{item:lifting} and Line \Cref{item:inner-loop} together approximately simulate one transition step of the field dynamics.
By \Cref{def:field-dynamics}, each transition step of field dynamics contains two sub-steps: (1) constructing the random set $S \subseteq V$; (2) resampling $X \sim (\theta * \mu)^{\*1_{V \setminus S}}$.

\begin{itemize}
    \item \textbf{Construction step}. In Line~\ref{item:lifting}, the algorithm lifts $X \in \{0,1\}^V$ to a configuration in $\{0,1,\star\}^V$.
    This lifting operation is equivalent to the construction of the random set $S \subseteq V$ in the field dynamics, where $S$ is constructed using parameters in \eqref{eq:field-dynamics-p}.
    After the lifting operation, a vertex $v \in S$ if and only if $X_v \neq \star$. 
    \item \textbf{Resampling step}. In the inner loop, as indicated in Line~\ref{item:inner-loop}, the algorithm simulates the Glauber dynamics on $\nu = (\theta * \mu)^{\*1_{V \setminus S}}$ for $T_2$ steps, which aims to approximate the resampling step of the field dynamics.
    Specifically, in each step, the Glauber dynamics selects a vertex $v \in V$ uniformly at random.
    If $v \in S$, or equivalently $X_v \neq \star$, the Glauber dynamics resamples $X_v \sim \nu_v^{X_{S \setminus \{v\}}}$, which is the marginal distribution on $v$ projected from $\nu = (\theta * \mu)^{\*1_{V \setminus S}}$ conditional on the current state of the other variables in $S \setminus \{v\}$. 
    If $v \notin S$, or equivalently $X_v = \star$, it means the value of $v$ is fixed to be 1 in the distribution $\nu$, and the algorithm keeps $X_v = \star$.
    After $T_2$ steps of the inner loop, the algorithm assumes that the distribution of $\textsf{contr}(X)$ is close to the distribution $\nu=(\theta * \mu)^{\*1_{V \setminus S}}$.
    After simulating one step of the field dynamics, the algorithm contracts $X \in \{0,1,\star\}^V$ back to a configuration in $\{0,1\}^V$ in Line~\ref{item:contracting}.
\end{itemize}

Since we consider the lifted space $\{0,1,\star\}^V$,
the simulation algorithm $\+A_{\mu}(\theta,T_1,T_2)$ itself is a Markov chain $(X^{(t)}_{\textnormal{alg}})_{t=0}^{T_1T_2}$. Specifically, $X^{(0)}_{\textnormal{alg}} = \textsf{lift}(\*1_V)$.
For $1 \leq t \leq T_1T_2$, $X^{(t)}_{\textnormal{alg}}$ is the random configuration $X$ after the $t$-th execution of Line~\ref{item:resample}. %i.e., the line for one resampling step of the Glauber dynamics for the distribution $(\theta * \mu)^{\*1_{V \setminus S}}$.
This Markov chain is \emph{time-inhomogeneous} because the transition $X^{(t)}_{\textnormal{alg}} \to X^{(t+1)}_{\textnormal{alg}}$ follows different rules depending on whether $(t \mod T_2) = 0$.
See \Cref{fig:lifted-simulation-algorithm} for a simple example.

\begin{figure}[ht]
    \centering % This command centers the content within the figure environment
    \begin{tikzpicture}
        % Define the background layer
        \pgfdeclarelayer{background}
        \pgfsetlayers{background,main}
        
        % Draw the main line segment from 0 to 13
        \draw[thick] (0, 0) -- (13, 0);
        
        % Draw the ticks and labels
        \foreach \x in {0, 1, ..., 12} {
            \draw (\x, 0.1) -- (\x, -0.1);
            \node[below] at (\x, -0.1) {\x};
        }
        
        % Draw transparent rounded boxes around segments
        \begin{pgfonlayer}{background}
            \filldraw[fill=red!20, draw=red, thick, rounded corners, opacity=0.3] (0.8, -0.5) rectangle (4.2, 0.5);
            \filldraw[fill=green!20, draw=green, thick, rounded corners, opacity=0.3] (4.8, -0.5) rectangle (8.2, 0.5);
            \filldraw[fill=blue!20, draw=blue, thick, rounded corners, opacity=0.3] (8.8, -0.5) rectangle (12.2, 0.5);
        \end{pgfonlayer}
        
        % Add labels on top of the boxes
        \node[above, align=center] at (2.5, 0.5) {Glauber dynamics \\ for $(\theta * \mu)^{\*1_{V \setminus {S_1}}}$};
        \node[above, align=center] at (6.5, 0.5) {Glauber dynamics \\ for $(\theta * \mu)^{\*1_{V \setminus {S_2}}}$};
        \node[above, align=center] at (10.5, 0.5) {Glauber dynamics \\ for $(\theta * \mu)^{\*1_{V \setminus {S_3}}}$};
        
        % Define a vertical shift
        \node[above, align=center] at (-1.3, -0.5) {start from \\ $X=\textsf{lift}(\*1_V)$};
        
        % Add "sample B_i" text and curved arrows
        \def\hshift{-0.2}
        \foreach \i/\x/\textx in {1/1/2.5} {
            \node[below, align=center] at (\textx, -0.85) { Contract and lift $X$, \\construct  $S_{\i}$ from $X$};
            \draw[thick, ->, >=Stealth] (\x+\hshift, -0.95) arc[start angle=240, end angle=120, radius=0.5];
        }
        \foreach \i/\x/\textx in {2/5/6.5, 3/9/10.5} {
            \node[below, align=center] at (\textx, -0.85) {Contract and lift $X$, \\ construct  $S_{\i}$ from $X$};
            \draw[thick, ->, >=Stealth] (\x+\hshift, -0.95) arc[start angle=240, end angle=120, radius=0.5];
        }
    \end{tikzpicture}
    \caption{This figure illustrates the algorithm \(\+A_{\mu}(\theta,T_1,T_2)\) with \(T_1 = 3\) and \(T_2 = 4\). Initially, at time 0, let $X = X^{(0)}_{\text{alg}} = \textsf{lift}(\*1_V)$. Before time 1, the algorithm first contracts $X^{(0)}_{\text{alg}}$ and then lifts it. The algorithm next constructs $S_1 = \{i\in V \mid X_i \neq \star\}$. Then, the algorithm runs Glauber dynamics for the conditional distribution $(\theta * \mu)^{\*1_{V \setminus {S_1}}}$ for $T_2 = 4$ steps to obtain $X^{(t)}_{\text{alg}}$ for $1\leq t \leq 4$. The same procedure is repeated twice to obtain $X= X^{(12)}_{\text{alg}}$.}
    \label{fig:lifted-simulation-algorithm}
    \end{figure}
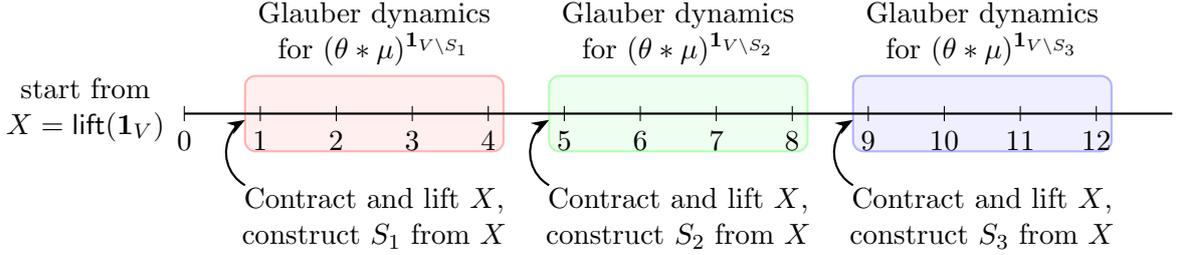

Using a coupling argument and a simple union bound, we can verify that the algorithm $\+A_\mu(\theta,T_1,T_2)$ produces an accurate sample if $T_1 \geq T_{\text{mix}}^\text{FD}(\mu,\theta,\frac{\epsilon}{2})$, which is the mixing time of the field dynamics with an error of $\frac{\epsilon}{2}$, and $T_2 \geq \mixtl(\mu,\theta,\frac{\epsilon}{2T_1})$, which is the mixing time of the Glauber dynamics for the tilted distribution with an error of $\frac{\epsilon}{2T_1}$.
Formally, we have the following lemma. The proof is deferred to \Cref{sec:simu-error}.

\begin{lemma}\label{lem:simu-error}
For any distribution $\mu$ over $\{0,1\}^V$ and any $\theta,\epsilon\in(0,1)$, if $T_1 \geq \mixfd(\mu,\theta,\frac{\epsilon}{2})$ and $T_2 \geq \mixtl(\mu,\theta,\frac{\epsilon}{2T_1})$, then the output $X = \textsf{contr}(X^{(T_1T_2)}_{\textnormal{alg}})$ of the algorithm $\+A_\mu(\theta,T_1,T_2)$ satisfies
\[\DTV{X}{\mu} \leq \epsilon.\]
\end{lemma}

%Our main technical contribution is the following comparison lemma. 
Let \( (X^{(t)}_{\textnormal{$\mu$-GD}} )_{t \geq 0} \) represent the Glauber dynamics for the distribution \( \mu \) starting from the maximum all-1 state \( X^{(0)}_{\textnormal{$\mu$-GD}} = \*1_V \). 
Our main technical contribution is the following comparison lemma between the Glauber dynamics and the simulation algorithm $\+A_{\mu}(\theta,T_1,T_2)$.
%For any integer \( T \geq 1 \), let \( (P^{\text{GD}}_\mu)^T(\*1_V,\cdot) \) denote the distribution of the output when the Glauber dynamics starts from the maximum all-1 state \( \*1_V \) and makes transitions for $T$ steps.

\begin{lemma} \label{lem:comparison}
Let $\mu$ be a monotone system over $\{0,1\}^V$.
Let \( (X^{(t)}_{\textnormal{$\mu$-GD}} )_{t \geq 0} \) be the Glauber dynamics for $\mu$ starting from $X^{(0)}_{\textnormal{$\mu$-GD}} = \*1_V$.
Let $T_1,T_2\in\mathbb Z_{\geq1}$, $\theta \in (0,1)$, and $(X^{(t)}_{\textnormal{alg}})_{t=0}^{T_1T_2}$ be the Markov chain generated by $\+A_{\mu}(\theta,T_1,T_2)$.
For any integer $0 \leq t \leq T_1T_2$, 
\begin{align}\label{eq:comparison-dtv}
    \DTV{X^{(t)}_{\textnormal{$\mu$-GD}}}{\mu} \leq \DTV{\textnormal{\textsf{contr}}(X^{(t)}_{\textnormal{alg}})}{\mu}.
\end{align}
\end{lemma}

%The LHS of \eqref{eq:comparison-dtv} is the total variation distance between $\mu$ and the output of the standard Glauber dynamics.
%The RHS of \eqref{eq:comparison-dtv} is the total variation distance between $\mu$ and the output of the simulation algorithm $\+A_{\mu}(\theta,T_1,T_2)$.
The lemma implies that the $(T_1T_2)$-step Glauber dynamics from the maximum all-1 state returns a more accurate sample than the simulation algorithm $\+A_{\mu}(\theta,T_1,T_2)$. 

\Cref{thm:FD-comparison} is a direct corollary of \Cref{lem:simu-error} and \Cref{lem:comparison}.

\subsection{Stochastic dominance in the lifted space}
We outline the proof of \Cref{lem:comparison}. The algorithm $\+A_{\mu}(\theta,T_1,T_2)$ is a Markov chain over the space $\{0,1,\star\}^V$, which generates a sequence of random configurations $X^{(t)}_{\textnormal{alg}}$ for $t=0,1,\ldots,T_1T_2$. In particular, $X^{(0)}_{\textnormal{alg}}=\textsf{lift}(\*1_V)$.
Let $\pi^{(t)}_{\textnormal{alg}}$ denote the distribution of $X^{(t)}_{\textnormal{alg}}$. 

Since the Markov chain is over the space $\{0,1,\star\}^V$, we also lift the distribution $\mu$ to $\{0,1,\star\}^V$. Let $\pi$ denote the distribution of $\textsf{lift}(X)$ for $X \sim \mu$, i.e., 
\begin{align}\label{eq:lift-pi}
    \textsf{lift}(X) \sim \pi \text{ if } X \sim \mu.
\end{align}
By definition, it is easy to see that $\textsf{contr}(Y) \sim \mu$ if $Y \sim \pi$.
Let $\Omega(\mu) \subseteq \{0,1\}^V$ denote the support of $\mu$ and $\Omega(\pi) \subseteq \{0,1,\star\}^V$ denote the support of $\pi$.

Let $P_{\pi\textnormal{-GD}}$ be the transition matrix of the Glauber dynamics on $\pi$. Given the current configuration $X \in \Omega(\pi)$, $P_{\pi\textnormal{-GD}}$ picks a variable $v \in V$ uniformly at random and then resamples $X_v$ from $\pi$ conditional on $X_{V \setminus \{v\}}$. Let $(X^{(t)}_{\pi\textnormal{-GD}})_{t \geq 0}$ denote the sequence generated by the Glauber dynamics on $\pi$ starting from $X^{(0)}_{\pi\textnormal{-GD}}= \textsf{lift}(\*1_V)$. Let $\pi^{(t)}_{\textnormal{GD}}$ denote the distribution of $X^{(t)}_{\pi\textnormal{-GD}}$.

A key step in proving \Cref{lem:comparison} is a stochastic dominance relation between distributions $\pi^{(t)}_{\textnormal{alg}}$ and $\pi^{(t)}_{\textnormal{GD}}$.
To state this stochastic dominance result,
we need to extend the partial ordering $\preceq$ on $\{0,1\}^V$ to $\{0,1,\star\}^V$. 
Define the ordering
\begin{align*}
    0 < 1 < \star,
\end{align*}
and the following partial ordering for $X,Y \in \{0,1,\star\}^V$
\begin{align*}
    X \preceq Y \iff \forall v \in V, \,\, X_v \leq Y_v.
\end{align*}
With the above partial ordering, we can use \Cref{def:stochastic-dominance} to define a stochastic dominance relation $\sd$ between distributions over $\{0,1,\star\}^V$.
Note that $\pi^{(0)}_{\textnormal{alg}}=\pi^{(0)}_{\textnormal{GD}}$ at the beginning. The following lemma shows that $\pi^{(t)}_{\textnormal{GD}}$ is always stochastically dominated by $\pi^{(t)}_{\textnormal{alg}}$. 

\begin{lemma}\label{lem:stochastic-dominance}
If $\mu$ is a monotone system, then for every integer $0\leq t\leq T_1T_2$, $\pi^{(t)}_{\textnormal{GD}} \sd \pi^{(t)}_{\textnormal{alg}}$.
\end{lemma}

We now explain the intuition behind why \Cref{lem:stochastic-dominance} can be used to prove \Cref{lem:comparison}.
One can verify that $\pi$ is also a monotone system. 
Using the properties of $\pi^{(0)}_{\textnormal{GD}}$ (see \Cref{lem:pi-0-gd}) and the Glauber dynamics, one can verify that $\frac{\pi^{(t)}_{\textnormal{GD}}}{\pi}: \Omega(\pi) \to \mathbb{R}_{\geq 0}$ is an increasing function for any $t \geq 0$.\footnote{See \Cref{sec:preliminaries} for the definition of the function $\frac{\pi^{(t)}_{\textnormal{GD}}}{\pi}$ and the definition of increasing functions.
The statement holds because the Glauber dynamics on $\pi$ is reversible w.r.t. $\pi$ and stochastically monotone.
For a formal verification, one can refer to the proof in \Cref{sec:proof-of-the-comparison-result}. The analysis in \Cref{sec:proof-of-the-comparison-result} is for $\mu^{(t)}_{\textnormal{GD}}$ but the same analysis can be applied to $\pi^{(t)}_{\textnormal{GD}}$.
}  Then \Cref{lem:gen-compare} immediately implies that $d_{\text{TV}}(\pi^{(T_1T_2)}_{\textnormal{GD}},\pi) \leq d_{\text{TV}}(\pi^{(T_1T_2)}_{\textnormal{alg}},\pi)$, which means the Glauber dynamics generates a more accurate sample than the simulation algorithm.
This result is similar to the result in \Cref{lem:comparison}. The main difference is that \Cref{lem:comparison} focuses on $\mu$ instead of $\pi$. Specifically, $X^{(T_1T_2)}_{\mu\textnormal{-GD}}$ in \Cref{lem:comparison} is generated by the Glauber dynamics on $\mu$ and \Cref{lem:comparison} considers the contracted configuration $\textsf{contr}(X^{(T_1T_2)}_{\textnormal{alg}})$ instead of $X^{(T_1T_2)}_{\textnormal{alg}}$ itself. 
By the definition in~\eqref{eq:lift-pi}, $\mu$ and $\pi$ are closely related. 
This discrepancy can be resolved by analyzing the contracting operation and comparing the Glauber dynamics on $\mu$ to the Glauber dynamics on $\pi$. 
The details of proving \Cref{lem:comparison} assuming \Cref{lem:stochastic-dominance} are deferred to \Cref{sec:comparison-between-original-and-lifted-Glauber-dynamics-1}.

\subsection{Comparison between the Glauber dynamics and the simulation algorithm}

%\todo{discuss why cannot use standard censoring inequality to prove the comparison}

We outline the proof of \Cref{lem:stochastic-dominance}. Note that we cannot directly use the standard censoring inequality in \cite{peres2013can} to prove \Cref{lem:stochastic-dominance}.  Here are two reasons: (1) in Line \Cref{item:resample} of $\+A_{\mu}(\theta,T_1,T_2)$, an update on $v$ is ignored if $X_v = \star$. Whether an update at $v$ is ignored \emph{depends} on the current configuration $X$. (2) For a vertex $v$ with $X_v \neq \star$, the simulation algorithm and the Glauber dynamics on $\pi$ use \emph{different} rules to update $X_v$, where the simulation algorithm resamples from the conditional distribution induced by $(\theta * \mu)$ while the Glauber dynamics  resamples from the conditional distribution induced by $\pi$.

We use a more general technique in \cite{FillK13}, namely, verifying the comparison inequality relation in \Cref{def:comparison-inequality-relation}.
However, instead of comparing the algorithm with the Glauber dynamics on $\pi$, we compare the algorithm with a \emph{modified Glauber dynamics}. 
We first formulate the simulation algorithm $\+A_{\mu}(\theta,T_1,T_2)$ as a Markov chain over $\{0,1,\star\}^V$, then define the modified Glauber dynamics, and finally compare the algorithm with the modified Glauber dynamics.

\paragraph{Simulation algorithm as a Markov chain}
The algorithm consists of $T_1$ phases, where each phase begins with contracting and lifting operations, and then simulates the Glauber dynamics on $(\theta *\mu)^{\*1_{V \setminus S}}$ for $T_2$ steps (see \Cref{fig:lifted-simulation-algorithm} for an example). Define the following two transition matrices from the simulation algorithm.
\begin{itemize}
    \item $P_{\textnormal{cl}}$: For any $X \in \Omega(\pi)$, transform $X$ to a random $Y = \textsf{lift}(\textsf{contr}(X)) \in \Omega(\pi)$.
    \item $\psim$: For any $X \in \Omega(\pi)$, apply the transitions defined in Line~\ref{item:random-variable} and Line~\ref{item:resample} of the simulation algorithm to transform $X$ to a random $Y \in \Omega(\pi)$.
\end{itemize}
%Let $P_{\textnormal{phase}}$ be the transition matrix of a phase of the simulation algorithm. Since each phase starts with $P_{\textnormal{cl}}$ and repeats $\psim$ for $T_2$ times, we have
%\begin{align}\label{eq:phase-transition}
%    P_{\textnormal{phase}} = P_{\textnormal{cl}} \psim^{T_2}.
%\end{align}
%We remark that $\psim$ is one transition step of simulating the Glauber dynamics on $(\theta *\mu)^{\*1_{V \setminus S}}$. Hence, the transition matrix of one phase is $ P_{\textnormal{cl}} \psim^{T_2}$.
The distribution $\pi^{(t)}_{\textnormal{alg}}$ of $X^{(t)}_{\textnormal{alg}}$ satisfies the following recurrence relation
\begin{align}\label{eq:recurrence-relation-simulation-algorithm}
\pi^{(t+1)}_{\textnormal{alg}} = \pi^{(t)}_{\textnormal{alg}}P_{\textnormal{alg}}^{(t+1)}, \text{ where } P_{\textnormal{alg}}^{(t+1)} = \begin{cases}
    P_{\textnormal{cl}} \psim & \text{if } t \text{ mod } T_2 = 0, \\
    \psim & \text{if } t \text{ mod } T_2 > 0.
\end{cases}
\end{align}
We remark that in $P_{\textnormal{alg}}^{(t+1)}$, $(t+1)$ is the time index of the transition matrix but not the exponent. We use parenthesized superscripts to denote time indices.

%The algorithm consists of $T_1$ phases and the transition in each phase is $P_{\textnormal{phase}}$. The overall transition matrix of the algorithm is $P_{\textnormal{phase}}^{T_1} = (P_{\textnormal{cl}} \psim^{T_2})^{T_1}$, i.e. $\pi^{(T_1T_2)}_{\textnormal{alg}} = \pi^{(0)} P_{\textnormal{phase}}^{T_1}$.

\paragraph{A modified Glauber dynamics}
Recall that $P_{\pi\textnormal{-GD}}$ is the transition matrix of the Glauber dynamics on $\pi$. The distribution $\pi^{(t)}_{\textnormal{GD}}$ satisfies the following simple recurrence relation
\begin{align}\label{eq:recurrence-relation-Glauber-dynamics}
    \pi^{(t+1)}_{\textnormal{GD}} = \pi^{(t)}_{\textnormal{GD}} P_{\pi\textnormal{-GD}}.
\end{align}
Note that the two recurrences~\eqref{eq:recurrence-relation-simulation-algorithm} and~\eqref{eq:recurrence-relation-Glauber-dynamics} are in different forms because $(P_{\textnormal{alg}}^{(t)})_{t \geq 1}$ is time-inhomogeneous but $P_{\pi\textnormal{-GD}}$ is time-homogeneous. To do the comparison, we need to modify the Glauber dynamics. The modification is based on the following property of the chain $P_{\textnormal{cl}}$.
\begin{lemma}\label{lem:many-stationary}
Recall that $\textsf{lift}(\*1_V) \sim \pi^{(0)}_{\textnormal{GD}}$ and $\pi^{(t)}_{\textnormal{GD}} = \pi^{(t-1)}_{\textnormal{GD}}P_{\pi\textnormal{-GD}}$ is the sequence of distributions generated by the Glauber dynamics on $\pi$. Then for any $t \geq 0$,
\begin{align*}
    \pi^{(t)}_{\textnormal{GD}} P_{\textnormal{cl}} = \pi^{(t)}_{\textnormal{GD}}.
\end{align*}
\end{lemma}
The lemma states that every $\pi^{(t)}_{\textnormal{GD}}$ is a stationary distribution of $P_{\textnormal{cl}}$. 
The lemma is proved in \Cref{sec:proof-comparison-inequality-relation}.
Using this lemma, we can write the following new recurrence relation for $\pi^{(t)}_{\textnormal{GD}}$
\begin{align}\label{eq:recurrence-relation-Glauber-dynamics-modified}
   \pi^{(t+1)}_{\textnormal{GD}} = \pi^{(t)}_{\textnormal{GD}} P_{\pi\textnormal{-mGD}}^{(t+1)}, \text{ where } P_{\pi\textnormal{-mGD}}^{(t+1)} = \begin{cases}
        P_{\textnormal{cl}} P_{\pi\textnormal{-GD}} & \text{if } t \text{ mod } T_2 = 0, \\
        P_{\pi\textnormal{-GD}} & \text{if } t \text{ mod } T_2 > 0.
    \end{cases}
\end{align}
The above recurrence relation gives a modified time-inhomogeneous Markov chain $(P_{\pi\textnormal{-mGD}}^{(t)})_{t \geq 1}$ for $\pi$. For the transition from time $t$ to time $t+1$, if $t \text{ mod } T_2 = 0$, then the transition matrix is $P_{\textnormal{cl}} P_{\pi\textnormal{-GD}}$; otherwise, the transition matrix is $P_{\pi\textnormal{-GD}}$. 
Thanks to \Cref{lem:many-stationary}, this modified Markov chain generates the same sequence of distributions $(\pi^{(t)}_{\textnormal{GD}})_{t \geq 1}$ as the Glauber dynamics on $\pi$. Now, the good property is that $(P_{\textnormal{alg}}^{(t)})_{t \geq 1}$ and $(P_{\pi\textnormal{-mGD}}^{(t)})_{t \geq 1}$ are in the same form.

\paragraph{Comparison between two Markov chains}
We now apply the general technique in \Cref{sec:preliminaries} to compare the two time-inhomogeneous Markov chains in~\eqref{eq:recurrence-relation-simulation-algorithm} and~\eqref{eq:recurrence-relation-Glauber-dynamics-modified}. 
Recall that by~\eqref{eq:mc-set}, $\+{MC}_\pi$ is the family of Markov chains that are reversible with respect to $\pi$ and stochastically monotone, where the stochastic monotonicity property is defined in \Cref{def:stochastically-monotone}.
The following properties for $P_{\pi\textnormal{-GD}}$, $P_{\textnormal{cl}}$, and $\psim$ can be verified by their definitions. 
\begin{lemma}\label{lem:properties-of-transition-matrices}
    If $\mu$ is a monotone system, then $P_{\textnormal{cl}} \in \+{MC}_\pi$, $\psim \in \+{MC}_\pi$, and $P_{\pi\textnormal{-GD}} \in \+{MC}_\pi$.
\end{lemma}
The proof of \Cref{lem:properties-of-transition-matrices} is deferred to \Cref{sec:proof-comparison-inequality-relation}. %The above lemma implies $P_{\textnormal{cl}}$ is reversible with respect to $\pi$, which means $\pi$ is also one of the stationary distributions of $P_{\textnormal{cl}}$.

Now, we use the comparison inequality relation in \Cref{def:comparison-inequality-relation} to prove the following
stochastic dominance statement.
\begin{lemma}\label{lem:phase-comparison}
    If $\mu$ is a monotone system, then for any distribution $\nu$ such that $\nu/\pi$ is increasing, $\nu P_{\pi\textnormal{-mGD}}^{(t)} \sd \nu P_{\textnormal{alg}}^{(t)}$ for all $t \geq 1$. 
\end{lemma}
By \Cref{lem:pi-0-gd}, the common initial distribution has an increasing density relative to $\pi$.
Since $P_{\pi\textnormal{-GD}}\in\+{MC}_\pi$, this remains true for $\pi^{(t)}_{\textnormal{GD}}$.
Each $P_{\textnormal{alg}}^{(t)}$ is stochastically monotone, so \Cref{lem:phase-comparison} and induction on $t$ give $\pi^{(t)}_{\textnormal{GD}}\sd\pi^{(t)}_{\textnormal{alg}}$, proving \Cref{lem:stochastic-dominance}.
The formal proof appears in \Cref{sec:proof-stochastic-dominance}.

\begin{remark}
In the above proof, we prove $\pi^{(t)}_{\textnormal{GD}} \sd \pi^{(t)}_{\textnormal{alg}}$ in \Cref{lem:stochastic-dominance} via the comparison between the simulation algorithm and the modified Glauber dynamics in \Cref{lem:phase-comparison}. One might wonder whether we can directly compare the simulation algorithm and the standard Glauber dynamics by showing that $\nu P_{\pi\textnormal{-GD}} \sd \nu P_{\textnormal{alg}}^{(t)}$ for all $t \geq 1$ and all distributions $\nu$ such that $\nu/\pi$ is increasing. The answer is no. A counterexample is given in Appendix \ref{sec:counterexample}. Hence, although $\pi^{(t)}_{\textnormal{GD}}$ is defined by the standard Glauber dynamics, when proving the stochastic dominance result for $\pi^{(t)}_{\textnormal{GD}}$, it is \emph{not} necessary to apply the comparison argument to the standard Glauber dynamics. This idea may be of independent interest.
\end{remark}

We give some high-level ideas for the proof of \Cref{lem:phase-comparison}. By comparing~\eqref{eq:recurrence-relation-simulation-algorithm} and~\eqref{eq:recurrence-relation-Glauber-dynamics-modified}, it suffices to show that $P_{\pi\textnormal{-GD}} \mc \psim$. In both $P_{\pi\textnormal{-GD}} $ and $\psim$, the Markov chains first pick a vertex $i \in V$ uniformly at random, and then update the value of $i$ according to their rules. Let $P_{\pi\textnormal{-GD}}^i$ and $\psim^i$ denote the transition matrices conditional on picking vertex $i \in V$. It suffices to show that $P_{\pi\textnormal{-GD}}^i \mc \psim^i$. Fix a distribution $\nu$ such that $\frac{\nu}{\pi}$ is an increasing function. We need to show that $\nu P_{\pi\textnormal{-GD}}^i \sd \nu \psim^i$. One challenge is as follows. Suppose $X \sim \nu$. If $X_i \in \{0,1\}$, then $P_{\pi\textnormal{-GD}}^i$ may update the value of $X_i$ to $\star$, but $\psim$ must update the value of $X_i$ to $\{0,1\}$. Since $\star > 1 > 0$, it seems that $P_{\pi\textnormal{-GD}}^i$ can update $X$ to a larger configuration. However, we can still show $\nu P_{\pi\textnormal{-GD}}^i \sd \nu \psim^i$ by using the fact that $\frac{\nu}{\pi}$ is an increasing function. The intuition is that the increasing property suggests that $X_i$ can also take the value $\star$ with certain probability. In the case of $X_i = \star$, $\psim^i$ always keeps the value of $X_i$ as $\star$. Overall, we can show the desired stochastic dominance result. The detailed proof of \Cref{lem:phase-comparison} is given in \Cref{sec:proof-dominance}.

\paragraph{Paper organization}
In \Cref{sec:comparison-between-original-and-lifted-Glauber-dynamics-1}, we first prove \Cref{lem:simu-error} and then prove \Cref{lem:comparison} assuming \Cref{lem:stochastic-dominance}. After that, our task is reduced to proving \Cref{lem:stochastic-dominance}. 
In \Cref{sec:proof-comparison-inequality-relation}, we verify some properties of Markov chains stated in \Cref{lem:many-stationary} and \Cref{lem:properties-of-transition-matrices}. With these properties, we can prove \Cref{lem:stochastic-dominance} assuming \Cref{lem:phase-comparison} in \Cref{sec:proof-stochastic-dominance}. We prove \Cref{lem:phase-comparison} in \Cref{sec:proof-dominance}. This finishes the proof of the main technical result in \Cref{thm:FD-comparison}. Finally,  \Cref{sec:proofs-of-general-results-and-applications} is devoted to proving \Cref{thm:main-comparison} and applications.

\section{Analysis of Glauber dynamics and simulation algorithm}\label{sec:comparison-between-original-and-lifted-Glauber-dynamics-1}
Our main result in \Cref{thm:FD-comparison} is a corollary of \Cref{lem:simu-error} and \Cref{lem:comparison}.
In this section, we first prove \Cref{lem:simu-error} in \Cref{sec:simu-error}.
Then, we prove \Cref{lem:comparison} assuming the stochastic dominance in \Cref{lem:stochastic-dominance} holds.
Specifically, in \Cref{sec:comparison-between-original-and-lifted-Glauber-dynamics-2} and \Cref{sec:analysis-of-contracting-operation}, we give a comparison result of the Glauber dynamics on $\pi$ and $\mu$ and some properties of the lifting and contracting operations, which are used in the proof of \Cref{lem:comparison} in \Cref{sec:proof-of-the-comparison-result}.

\subsection{Correctness of the simulation algorithm}\label{sec:simu-error}
\begin{proof}[Proof of \Cref{lem:simu-error}]
If we can perfectly simulate the field dynamics for $T_1$ steps, then the total variation distance between the output and $\mu$ is at most $\epsilon/2$. In the simulation algorithm, every transition of the field dynamics is simulated by running the Glauber dynamics for $T_2$ steps. By the definition of the mixing time of tilted distributions, for each transition, we can couple the perfect simulation with the simulation algorithm successfully with probability at least $1-\epsilon/(2T_1)$. By a union bound over $T_1$ steps, the total variation distance between the output of the simulation algorithm and the output of the perfect simulation is at most $\epsilon/2$. The lemma follows by the triangle inequality of the total variation distance.
\end{proof}

\subsection{Comparison between Glauber dynamics on \texorpdfstring{$\pi$}{} and \texorpdfstring{$\mu$}{}}\label{sec:comparison-between-original-and-lifted-Glauber-dynamics-2}

We need a middle step to relate the Glauber dynamics on $\pi$ with the Glauber dynamics on $\mu$. 
Recall that $P_{\pi\textnormal{-GD}}$ and $P_{\mu\textnormal{-GD}}$ are the transition matrices of the Glauber dynamics on $\pi$ and $\mu$, respectively. 
Let $(X^{(t)}_{\textnormal{$\pi$-GD}})_{t \geq 0}$ (resp. $(X^{(t)}_{\textnormal{$\mu$-GD}})_{t \geq 0}$) be the sequence generated by the Glauber dynamics for $\pi$ (resp. $\mu$) starting from $X^{(0)}_{\textnormal{$\pi$-GD}} = \textsf{lift}(\*1_V)$ (resp. $X^{(0)}_{\textnormal{$\mu$-GD}} = \*1_V$).
\begin{lemma}\label{lem:comparison-lifted}
    For any integer $t \geq 0$, $X^{(t)}_{\textnormal{$\pi$-GD}}$ and $\textnormal{\textsf{lift}}(X^{(t)}_{\textnormal{$\mu$-GD}})$ are identically distributed, and moreover, $\textnormal{\textsf{contr}}(X^{(t)}_{\textnormal{$\pi$-GD}})$ and $X^{(t)}_{\textnormal{$\mu$-GD}}$ are identically distributed.
\end{lemma}
\begin{proof}

    We first use induction to show that $X^{(t)}_{\textnormal{$\pi$-GD}}$ and $\textnormal{\textsf{lift}}(X^{(t)}_{\textnormal{$\mu$-GD}})$ are identically distributed for any integer $t\geq 0$.
    Let $\mu^{(t)}_{\textnormal{GD}}$ be the distribution of $X^{(t)}_{\mu\textnormal{-GD}}$.
    Recall that $\pi^{(t)}_{\textnormal{GD}}$ denotes the distribution of $X^{(t)}_{\pi\textnormal{-GD}}$.
    When $t=0$, $X^{(0)}_{\pi\textnormal{-GD}}=\textnormal{lift}(\*1_V)=\textnormal{lift}(X^{(0)}_{\mu\textnormal{-GD}})$.

    Assume $X^{(t)}_{\textnormal{$\pi$-GD}}$ and $\textnormal{\textsf{lift}}(X^{(t)}_{\textnormal{$\mu$-GD}})$ are identically distributed for some $t\in \mathbb{N}$. We prove the statement for $t+1$.
    Let $\mu^{(t+1)}_{\textnormal{GD}}\mid_i$ be the distribution of $X^{(t+1)}_{\textnormal{$\mu$-GD}}$ conditional on $i\in V$ being updated by the Glauber dynamics. Then we have $\mu^{(t+1)}_{\textnormal{GD}}=\frac{1}{n}\sum_{i\in V}\mu^{(t+1)}_{\textnormal{GD}}\mid_i$. Similarly we can define $\pi^{(t+1)}_{\textnormal{GD}}\mid_i$ and $\pi^{(t+1)}_{\textnormal{GD}}=\frac{1}{n}\sum_{i\in V}\pi^{(t+1)}_{\textnormal{GD}}\mid_i$. Define two functions $C_1,C_\star:\{0,1,\star\}^V \to \mathbb{N}$ that count the number of 1s and $\star$s in a configuration. Formally, for any $\sigma\in \{0,1,\star\}^V$,
    \begin{align}\label{def:c-sigma}
    \begin{split}
    C_1(\sigma)=\abs{\{i\in V\mid\sigma_i=1 \}};\\
    C_\star(\sigma)=\abs{\{i\in V\mid\sigma_i=\star \}}.   
    \end{split}
    \end{align}
    Now, we fix a variable $i \in V$ and fix a configuration $\sigma\in \Omega(\pi) \subseteq \{0,1,\star\}^V$. Define $\tau=\textnormal{contr}(\sigma) \in \Omega(\mu)$. We only need to consider feasible $\sigma,\tau$ because Glauber dynamics stays in the feasible space.
    Define $\sigma^{x}$ and $\tau^y$ for $x\in \{0,1,\star\}$ and $y\in\{0,1\}$ by
    \begin{align*}
    &\sigma^x_i=x,\quad \tau^y_i = y;\\
    &\sigma^x_j = \sigma_j,\quad \tau^y_j=\tau_j,\quad \forall j\in V\setminus \{i\}.
    \end{align*}
    In words, $\sigma^x$ and $\tau^y$ are configurations obtained from $\sigma$ and $\tau$ by changing the value of $i$ to $x$ and $y$, respectively.
    We can compute $(\pi^{(t+1)}_{\textnormal{GD}}\mid_i)(\sigma)$ and $(\mu^{(t+1)}_{\textnormal{GD}}\mid_i)(\tau)$ in terms of $\pi^{(t)}_{\textnormal{GD}}$ and $\mu^{(t)}_{\textnormal{GD}}$, respectively. Since Glauber dynamics resamples the value at $i$ from the conditional distribution, we have
    \begin{align}\label{eq:update-pi-mu}
    \begin{split}
     (\pi^{(t+1)}_{\textnormal{GD}}\mid_i)(\sigma) &= \sum_{x\in \{0,1,\star\}}\pi^{(t)}_{\textnormal{GD}}(\sigma^x)\frac{\pi(\sigma)}{\sum_{z\in \{0,1,\star\}}\pi(\sigma^z)};\\
    (\mu^{(t+1)}_{\textnormal{GD}}\mid_i)(\tau) &= \sum_{y\in \{0,1\}}\mu^{(t)}_{\textnormal{GD}}(\tau^y)\frac{\mu(\tau)}{\sum_{z\in \{0,1\}}\mu(\tau^z)}.   
    \end{split}
    \end{align}
    By the definition of $\pi$, $\textsf{lift}(\mu)=\pi$ and $\pi(\sigma) = \mu(\tau)\theta^{C_1(\sigma)}(1-\theta)^{C_\star(\sigma)}$. Hence, $\pi(\sigma^1) + \pi(\sigma^\star)$ equals $\mu(\tau^1) \theta^{C_1(\sigma^0)}(1-\theta)^{C_\star(\sigma^0)}$. This is because $\tau^1$ is lifted to $\sigma^1$ or $\sigma^\star$ if and only if every $j \neq i$ with $\tau_j = 1$ is lifted correctly.
    It holds that
    \begin{align*}
    \frac{\pi(\sigma)}{\sum_{z\in \{0,1,\star\}}\pi(\sigma^z)} &=\frac{\mu(\tau)\theta^{C_1(\sigma)}(1-\theta)^{C_\star(\sigma)}}{\mu(\tau^0)\theta^{C_1(\sigma^0)}(1-\theta)^{C_\star(\sigma^0)}+\mu(\tau^1)\theta^{C_1(\sigma^0)}(1-\theta)^{C_\star(\sigma^0)}}\\
    &=\frac{\mu(\tau)\theta^{\*1\{\sigma_i=1\}} (1-\theta)^{\*1\{\sigma_i=\star\}}}{\mu(\tau^0)+\mu(\tau^1)},\quad \text{where $\*1\{\cdot\}$ is the indicator function.}
    \end{align*}
    By the induction hypothesis, $\textsf{lift}(\mu^{(t)}_{\textnormal{GD}})=\pi^{(t)}_{\textnormal{GD}}$. By the definition of the lifting, we can compute the sum of $\pi^{(t)}_{\textnormal{GD}}(\sigma^z)$ for $z\in \{0,1,\star\}$. Again, the sum of $\pi^{(t)}_{\textnormal{GD}}(\sigma^1) + \pi^{(t)}_{\textnormal{GD}}(\sigma^\star)$ equals $\mu^{(t)}_{\textnormal{GD}}(\tau^1) \theta^{C_1(\sigma^0)}(1-\theta)^{C_\star(\sigma^0)}$. We have the following equation
    \begin{align*}
    \sum_{z\in \{0,1,\star\}}\pi^{(t)}_{\textnormal{GD}}(\sigma^z)&=\mu^{(t)}_{\textnormal{GD}}(\tau^0)\theta^{C_1(\sigma^0)}(1-\theta)^{C_\star(\sigma^0)}+\mu^{(t)}_{\textnormal{GD}}(\tau^1)\theta^{C_1(\sigma^0)}(1-\theta)^{C_\star(\sigma^0)}\\
    &=(\mu^{(t)}_{\textnormal{GD}}(\tau^0)+\mu^{(t)}_{\textnormal{GD}}(\tau^1))\theta^{C_1(\sigma^0)}(1-\theta)^{C_\star(\sigma^0)}.
    \end{align*}
    Notice that $\theta^{\*1\{\sigma_i=1\}} (1-\theta)^{\*1\{\sigma_i=\star\}}\theta^{C_1(\sigma^0)}(1-\theta)^{C_\star(\sigma^0)}=\theta^{C_1(\sigma)}(1-\theta)^{C_\star(\sigma)}$. Combining the above two equations and~\eqref{eq:update-pi-mu}, we have
    \begin{align*}
    (\pi^{(t+1)}_{\textnormal{GD}}\mid_i)(\sigma)&= \frac{\mu(\tau)\theta^{\*1\{\sigma_i=1\}} (1-\theta)^{\*1\{\sigma_i=\star\}}}{\mu(\tau^0)+\mu(\tau^1)} (\mu^{(t)}_{\textnormal{GD}}(\tau^0)+\mu^{(t)}_{\textnormal{GD}}(\tau^1))\theta^{C_1(\sigma^0)}(1-\theta)^{C_\star(\sigma^0)}\\
    &=\theta^{C_1(\sigma)}(1-\theta)^{C_\star(\sigma)}\sum_{x\in \{0,1\}}\mu^{(t)}_{\textnormal{GD}}(\tau^x) \frac{\mu(\tau)}{\sum_{z\in \{0,1\}}\mu(\tau^z)}\\
    &=\theta^{C_1(\sigma)}(1-\theta)^{C_\star(\sigma)} \cdot (\mu^{(t+1)}_{\textnormal{GD}}\mid_i)(\tau).
    \end{align*}
    The above equation can be interpreted as follows: to draw a random sample $\sigma$ from $\pi^{(t+1)}_{\textnormal{GD}}\mid_i$, we first draw a sample $\tau \sim \mu^{(t+1)}_{\textnormal{GD}}\mid_i$, and then apply the lifting operation to $\tau$. We slightly abuse the notation by using $\textsf{lift}(\nu)$ to denote the distribution of $\textsf{lift}(X)$ for $X \sim \nu$. The above equation shows $\pi^{(t+1)}_{\textnormal{GD}}\mid_i = \textsf{lift}(\mu^{(t+1)}_{\textnormal{GD}}\mid_i)$. We have 
    \begin{align*}
    \pi^{(t+1)}_{\textnormal{GD}}=\frac{1}{n}\sum_{i\in V}\pi^{(t+1)}_{\textnormal{GD}}\mid_i = \frac{1}{n} \sum_{i\in V} \textnormal{\textsf{lift}}(\mu^{(t+1)}_{\textnormal{GD}}\mid_i) =\textnormal{\textsf{lift}}(\mu^{(t+1)}_{\textnormal{GD}}),
    \end{align*}
    where the last equality follows from $\mu^{(t+1)}_{\textnormal{GD}} = \frac{1}{n}\sum_{i\in V} \mu^{(t+1)}_{\textnormal{GD}}\mid_i$ and it is straightforward to verify that $\frac{1}{n} \sum_{i\in V} \textnormal{\textsf{lift}}(\mu^{(t+1)}_{\textnormal{GD}}\mid_i) =\textnormal{\textsf{lift}}(\mu^{(t+1)}_{\textnormal{GD}})$. Hence, the above equation shows that $X^{(t+1)}_{\textnormal{$\pi$-GD}} \sim \pi^{(t+1)}_{\textnormal{GD}}$ and $\textnormal{\textsf{lift}}(X^{(t+1)}_{\textnormal{$\mu$-GD}}) \sim \textsf{lift}(\mu^{(t+1)}_{\textnormal{GD}})$ are identically distributed. This proves the induction step.
    
    We next prove the second result. We slightly abuse the notation by using $\textsf{contr}(\nu)$ to denote the distribution of $\textsf{contr}(X)$ for $X \sim \nu$. 
    For any $\tau \in \{0,1\}^V$, it is easy to see $\tau = \textsf{contr}(\textsf{lift}(\tau))$. 
    Hence, for any distribution $\nu$ over $\{0,1\}^V$,  $\nu  =\textsf{contr}(\textsf{lift}(\nu))$. Then for any integer $t\geq 0$, 
    \begin{align*}
    \textnormal{\textsf{contr}}(\pi^{(t)}_{\textnormal{GD}}) = \textnormal{\textsf{contr}}(\textsf{lift}(\mu^{(t)}_{\textnormal{GD}}))= \mu^{(t)}_{\textnormal{GD}}, 
    \end{align*}
    where the first equality holds because of the first result $\pi^{(t)}_{\textnormal{GD}} = \textsf{lift}(\mu^{(t)}_{\textnormal{GD}})$ of the lemma.
    \end{proof}

\subsection{Analysis of the contracting operation}\label{sec:analysis-of-contracting-operation}
We also need the following lemma, which shows that the stochastic dominance relation is preserved by the $\textsf{contr}$ operation. The following lemma gives the result for both the lifting and contracting operations. The result for the lifting operation will be used in later sections.
\begin{lemma}\label{lem:lift-and-contr-keep-dominate} 
Let $\nu$ and $\nu'$ be two distributions over $\{0,1\}^V$. Let $\nu_{L}$ be the distribution of $\textsf{lift}(X)$ for $X \sim \nu$ and $\nu_{L}'$ be the distribution of $\textsf{lift}(Y)$ for $Y \sim \nu'$. If $\nu \sd \nu'$, then $\nu_{L} \sd \nu_{L}'$.

Let $\nu$ and $\nu'$ be two distributions over  $\{0,1,\star\}^V$. Let $\nu_{C}$ be the distribution of $\textsf{contr}(X)$ for $X \sim \nu$ and $\nu_{C}'$ the distribution of $\textsf{contr}(Y)$ for $Y \sim \nu'$. If $\nu \sd \nu'$, then $\nu_{C} \sd \nu_{C}'$.
\end{lemma}
\begin{proof}   
Consider two distributions $\nu,\nu'$ over $\{0,1\}^V$ such that $\nu\sd \nu'$. By \Cref{def:stochastic-dominance}, there exists a coupling $(X,Y)$ where $X\sim \nu$, $Y\sim \nu'$ such that $X\preceq Y$. We use this $(X,Y)$ to construct a coupling $(P,Q)$ between $\nu_L$ and $\nu_L'$.
For every $i\in V$, $X_i\leq Y_i$. Since the lifting operation lifts each value of $i\in V$ independently, we only need to couple the lifting operation on each $i \in V$ such that $\textsf{lift}(X)_i\leq \textsf{lift}(Y)_i$.
If $X_i=Y_i$, we couple the lifted values so that $\textsf{lift}(X)_i=\textsf{lift}(Y)_i$. If $X_i \neq Y_i$, then the only possible case is $X_i=0$ and $Y_i=1$, and $\textsf{lift}(X)_i=0$, so $\textsf{lift}(X)_i < \textsf{lift}(Y)_i$ because $\textsf{lift}(Y)_i \in \{1,\star\}$.

The contracting operation is a deterministic operation. The second property is easier to verify. 
Consider $\nu,\nu'$ over $\{0,1,\star\}^V$ satisfying $\nu\sd \nu'$.
There exists a coupling $(X,Y)$ where $X\sim \nu$, $Y\sim \nu'$ such that $X\preceq Y$.
We apply the contracting operation to both $X$ and $Y$ to obtain a coupling between $\nu_C$ and $\nu_C'$.
For any $i\in V$, if $X_i=0$, then $\textsf{contr}(X)_i=0\leq \textsf{contr}(Y)_i$. If $X_i=1$ or $X_i=\star$, then $Y_i=1$ or $Y_i=\star$, which implies $\textsf{contr}(X)_i= \textsf{contr}(Y)_i=1$. Hence, there is a coupling between $\nu_C$ and $\nu_C'$ such that $\textsf{contr}(X) \preceq \textsf{contr}(Y)$.
\end{proof}

\subsection{Proof of the comparison result}\label{sec:proof-of-the-comparison-result}
To prove \Cref{lem:comparison}, we need the following lemma.
\begin{lemma}\label{lem:reversible-transition}
Let $P$ be a reversible chain with respect to the distribution $\mu$. For any distribution $\nu$, $P \frac{\nu}{\mu} = \frac{\nu P}{\mu}$, where $\frac{\nu}{\mu}$ is a column vector such that $\frac{\nu}{\mu}(\sigma) = \frac{\nu(\sigma)}{\mu(\sigma)}$ and $\nu,\mu$ are row vectors.
\end{lemma}
\begin{proof}
    To verify this, for any $x$, we have $(P \frac{\nu}{\mu})(x) = \sum_{y} P(x,y) \frac{\nu(y)}{\mu(y)}$. Since $P$ is reversible, it can be written as $\sum_{y} P(y,x) \frac{\nu(y)}{\mu(x)} = \frac{(\nu P)(x)}{\mu(x)}$.
\end{proof}
Now, we are ready to prove \Cref{lem:comparison}.
\begin{proof}[Proof of \Cref{lem:comparison} assuming \Cref{lem:stochastic-dominance}]
    By \Cref{lem:stochastic-dominance}, $\pi^{(t)}_{\textnormal{GD}} \sd \pi^{(t)}_{\textnormal{alg}}$ for every integer $0\leq t\leq T_1T_2$. For any distribution $\nu$, we use $\textsf{contr}(\nu)$ to denote the distribution of $\textsf{contr}(X)$ for $X \sim \nu$. By \Cref{lem:comparison-lifted}, $\textsf{contr}(\pi^{(t)}_{\textnormal{GD}}) = \mu^{(t)}_{\textnormal{GD}}$.
    Using \Cref{lem:lift-and-contr-keep-dominate}, we have $\textsf{contr}(\pi^{(t)}_{\textnormal{GD}}) = \mu^{(t)}_{\textnormal{GD}} \sd \textsf{contr}(\pi^{(t)}_{\textnormal{alg}})$.
   
    We next show that $\frac{\mu^{(t)}_{\textnormal{GD}}}{\mu}: \Omega(\mu) \to \mathbb{R}_{\geq 0}$ is an increasing function.
   Since $\mu$ is a monotone system, by \Cref{def:monotone-system}, the third definition of a stochastically monotone chain in \Cref{prop:stochastically-monotone-equivalent}, and the transition of Glauber dynamics, it is straightforward to verify that the Glauber dynamics $P_{\mu\textnormal{-GD}}$ is stochastically monotone. Since $\mu^{(0)}_{\textnormal{GD}}$ takes the maximum configuration $\*1_V$ with probability 1, $\frac{\mu^{(0)}_{\textnormal{GD}}}{\mu}$ is an increasing function. Because $P_{\mu\textnormal{-GD}}$ is reversible with respect to $\mu$, by \Cref{lem:reversible-transition}, for any integer $i\geq 0$,
    $P_{\mu\textnormal{-GD}} \frac{\mu^{(i)}_{\textnormal{GD}}}{\mu} = \frac{\mu^{(i)}_{\textnormal{GD}}P_{\mu\textnormal{-GD}}}{\mu} = \frac{\mu^{(i+1)}_{\textnormal{GD}}}{\mu}$.
    Since $P_{\mu\textnormal{-GD}}$ is stochastically monotone, using \Cref{def:stochastically-monotone}, we verify that $\frac{\mu^{(t)}_{\textnormal{GD}}}{\mu}: \Omega(\mu) \to \mathbb{R}_{\geq 0}$ is an increasing function.

    Note that $X^{(t)}_{\textnormal{alg}}$ in \Cref{lem:comparison} follows the law of $\pi^{(t)}_{\textnormal{alg}}$ and $X^{(t)}_{\textnormal{$\mu$-GD}}$ follows the law of $\mu^{(t)}_{\textnormal{GD}}$.
    With the above two properties, the lemma is proved by \Cref{lem:gen-compare}.
\end{proof}

\section{Verification of properties of Markov chains}\label{sec:proof-comparison-inequality-relation}
In this section, we prove the properties of Markov chains stated in \Cref{lem:many-stationary} and \Cref{lem:properties-of-transition-matrices}.

\Cref{lem:many-stationary} is a straightforward consequence of \Cref{lem:comparison-lifted}.
\begin{proof}[Proof of \Cref{lem:many-stationary}]
The Markov chain $P_{\textnormal{cl}}$ can be decomposed as $P_CP_L$. Using \Cref{lem:comparison-lifted}, $\pi^{(t)}_{\textnormal{GD}} P_C = \mu^{(t)}_{\textnormal{GD}}$ and $\mu^{(t)}_{\textnormal{GD}} P_L = \pi^{(t)}_{\textnormal{GD}}$. Therefore, $\pi^{(t)}_{\textnormal{GD}} P_{\textnormal{cl}} = \pi^{(t)}_{\textnormal{GD}}$, where $\mu^{(t)}_{\textnormal{GD}}$ is the distribution of $X^{(t)}_{\mu\textnormal{-GD}}$ in \Cref{lem:comparison-lifted}. These distributions are generated by Glauber dynamics on $\mu$ starting from $\*1_V$. This proves the lemma.

The above proof suggests that for any distribution $\rho$ over $\Omega(\mu)$, $\rho P_{\textnormal{L}}$ over $\Omega(\pi)$ is a stationary distribution of the chain $P_{\textnormal{cl}}$. This is because $\rho P_{\textnormal{L}}P_{\textnormal{cl}} = \rho P_{\textnormal{L}}P_{\textnormal{C}}P_{\textnormal{L}} = \rho P_{\textnormal{L}}$ as $P_{\textnormal{L}}P_{\textnormal{C}} = I$.
\end{proof}

To prove \Cref{lem:properties-of-transition-matrices}, we need to show that the Markov chains $P_{\pi\textnormal{-GD}}$, $P_{\textnormal{cl}}$ and $\psim$ are stochastically monotone and reversible. We prove these two properties in \Cref{sec:proof-stochastically-monotone} and \Cref{sec:proof-reversible}, respectively.

\subsection{Stochastic monotonicity of Markov chains}\label{sec:proof-stochastically-monotone}
We first show that $\pi$ is a monotone system over $\Omega(\pi) \subseteq \{0,1,\star\}^V$.
We extend the definition of a monotone system in \Cref{def:monotone-system} to the case where the distribution is over $\{0,1,\star\}^V$. We say the lifted distribution $\pi$ is a monotone system if for any $i \in V$, for any feasible and comparable configurations $\sigma\preceq \tau$ with $\sigma,\tau \in \{0,1,\star\}^{V \setminus \{i\}}$, the conditional distribution $\pi^{\sigma}_i$ is stochastically dominated by $\pi^{\tau}_i$, which means there exists a coupling $(x,y)$ of $\pi^{\sigma}_i$ and $\pi^{\tau}_i$ such that $x\leq y$ with probability~1, where we use the ordering $0 < 1 < \star$.
By the monotonicity of the distribution $\mu$, we have the following property for the lifted distribution $\pi$, which implies that the Glauber dynamics $P_{\pi\textnormal{-GD}}$ is stochastically monotone.
\begin{lemma}\label{lem:monotone-pi}
If $\mu$ is a monotone system, then $\pi$ is also a monotone system. This implies the Glauber dynamics $P_{\pi\textnormal{-GD}}$ on $\pi$ is stochastically monotone.
\end{lemma}
\begin{proof}
Given any configuration $\sigma\in \Omega(\pi)$ and any $i \in V$, for any $x\in \{0,1,\star\}$, let $\sigma^x$ be the configuration obtained from $\sigma$ by setting $\sigma_i=x$.
The conditional distribution of $\pi$ is
\begin{align*}
\pi^{\sigma_{V \setminus \{i\}}}_i(0)=\frac{\pi(\sigma^0)}{\pi(\sigma^0)+\pi(\sigma^\star)+\pi(\sigma^1)}.
\end{align*}
Note that $\pi(\sigma^0) = \mu(\textsf{contr}(\sigma^0)) \theta^{C_1(\sigma^0)}(1-\theta)^{C_\star(\sigma^0)}$, where $C_1$ and $C_\star$ are defined in~\eqref{def:c-sigma}, and $\pi(\sigma^\star) + \pi(\sigma^1) = \mu(\textsf{contr}(\sigma^1)) \theta^{C_1(\sigma^1)-1}(1-\theta)^{C_\star(\sigma^1)}$. The second equation holds because to obtain $\sigma^1$ or $\sigma^\star$ from $\pi$, we first need to sample $X \sim \mu$ such that $X = \textsf{contr}(\sigma^1)$, and for any $j \neq i$ with $X_j = 1$, if $\sigma_j = 1$, we need to keep the value of $X_j$ in the lifting operation, which happens with probability $\theta$; if $\sigma_j = \star$, we need to lift the value of $X_j$ to $\star$ in the lifting operation, which happens with probability $1-\theta$. Combining these two equations, we have 
\begin{align}\label{eq:pi-0-i}
    \pi^{\sigma_{V \setminus \{i\}}}_i(0) = \frac{\mu(\textsf{contr}(\sigma^0))}{\mu(\textsf{contr}(\sigma^0))+\mu(\textsf{contr}(\sigma^1))}.
\end{align}
Similarly, we have 
\begin{align}\label{eq:pi-1-i}
    \pi^{\sigma_{V \setminus \{i\}}}_i(1) = \frac{\theta \mu(\textsf{contr}(\sigma^1))}{\mu(\textsf{contr}(\sigma^0))+\mu(\textsf{contr}(\sigma^1))},\quad \pi^{\sigma_{V \setminus \{i\}}}_i(\star) =  \frac{(1-\theta)\mu(\textsf{contr}(\sigma^1))}{\mu(\textsf{contr}(\sigma^0))+\mu(\textsf{contr}(\sigma^1))}.
\end{align}

Hence, we can use the following process to sample from the distribution $\pi^{\sigma_{V \setminus \{i\}}}_i$: (1) Sample $x \sim  \mu_i^{\textsf{contr}(\sigma)_{V \setminus \{i\}}}$; (2) If $x = 0$, then output $0$; (3) If $x = 1$, then output $1$ with probability $\theta$ and output $\star$ with the remaining probability $1-\theta$.

For $\sigma, \tau \in \Omega(\pi)$ such that $\sigma\preceq \tau$, construct the following coupling of $\pi^{\sigma_{V \setminus \{i\}}}_i$ and $\pi^{\tau_{V \setminus \{i\}}}_i$.
\begin{enumerate}
    \item Since $\sigma \preceq \tau$, we have $\textsf{contr}(\sigma) \preceq \textsf{contr}(\tau)$. Since $\mu$ is a monotone system, we can sample $(x,y)$ from the coupling of $\mu^{\textsf{contr}(\sigma)_{V \setminus \{i\}}}_i$ and $\mu^{\textsf{contr}(\tau)_{V \setminus \{i\}}}_i$ such that  $x \leq y$.
    \item If $x= 0$, then keep $x = 0$. Randomly lift the value of $y$ if $y = 1$. In this case, $0 = x \leq y$.
    \item If $x = 1$, then $1 = y \geq x = 1$. Use the perfect coupling to lift both $x$ and $y$, so that $x = y$.
\end{enumerate}
The above coupling of $\pi^{\sigma_{V \setminus \{i\}}}_i$ and $\pi^{\tau_{V \setminus \{i\}}}_i$ ensures that $x \leq y$ with probability~1. Therefore, $\pi^{\sigma_{V \setminus \{i\}}}_i \sd \pi^{\tau_{V \setminus \{i\}}}_i$ for any $i \in V$ and any  $\sigma \preceq \tau$. 

Using the third definition in \Cref{prop:stochastically-monotone-equivalent}, it is straightforward to show that the Glauber dynamics $P_{\pi\textnormal{-GD}}$ is stochastically monotone.
\end{proof}

Next, we show that the chain $\psim$ is stochastically monotone. We need the following lemma, which says the monotone property is preserved by tilted distributions.
\begin{lemma}\label{lem:monotone-pi-sgd}
If $\mu$ over $\{0,1\}^V$ is a monotone system, then for any $\theta > 0$, the tilted distribution $(\theta * \mu)$ is a monotone system.
\end{lemma}
\begin{proof}
Fix two feasible configurations $\sigma$ and $\tau$ in $\Omega(\mu) = \Omega(\theta * \mu)$ such that $\sigma\preceq \tau$. We need to show that for any $i \in V$, $(\theta * \mu)^{\sigma_{V \setminus \{i\}}}_i(1) \leq (\theta * \mu)^{\tau_{V \setminus \{i\}}}_i(1)$. A simple calculation shows that
\begin{align}\label{eq:theta-mu-monotone}
    (\theta * \mu)^{\sigma_{V \setminus \{i\}}}_i(1) = \frac{\theta \mu^{\sigma_{V \setminus \{i\}}}_i(1)}{\theta \mu^{\sigma_{V \setminus \{i\}}}_i(1) + \mu^{\sigma_{V \setminus \{i\}}}_i(0)} \text{ and } (\theta * \mu)^{\tau_{V \setminus \{i\}}}_i(1) = \frac{\theta \mu^{\tau_{V \setminus \{i\}}}_i(1)}{\theta \mu^{\tau_{V \setminus \{i\}}}_i(1) + \mu^{\tau_{V \setminus \{i\}}}_i(0)}.
\end{align}
The lemma follows from $\mu^{\sigma_{V \setminus \{i\}}}_i(1) \leq \mu^{\tau_{V \setminus \{i\}}}_i(1)$ by the monotonicity of $\mu$.
\end{proof}

\begin{lemma}\label{lem:monotone-sgd}
If $\mu$ is a monotone system, then the chain $\psim$ is stochastically monotone.
\end{lemma}
\begin{proof}
We use the third definition in \Cref{prop:stochastically-monotone-equivalent} to prove the result. Fix two feasible and comparable configurations $\sigma$ and $\tau$ in $\Omega(\pi) \subseteq \{0,1,\star\}^V$ such that $\sigma\preceq \tau$. We need to show that $\psim(\sigma,\cdot) \sd \psim(\tau,\cdot)$. The Markov chain $\psim$ picks a vertex $i$ uniformly at random and then updates the value of $i$. We show that stochastic dominance holds conditional on any $i \in V$ being picked by the Markov chain. 
Since $\sigma\preceq \tau$, we have the following cases.

Case $\tau_i = \star$. In this case, $\psim$ must keep the value of $\tau_i = \star$. Since $\star$ is the largest value, it is straightforward to verify the stochastic dominance result.
 
Case $\tau_i \in \{0,1\}$. In this case, it must hold that $\sigma_i \in \{0,1\}$. The Markov chain $\psim$ resamples $\tau_i$ from the distribution $(\theta * \mu)^{\tau'}_i$, where $\tau'$ is a configuration over $V \setminus \{i\}$ such that $\tau' = \textsf{contr}(\tau)_{V \setminus \{i\}}$. Similarly, $\psim$ resamples $\sigma_i$ from the distribution $(\theta * \mu)^{\sigma'}_i$, where $\sigma'$ is a configuration over $V \setminus \{i\}$ such that $\sigma' = \textsf{contr}(\sigma)_{V \setminus \{i\}}$. Since $\sigma \preceq \tau$, by the definition of the contracting operation, we have $\sigma' \preceq \tau'$. We know that $(\theta * \mu)$ is a monotone system by \Cref{lem:monotone-pi-sgd}. Therefore, we have $(\theta * \mu)^{\sigma'}_i \sd(\theta * \mu)^{\tau'}_i$.
\end{proof}

\begin{lemma}\label{lem:monotone-pcl}
The chain $P_{\textnormal{cl}}$ is always stochastically monotone.
\end{lemma}
\begin{proof}
The chain $P_{\textnormal{cl}}$ can be viewed as a Markov chain defined over the whole space $\{0,1,\star\}^V$. We use the second definition in \Cref{prop:stochastically-monotone-equivalent} to prove the result. For two distributions $\nu$ and $\nu'$ over $\{0,1,\star\}^V$ such that $\nu \sd \nu'$, we need to show that $\nu P_{\textnormal{cl}}\sd\nu'P_{\textnormal{cl}}$. This result is a straightforward consequence of \Cref{lem:lift-and-contr-keep-dominate}.
\end{proof}

\subsection{Reversibility of Markov chains}\label{sec:proof-reversible}
\begin{lemma}\label{lem:reversible-pcl}
All Markov chains $P_{\pi\textnormal{-GD}}$, $P_{\textnormal{cl}}$ and $\psim$ are reversible with respect to $\pi$.
\end{lemma}
\begin{proof}
The reversibility of $P_{\pi\textnormal{-GD}}$ is well-known. We prove the reversibility of $P_{\textnormal{cl}}$ and $\psim$ with respect to $\pi$.
Note that starting from any feasible configuration, the two chains always stay in the feasible space $\Omega(\pi)$. We only need to prove the detailed balance equation for any pair of configurations in $\Omega(\pi)$.

Fix two configurations $\sigma$ and $\tau$ in $\Omega(\pi)$. For $P_{\textnormal{cl}}$, we need to show that 
\[\pi(\sigma) P_{\textnormal{cl}}(\sigma,\tau) = \pi(\tau) P_{\textnormal{cl}}(\tau,\sigma).\]
By the definition of the contracting and lifting operations, if $\{v \in V \mid \sigma_v = 0\} \neq \{v \in V \mid \tau_v = 0\}$, then $P_{\textnormal{cl}}(\sigma,\tau) = P_{\textnormal{cl}}(\tau,\sigma) = 0$. Assume the two sets are the same, which means $\textsf{contr}(\sigma) = \textsf{contr}(\tau)$. Recall that $C_1,C_\star$ in \eqref{def:c-sigma} count the number of $1$s and $\star$s in the configuration. We have $P_{\textnormal{cl}}(\sigma,\tau) = \theta^{C_1(\tau)}(1 - \theta)^{C_\star(\tau)}$ and $P_{\textnormal{cl}}(\tau,\sigma) = \theta^{C_1(\sigma)}(1 - \theta)^{C_\star(\sigma)}$. The reversibility is proved by noting that $\pi(\sigma) = \mu(\textsf{contr}(\sigma))\theta^{C_1(\sigma)}(1 - \theta)^{C_\star(\sigma)}$ and $\pi(\tau) = \mu(\textsf{contr}(\tau))\theta^{C_1(\tau)}(1 - \theta)^{C_\star(\tau)}$.

For $\psim$, we need to show the following detailed balance equation.
\[\pi(\sigma) \psim(\sigma,\tau) = \pi(\tau) \psim(\tau,\sigma).\]
If $\sigma = \tau$, then the equality holds trivially. Assume $\sigma \neq \tau$. By the definition of $\psim$, we only need to consider the case where there exists exactly one $i \in V$ such that $\sigma_i \neq \tau_i$ and the two values are $0$ or $1$. Otherwise, the transition probabilities on both sides are 0. Assume $\sigma_i = 0$ and $\tau_i = 1$. Let $\rho = \sigma_{V \setminus \{i\}} = \tau_{V \setminus \{i\}}$. The above equality is equivalent to 
\[\pi^\rho_i(0) \psim(\sigma,\tau) = \pi^\rho_i(1) \psim(\tau,\sigma).\]
Using~\eqref{eq:pi-0-i} and~\eqref{eq:pi-1-i}, we have
\begin{align}\label{eq:pi-rho-i}
    \pi^{\rho}_i(0) = \frac{\mu(\textsf{contr}(\sigma))}{\mu(\textsf{contr}(\sigma))+\mu(\textsf{contr}(\tau))}, \quad  \pi^{\rho}_i(1) = \frac{\theta \mu(\textsf{contr}(\tau))}{\mu(\textsf{contr}(\sigma))+\mu(\textsf{contr}(\tau))},
\end{align}
By the definition of the Markov chain $\psim$, using~\eqref{eq:theta-mu-monotone} and a simple normalization, we have 
\begin{align}\label{eq:psim-rho}
    \begin{split}
        \psim(\sigma,\tau) &= \frac{1}{n}(\theta * \mu)^{\textsf{contr}(\sigma)_{V \setminus \{i\}}}_{i}(1) =\frac{1}{n}\frac{\theta \mu(\textsf{contr}(\tau))}{\theta \mu(\textsf{contr}(\tau))+ \mu(\textsf{contr}(\sigma))},\\
        \psim(\tau,\sigma) &= \frac{1}{n}(\theta * \mu)^{\textsf{contr}(\tau)_{V \setminus \{i\}}}_{i}(0) =\frac{1}{n}\frac{\mu(\textsf{contr}(\sigma))}{\theta \mu(\textsf{contr}(\tau))+ \mu(\textsf{contr}(\sigma))}.
    \end{split}
\end{align}
Combining~\eqref{eq:pi-rho-i} and~\eqref{eq:psim-rho} proves the detailed balance equation.
\end{proof}

\section{Proof of stochastic dominance}\label{sec:proof-stochastic-dominance}
%In \Cref{sec:proof-comparison-inequality-relation}, we proved \Cref{lem:many-stationary} and \Cref{lem:properties-of-transition-matrices}. 
In this section, we assume that  \Cref{lem:phase-comparison} holds and prove the stochastic dominance result in \Cref{lem:stochastic-dominance}. The proof of \Cref{lem:phase-comparison} is deferred to \Cref{sec:proof-dominance}. 

Recall that $\pi^{(0)}_{\textnormal{GD}}$ is the distribution of $\textsf{lift}(\*1_V)$. The following lemma gives the property that $\pi^{(0)}_{\textnormal{GD}}/\pi$ is increasing, which will be used in the proof of \Cref{lem:stochastic-dominance}.

\begin{lemma}\label{lem:pi-0-gd}
The function $\frac{\pi^{(0)}_{\textnormal{GD}}}{\pi}$ is an increasing function.
\end{lemma}
\begin{proof}
    Recall that the two functions $C_1$ and $C_\star$ are defined in~\eqref{def:c-sigma}. 
    \begin{align*}
     \forall \sigma\in \{0,1,\star\}^V,\quad   \pi^{(0)}_{\textnormal{GD}}(\sigma) = \begin{cases}
            0, & \text{if } \exists i\in V, \sigma_i=0;\\
            \theta^{C_1(\sigma)}(1-\theta)^{C_\star(\sigma)}, & \text{otherwise}.
        \end{cases}  
    \end{align*}
    Consider two configurations $\sigma$ and $\tau$ such that $\sigma\preceq \tau$. We show that $\frac{\pi^{(0)}_{\textnormal{GD}}(\sigma)}{\pi(\sigma)}\leq \frac{\pi^{(0)}_{\textnormal{GD}}(\tau)}{\pi(\tau)}$. By the definition of $\preceq$, $\sigma_i\leq\tau_i$ for all $i\in V$. If there exists an $i\in V$ such that $\sigma_i=0$,
    \begin{align*}
    \frac{\pi^{(0)}_{\textnormal{GD}}(\sigma)}{\pi(\sigma)}=0\leq \frac{\pi^{(0)}_{\textnormal{GD}}(\tau)}{\pi(\tau)}.
    \end{align*}
    We only need to consider the case where $1\leq \sigma_i\leq \tau_i$ for all $i\in V$.
    By the definition of $\pi$, for these configurations, $\pi(\sigma)=\mu(\*1_V)\theta^{C_1(\sigma)}(1-\theta)^{C_\star(\sigma)}$. We have
    \begin{align*}
    \frac{\pi^{(0)}_{\textnormal{GD}}(\sigma)}{\pi(\sigma)}=\frac{\theta^{C_1(\sigma)}(1-\theta)^{C_\star(\sigma)}}{\mu(\*1_V)\theta^{C_1(\sigma)}(1-\theta)^{C_\star(\sigma)}}=\frac{1}{\mu(\*1_V)}=\frac{\pi^{(0)}_{\textnormal{GD}}(\tau)}{\pi(\tau)}.
    \end{align*}
    Therefore, $\frac{\pi^{(0)}_{\textnormal{GD}}}{\pi}$ is an increasing function. 
\end{proof}

Now, we are ready to prove \Cref{lem:stochastic-dominance}.

\begin{proof}[Proof of \Cref{lem:stochastic-dominance} assuming \Cref{lem:phase-comparison}]
By \Cref{lem:pi-0-gd},
$\pi^{(0)}_{\textnormal{GD}}/\pi$ is increasing.
Since $P_{\pi\textnormal{-GD}}\in\+{MC}_\pi$,
\Cref{lem:reversible-transition} gives
\[
\frac{\pi^{(t+1)}_{\textnormal{GD}}}{\pi}
=
P_{\pi\textnormal{-GD}}
\left(\frac{\pi^{(t)}_{\textnormal{GD}}}{\pi}\right).
\]
The stochastic monotonicity of $P_{\pi\textnormal{-GD}}$
therefore implies that
$\pi^{(t)}_{\textnormal{GD}}/\pi$ is increasing for every $t\geq 0$.
Furthermore, each $P_{\textnormal{alg}}^{(t)}$ is either
$\psim$ or $P_{\textnormal{cl}}\psim$,
and is thus stochastically monotone by
\Cref{lem:properties-of-transition-matrices}.

We now prove
$\pi^{(t)}_{\textnormal{GD}}\sd\pi^{(t)}_{\textnormal{alg}}$
by induction on $t$.
The case $t=0$ follows from equality of the initial distributions.
Suppose the statement holds at time $t$.
Applying \Cref{lem:phase-comparison} with
$\nu=\pi^{(t)}_{\textnormal{GD}}$, followed by the stochastic
monotonicity of $P_{\textnormal{alg}}^{(t+1)}$, gives
\begin{align*}
\pi^{(t+1)}_{\textnormal{GD}}
&=
\pi^{(t)}_{\textnormal{GD}}
P_{\pi\textnormal{-mGD}}^{(t+1)}
\\
&\sd
\pi^{(t)}_{\textnormal{GD}}
P_{\textnormal{alg}}^{(t+1)}
\\
&\sd
\pi^{(t)}_{\textnormal{alg}}
P_{\textnormal{alg}}^{(t+1)}
=
\pi^{(t+1)}_{\textnormal{alg}}.
\end{align*}
This completes the induction.
\end{proof}

\section{Proof of the comparison inequality relation}\label{sec:proof-dominance}

In this section, we prove \Cref{lem:phase-comparison}. Fix a time $t$. There are two cases: (1) $P^{(t)}_{\textnormal{alg}} = \psim$ and $P^{(t)}_{\pi\textnormal{-mGD}} = P_{\pi\textnormal{-GD}}$; (2) $P^{(t)}_{\textnormal{alg}} = P_{\textnormal{cl}}\psim$ and $P^{(t)}_{\pi\textnormal{-mGD}} = P_{\textnormal{cl}}P_{\pi\textnormal{-GD}}$. Since $P_{\textnormal{cl}}\mc P_{\textnormal{cl}}$, the second part of \Cref{lem:comparison-inequality-properties} shows that both cases follow from $P_{\pi\textnormal{-GD}}\mc\psim$. For any vertex $i \in V$, let $\psim^i$ denote the transition in $\psim$ conditional on selecting vertex $i \in V$, and let $P^i_{\pi\textnormal{-GD}}$ denote the transition in $P_{\pi\textnormal{-GD}}$ conditional on selecting vertex $i \in V$. It holds that 
\begin{align*}
   \psim = \frac{1}{n}\sum_{i\in V} \psim^i, \quad P_{\pi\textnormal{-GD}} = \frac{1}{n}\sum_{i\in V} P^i_{\pi\textnormal{-GD}}.
\end{align*}
For each $i\in V$, the same monotonicity and detailed-balance arguments in
\Cref{sec:proof-stochastically-monotone,sec:proof-reversible},
with the updated vertex fixed to $i$, show that
$P^i_{\pi\textnormal{-GD}},\psim^i\in\+{MC}_\pi$.
Using the first property of $\mc$ in \Cref{lem:comparison-inequality-properties}, to prove that $P_{\pi\textnormal{-GD}} \mc \psim$, it suffices to show that for any vertex $i \in V$, $P^i_{\pi\textnormal{-GD}} \mc \psim^i$. Hence, the proof of \Cref{lem:phase-comparison} can be reduced to proving the following lemma.
\begin{lemma}\label{lem:phase-comparison-single-vertex}
For any vertex $i \in V$, $P^i_{\pi\textnormal{-GD}} \mc \psim^i$.
\end{lemma}
\begin{proof}
Fix $i\in V$.
Let $\nu$ be a distribution on $\Omega(\pi)$ with increasing density $h=\nu/\pi$.
Write $X_{-i}=X_{V\setminus\{i\}}$.
Fix an exterior configuration $\tau$ with $\nu(X_{-i}=\tau)>0$, and write
$(\tau,a)$ for the configuration with exterior $\tau$ and spin $a$ at $i$.

The Glauber update $P^i_{\pi\textnormal{-GD}}$ draws from $\pi_i^\tau$.
The simulation update $\psim^i$ keeps $\star$ and redraws a nonstar spin from
$(\theta*\mu)_i^{\textsf{contr}(\tau)}$.
By \eqref{eq:pi-0-i}, \eqref{eq:pi-1-i}, and \eqref{eq:theta-mu-monotone},
$\pi_i^\tau(\star)\le 1-\theta<1$ and
\[
(\theta*\mu)_i^{\textsf{contr}(\tau)}(a)
=\frac{\pi_i^\tau(a)}{1-\pi_i^\tau(\star)},\qquad a\in\{0,1\}.
\]

Fix an increasing function $f$ on $\Omega(\pi)$.
If $\pi_i^\tau(\star)=0$, then $\nu_i^\tau(\star)=0$ since $\nu\ll\pi$,
so the two updates give the same conditional expectation of $f$.
If $\pi_i^\tau(\star)>0$, substituting $\nu=h\pi$ gives
\[
\nu_i^\tau(\star)
 =\frac{\pi(\tau,\star)h(\tau,\star)}
       {\sum_{a:\,\pi(\tau,a)>0}\pi(\tau,a)h(\tau,a)}
 =\pi_i^\tau(\star)\,\frac{h(\tau,\star)}{\mathbb E_\pi[h(X)\mid X_{-i}=\tau]}
 \ge \pi_i^\tau(\star).
\]
The conditional mean is positive since $\nu(X_{-i}=\tau)>0$, and it is at most
$h(\tau,\star)$ since $h$ is increasing and $\star$ is the largest spin.

Conditioning on whether the output is $\star$ gives
\begin{align*}
\mathbb E_{\nu \psim^i}[f(X)\mid X_{-i}=\tau]
 &=\nu_i^\tau(\star)f(\tau,\star)
 +(1-\nu_i^\tau(\star))\,\mathbb E_\pi[f(X)\mid X_{-i}=\tau,\,X_i\ne\star],\\
\mathbb E_{\nu P^i_{\pi\textnormal{-GD}}}[f(X)\mid X_{-i}=\tau]
 &=\pi_i^\tau(\star)f(\tau,\star)
 +(1-\pi_i^\tau(\star))\,\mathbb E_\pi[f(X)\mid X_{-i}=\tau,\,X_i\ne\star].
\end{align*}
Subtracting the second expression from the first and collecting terms gives
\[
\bigl(\nu_i^\tau(\star)-\pi_i^\tau(\star)\bigr)
\Bigl(f(\tau,\star)-\mathbb E_\pi[f(X)\mid X_{-i}=\tau,\,X_i\ne\star]\Bigr)\ge0,
\]
since $f(\tau,\star)\ge f(\tau,a)$ for every feasible $a\in\{0,1\}$.
Both updates preserve $X_{-i}$, so their common exterior marginal is $\nu_{-i}$.
Summing the conditional inequality with respect to this marginal gives
\begin{align*}
\mathbb E_{\nu P^i_{\pi\textnormal{-GD}}}[f(X)]
&=\sum_{\tau:\,\nu_{-i}(\tau)>0}\nu_{-i}(\tau)\,
  \mathbb E_{\nu P^i_{\pi\textnormal{-GD}}}[f(X)\mid X_{-i}=\tau]\\
&\le\sum_{\tau:\,\nu_{-i}(\tau)>0}\nu_{-i}(\tau)\,
  \mathbb E_{\nu \psim^i}[f(X)\mid X_{-i}=\tau]
 =\mathbb E_{\nu \psim^i}[f(X)].
\end{align*}
Since this holds for every increasing $f$, we conclude that $\nu P^i_{\pi\textnormal{-GD}}\sd\nu \psim^i$.
\end{proof}
\begin{remark}[Comparison from the minimum state]
\label{rem:minimum-start-comparison}
Assume $\mu(\*0_V)>0$.
The comparison in \Cref{lem:comparison} also holds when Glauber
dynamics starts from $\*0_V$ and the simulation algorithm is
initialized with $\textsf{lift}(\*0_V)$.
In this remark, all distributions of these processes refer to
these initializations.

The lifting identity in \Cref{lem:comparison-lifted} and the
invariance under $P_{\textnormal{cl}}$ in the proof of
\Cref{lem:many-stationary} hold for any initial law on $\Omega(\mu)$,
provided that the lifted chain starts from the lift of that law.
Moreover, by reversibility and stochastic monotonicity,
$\pi^{(t)}_{\textnormal{GD}}/\pi$ and
$\mu^{(t)}_{\textnormal{GD}}/\mu$ remain non-increasing,
since their initial densities are non-increasing.
For any distribution $\nu$ with non-increasing $\nu/\pi$, let
\[
\widetilde\nu=\frac{c\pi-\nu}{c-1},
\qquad c>\max_{\sigma\in\Omega(\pi)}\frac{\nu(\sigma)}{\pi(\sigma)}.
\]
Then $\widetilde\nu/\pi$ is increasing.
Applying \Cref{lem:phase-comparison} to $\widetilde\nu$ and using
that both transition matrices preserve $\pi$ gives
\[
\nu P_{\pi\textnormal{-mGD}}^{(s)}
\succeq_{\textnormal{sd}}
\nu P_{\textnormal{alg}}^{(s)}.
\]
The same induction as in the proof of \Cref{lem:stochastic-dominance}
therefore yields
$\pi^{(t)}_{\textnormal{GD}}\succeq_{\textnormal{sd}}
\pi^{(t)}_{\textnormal{alg}}$ for every integer $0\leq t\leq T_1T_2$.
After contraction, this gives
\[
\mu^{(t)}_{\textnormal{GD}}\succeq_{\textnormal{sd}}
\textsf{contr}(\pi^{(t)}_{\textnormal{alg}}).
\]
The set $A=\{\sigma\in\Omega(\mu):
\mu^{(t)}_{\textnormal{GD}}(\sigma)>\mu(\sigma)\}$ is decreasing
because $\mu^{(t)}_{\textnormal{GD}}/\mu$ is non-increasing. Hence
\begin{align*}
\DTV{\mu^{(t)}_{\textnormal{GD}}}{\mu}
&=\mu^{(t)}_{\textnormal{GD}}(A)-\mu(A)\\
&\leq \textsf{contr}(\pi^{(t)}_{\textnormal{alg}})(A)-\mu(A)\\
&\leq \DTV{\textsf{contr}(\pi^{(t)}_{\textnormal{alg}})}{\mu}.
\end{align*}
\end{remark}

\section{Proofs of general results and applications}\label{sec:proofs-of-general-results-and-applications}

In this section, we first use the comparison result in \Cref{thm:FD-comparison} to prove the general result in \Cref{thm:main-comparison}. For applications, the main task is to verify the entropic independence condition and then apply our general result. The entropic independence results for the random cluster model and the bipartite hardcore model are given in \cite{CZ23} and \cite{Chen0Y23}. However, if we directly apply the results, we will get $m (\log n)^{O(1/\delta^2)}$ and $n^2 (\log n)^{O(1/\delta)}$ mixing times for the random cluster model and the bipartite hardcore model, respectively, where $\max_{v \in V}\lambda_v \leq 1 - \delta$ for the random cluster model and $\lambda \leq (1-\delta)\lambda_c(\Delta_L)$ for the bipartite hardcore model. We give a more refined analysis of the entropic independence and obtain the improved $m (\log n)^{O(1)}$ and $n^2 (\log n)^{O(1)}$ mixing times for the random cluster model and the bipartite hardcore model, respectively, where the exponent of $\log n$ is a universal constant.

\subsection{Mixing from entropic independence}\label{sec:proof-GR}

The following theorem was introduced in \cite{ChenE22}. Also see the theorem in \cite{Chen0Y23}.
\begin{theorem}[\text{\cite{ChenE22} and \cite{Chen0Y23}}]\label{thm:SIandMS-KL-Decay}
    Let $\theta\in (0,1)$. Let $\alpha: [0,-\log \theta]\to \mathbb{R}_{> 0}$ be an integrable function. Let $\mu$ be a distribution over $\{0,1\}^V$. 
    Suppose for any $t \in [0,-\log \theta]$, for any $\Lambda \subseteq V$ such that $\mu_{\Lambda}(\*1_\Lambda) > 0$, the distribution $(e^{-t} * \mu)^{\*1_\Lambda}$ is $\alpha(t)$-entropically independent.
    If $\mu$ is a point mass, the mixing time is zero. Otherwise, for $0<\epsilon\leq1/2$, the mixing time of field dynamics $\mixfd(\mu,\theta,\epsilon)$ satisfies
    \begin{align*}
        \mixfd(\mu,\theta,\epsilon)\leq \kappa^{-1}\left(\log \log \frac{1}{\mu_{\min}} +\log \frac{1}{2\epsilon^2} \right)+1,
    \end{align*}
    where $\kappa:=\exp({-\int_0^{-\log \theta}4\alpha(t)\d t})$ and $\mu_{\min}=\min_{\sigma\in \{0,1\}^V:\mu(\sigma)>0}\mu(\sigma)$.
    \end{theorem}

    \Cref{thm:SIandMS-KL-Decay} is almost proved in \cite{ChenE22} and \cite{Chen0Y23}. The only difference is that they use a lower bound $\kappa'$ of $\kappa$, where $\kappa' = \exp(\log \theta \cdot \sup_{t \in [0,-\log \theta]} 4\alpha(t))$. \Cref{thm:SIandMS-KL-Decay} can be proved by a simple modification of the proof in \cite{Chen0Y23}. In Appendix~\ref{sec:mixing-of-field-dynamics-from-entropic-independence}, we recall some key definitions and proof steps in \cite{Chen0Y23} and show how to do the modification.

    \Cref{thm:main-comparison} is a straightforward corollary of \Cref{thm:FD-comparison} and \Cref{thm:SIandMS-KL-Decay}.
\begin{corollary}[Comparison from an extremal state]
\label{cor:extremal-main-comparison}
Suppose that $\mu,\theta,\alpha$ satisfy the monotonicity and entropic-independence hypotheses of \Cref{thm:main-comparison}.
Let $x\in\{\*0_V,\*1_V\}$ satisfy $\mu(x)>0$.
Let $\kappa$ be as in \Cref{thm:SIandMS-KL-Decay}.
For $0<\epsilon<1$, set $H_x=\log(1/\mu(x))$ and
\[
T_x=\left\lceil
\kappa^{-1}\log\frac{2(1+H_x)}{\epsilon^2}
\right\rceil.
\]
Then
\[
T^{\textnormal{GD}}_{\textnormal{mix-}x}(\mu,\epsilon)
\leq
T_x\cdot
\mixtl\left(\mu,\theta,\frac{\epsilon}{2T_x}\right).
\]
\end{corollary}

\begin{proof}
The conclusion is trivial if $\mu$ is a point mass or the
right-hand side is infinite. Otherwise, the KL contraction in
Appendix~\ref{sec:mixing-of-field-dynamics-from-entropic-independence}
and Pinsker's inequality give
\[
d_{\mathrm{TV}}\left(P_{\theta,\mu}^{T_x}(x,\cdot),\mu\right)
\leq\sqrt{\tfrac12 H_xe^{-\kappa T_x}}
\leq\epsilon/2.
\]

Initialize the simulation algorithm with $\textsf{lift}(x)$
and take
\[
T_1=T_x,\qquad
T_2=\mixtl\left(\mu,\theta,\frac{\epsilon}{2T_x}\right).
\]
The same coupling argument as in the proof of
\Cref{lem:simu-error} bounds the accumulated simulation error
by $T_x\cdot\epsilon/(2T_x)=\epsilon/2$.
Thus the contracted simulation output is within TV distance
$\epsilon$ of $\mu$.
The result follows from \Cref{lem:comparison} when $x=\*1_V$
and from \Cref{rem:minimum-start-comparison} when $x=\*0_V$.
\end{proof}

\subsection{Random cluster model}
In this section, we apply our general result to the random cluster model $\mu_{p,\lambda}^{\textnormal{RC}}$ in~\eqref{eq:weight-rc} and prove \Cref{thm:RC-model}. The random cluster model is defined as a distribution over subsets of edges $S \subseteq E$ of a graph $G=(V,E)$. We can map every subset $S$ to a configuration $\sigma \in \{0,1\}^E$ by setting $\sigma_e = 1$ if $e \in S$ and $\sigma_e = 0$ otherwise. Then $\mu_{p,\lambda}^{\textnormal{RC}}$ can be viewed as a distribution over $\{0,1\}^E$. To apply \Cref{cor:extremal-main-comparison}, we consider the following flipped distribution of $\mu_{p,\lambda}^{\textnormal{RC}}$.

\begin{definition}[Flipped distribution]\label{def:flipped-distribution}
    Let $V$ be a set of variables and $\mu$ be a distribution over $\{0,1\}^V$.
    For any configuration $\sigma\in \{0,1\}^V$, let $\bar{\sigma}_i=1-\sigma_i$ for all $i\in V$. The \emph{flipped distribution} $\bar{\mu}$ of $\mu$ is defined by $\forall \sigma\in \{0,1\}^V, \bar{\mu}(\sigma) = \mu(\bar{\sigma})$.
\end{definition}

Coordinatewise flipping maps Glauber dynamics for
$\mu_{p,\lambda}^{\textnormal{RC}}$ started from $E$
to Glauber dynamics for $\bar\mu_{p,\lambda}^{\textnormal{RC}}$
started from $\*0_E$.
To apply \Cref{cor:extremal-main-comparison} with this initial state,
we will show that $\bar{\mu}_{p,\lambda}^{\textnormal{RC}}$
is a monotone system and that there exist $\theta\in(0,1)$
and an integrable function $\alpha(t)$ such that
\begin{itemize}
    \item $\mixtl(\bar{\mu}_{p,\lambda}^{\textnormal{RC}},\theta,\epsilon) = \tilde{O}(m)$, where $m = |E|$ and $\mixtl(\bar{\mu}_{p,\lambda}^{\textnormal{RC}},\theta,\epsilon)$ is defined in \eqref{def:mixtl}.
    \item $(\e^{-t} *\bar{\mu}_{p,\lambda}^{\textnormal{RC}})^{\*1_\Lambda}$ is $\alpha(t)$-entropically independent for every feasible pinning $\*1_\Lambda$ for $t\in [0,\ln \frac{1}{\theta}]$;
\end{itemize}
%where $\bar{\mu}_{p,\lambda}^{\textnormal{RC}}$ is the flipped distribution of $\mu_{p,\lambda}^{\textnormal{RC}}$ defined in \Cref{def:flipped-distribution} later. 

We start with the proof of the monotonicity of $\bar{\mu}_{p,\lambda}^{\textnormal{RC}}$. It is well-known that the random cluster model $\mu_{p,\lambda}^{\textnormal{RC}}$ is a monotone system (e.g. see~\cite[Lem. 8.2]{feng2023swendsen}). We use the following lemma to prove the monotonicity of the flipped distribution $\bar{\mu}_{p,\lambda}^{\textnormal{RC}}$.
\begin{lemma}\label{lem:monotone-flipped}
    If $\mu$ over $\{0,1\}^V$ is a monotone system, then $\bar{\mu}$ is also a monotone system.    
    \end{lemma}
    \begin{proof}
    It suffices to show that, for every $v\in V$ and every pair of feasible and comparable pinnings $\sigma,\tau \in \{0,1\}^{V\setminus \{v\}}$ with $\sigma\preceq \tau$, we have $\bar{\mu}_v^{\sigma}(1)\leq \bar{\mu}_v^{\tau}(1)$.
    Note that $\bar{\tau}\preceq \bar{\sigma}$. By our assumption, $\mu_v^{\bar{\tau}}(1)\leq \mu_v^{\bar{\sigma}}(1)$ and then $\mu_v^{\bar{\tau}}(0)\geq \mu_v^{\bar{\sigma}}(0)$. Then $\bar{\mu}_v^{\sigma}(1) = \mu_v^{\bar{\sigma}}(0)\leq \mu_v^{\bar{\tau}}(0) = \bar{\mu}_v^{\tau}(1)$.
    \end{proof}

The cases $m\leq1$ are immediate, so assume $m\geq2$. We choose the parameter $\theta$ and the function $\alpha(t)$ as follows:
\begin{align}\label{eq:def-the}
    \theta=\frac{p_{\min}\min \{10^{-7},\frac{1-\lambda_{\max}}{27}\} }{\log n},
 \end{align}
 where $n$ is the number of vertices in $G$, $\lambda_{\max}=\max_{i\in V}\lambda_i$ and $p_{\min}=\min_{i\in E}p_i$, and 
 \begin{align}\label{eq:alpha-t}
     \alpha(t)=\begin{cases}
         3(1-\lambda_{\max})^{-2} & \textnormal{if } 0\leq t \leq \log\frac{1}{\theta_0},\\
         5*10^4 & \textnormal{if } \log\frac{1}{\theta_0} < t\leq \log\frac{1}{\theta},
     \end{cases} \quad \text{where } \theta_0 = \frac{1}{2}p_{\min}(1-\lambda_{\max})^{2}.
 \end{align}
Note that the function $\alpha$ has two parts and only the first part depends on $\lambda_{\max}$. As we can see from the later proof, this is the main reason why we can make the degree of the  $\mathrm{polylog}(n)$ factor in the mixing time independent of $\delta_\lambda$ in \Cref{thm:RC-model}.

We analyze the mixing time of tilted distributions $\mixtl(\bar{\mu}_{p,\lambda}^{\textnormal{RC}},\theta,\epsilon)$. By the definition in~\eqref{def:mixtl}, we need to consider $\pi = (\theta * \bar{\mu}_{p,\lambda}^{\textnormal{RC}})^{\*1_\Lambda}$ with pinning $\*1_\Lambda$ on an arbitrary subset $\Lambda\subseteq E$. Since the Glauber dynamics on $\pi$ and $\bar{\pi}$ are essentially the same, we can analyze the mixing time by considering $\bar{\pi}$.
Note that $\bar{\pi}$ is a random cluster model $(\mu_{p',\lambda}^{\textnormal{RC}})^{\*0_{\Lambda}}$ such that the values in $\Lambda$ are fixed to be 0 and the parameters $p'_e$ on every edge satisfy $\frac{p'_e}{1-p'_e} = \frac{1}{\theta} \cdot \frac{p_e}{1-p_e}$.
Hence, $\bar{\pi}$ fixes all values in $\Lambda$ to 0. In the subset $E \setminus \Lambda$, it is another random cluster model on the subgraph $G' = (V, E \setminus \Lambda)$ with parameters $p'_e$ on every edge.
By our choice of the parameter $\theta$ in~\eqref{eq:def-the}, all parameters $p_e'$ are large enough. For $m\geq2$ and $0<\epsilon\leq1/2$, using \cite[Lem. 3.8]{CZ23} and accounting for updates at pinned edges, the mixing time of Glauber dynamics for $\bar{\pi}$ starting from an arbitrary initial configuration is $O(m \log m \log \frac{1}{\epsilon})$. The cases $|E\setminus\Lambda|\leq1$ follow directly. Since the same mixing time bound applies to $\pi$, we have
 \begin{align}\label{eq:mixtl-rc}
    \mixtl(\bar{\mu}_{p,\lambda}^{\textnormal{RC}},\theta,\epsilon) = O\left(m \log m \log \frac{1}{\epsilon}\right).
\end{align}

Now, we only need to show the following entropic independence result.
\begin{lemma}\label{lem:RC-EI}
$\forall\, t\in [0,\log \frac{1}{\theta}]$, for any $\Lambda \subseteq E$,  $(\e^{-t} *\bar{\mu}_{p,\lambda}^{\textnormal{RC}})^{\*1_\Lambda}$ is $\alpha(t)$-entropically independent for  $\alpha$ defined in \eqref{eq:alpha-t}.
\end{lemma}

\Cref{lem:RC-EI} improves the entropic independence result in \cite{CZ23}, which is achieved by an improved analysis of the coupling between the random cluster model and the subgraph world model. The detailed proof is deferred to Appendix \ref{sec:RC-EI-proof}.

Let us assume \Cref{lem:RC-EI} holds and prove the result for the random cluster model.

\begin{proof}[Proof of \Cref{thm:RC-model} assuming \Cref{lem:RC-EI}]
We apply \Cref{cor:extremal-main-comparison} to $\bar{\mu}_{p,\lambda}^{\textnormal{RC}}$
with initial state $\*0_E$.
Its monotonicity has been verified above, and the entropic
independence hypothesis follows from \Cref{lem:RC-EI}.
Flipping maps $\*0_E$ to $E$, so
\[
H_{\*0_E}=\log\frac{1}{\bar{\mu}_{p,\lambda}^{\textnormal{RC}}(\*0_E)}
=\log\frac{1}{\mu_{p,\lambda}^{\textnormal{RC}}(E)}.
\]

To bound $H_{\*0_E}$, we follow \cite[Appendix A.2]{CZ23}
and let $B=\bigotimes_{e\in E}\mathrm{Ber}(p_e)$.
The unnormalized density of $\mu_{p,\lambda}^{\textnormal{RC}}$ with respect to $B$
is the cluster factor in~\eqref{eq:weight-rc},
which lies in $[1,2^n]$. Hence
\[
\mu_{p,\lambda}^{\textnormal{RC}}(E)\geq2^{-n}B(E)\geq2^{-n}p_{\min}^m,
\qquad H_{\*0_E}\leq n\log2+m\log\frac1{p_{\min}}.
\]

By~\eqref{eq:alpha-t}, regardless of the order of $\theta$ and $\theta_0$,
\begin{align*}
\kappa^{-1}
&\leq
\exp\left(
12(1-\lambda_{\max})^{-2}\log\frac1{\theta_0}
+2\cdot10^5\log\frac1\theta
\right)\\
&=
\theta_0^{-12(1-\lambda_{\max})^{-2}}
\theta^{-2\cdot10^5}\\
&\leq C_1(\delta_p,\delta_\lambda)
(\log n)^{2\cdot10^5},
\end{align*}
where the last inequality follows from the choices of
$\theta,\theta_0$, together with
$p_{\min}\geq\delta_p$ and
$\lambda_{\max}\leq1-\delta_\lambda$.

Combining the bounds on $H_{\*0_E}$ and $\kappa^{-1}$, and using
$m\leq n^2$ and $0<\epsilon\leq1/2$, gives
\[
T_{\*0_E}\leq
C_2(\delta_p,\delta_\lambda)
(\log n)^{O(1)}\log\frac1\epsilon.
\]
By \Cref{cor:extremal-main-comparison} and
\eqref{eq:mixtl-rc},
\begin{align*}
T^{\textnormal{GD}}_{\textnormal{mix-}E}(\mu_{p,\lambda}^{\textnormal{RC}},\epsilon)
&=T^{\textnormal{GD}}_{\textnormal{mix-}\*0_E}(\bar{\mu}_{p,\lambda}^{\textnormal{RC}},\epsilon)\\
&\leq
T_{\*0_E}\cdot\mixtl\left(
\bar{\mu}_{p,\lambda}^{\textnormal{RC}},\theta,\frac{\epsilon}{2T_{\*0_E}}
\right)\\
&\leq O\left(
T_{\*0_E}\,m\log m\log\frac{2T_{\*0_E}}{\epsilon}
\right)\\
&\leq
C(\delta_p,\delta_\lambda)\cdot m
(\log n)^{O(1)}
\left(\log\frac1\epsilon\right)^2.
\end{align*}
\end{proof}

%\paragraph{Proof of entropic independence}

\subsection{Bipartite hardcore model}
We now prove the mixing time of the Glauber dynamics on the bipartite hardcore model in \Cref{thm:bhc-mixing-time}.
We need a more general notation in the analysis.
Let $G=((V_L,V_R),E)$ be an undirected bipartite graph.
Let $n=\abs{V_L\cup V_R}$ be the total number of vertices and let $\Delta=\Delta_L\geq3$ bound the degrees in $V_L$; the actual maximum degree may be below $3$.
Let $\lambda > 0$ and $\beta > 0$ be the external fields on $V_L$ and $V_R$ respectively.
Define a bipartite hardcore model with distribution $\mu_{G,\lambda,\beta}^{\textnormal{HC}}$ specified by $(G,\lambda,\beta)$ such that for any independent set $S \subseteq V_L \cup V_R$, $\mu_{G,\lambda,\beta}^{\textnormal{HC}}(S) \propto \lambda^{|S \cap V_L|}\beta^{|S \cap V_R|}$.
In the analysis, we will view $\mu_{G,\lambda,\beta}^{\textnormal{HC}}$ as a distribution over Boolean vectors in $\{0,1\}^{V_L\cup V_R}$. For any subset $S \subseteq V_L \cup V_R$, we uniquely map $S$ to a Boolean vector $\sigma^S$ such that
\begin{align}\label{eq:bhc-mapping}
    \begin{cases}
        \sigma^S_v = 1 - \*1[v \in S] \quad &\textnormal{if } v\in V_L;\\
    \sigma^S_v = \*1[v \in S] \quad &\textnormal{if } v\in V_R.
    \end{cases}
\end{align}
Note that for a vertex $v \in V_L$, $\sigma^S_v = 1$ indicates that $v$ is \emph{not} in the independent set $S$. 
Now, $\mu_{G,\lambda,\beta}^{\textnormal{HC}}$ can be viewed as a distribution such that $\mu_{G,\lambda,\beta}^{\textnormal{HC}}(\sigma^S) = \mu_{G,\lambda,\beta}^{\textnormal{HC}}(S)$.
As we will see later, such a definition guarantees the monotonicity property.
Furthermore, let $\mu_{L,\lambda,\beta}^{\textnormal{HC}}$ be the marginal distribution on $V_L$ projected from $\mu_{G,\lambda,\beta}^{\textnormal{HC}}$.
In other words, $\sigma_{V_L} \sim \mu_{L,\lambda,\beta}^{\textnormal{HC}}$ if $\sigma \sim \mu_{G,\lambda,\beta}^{\textnormal{HC}}$.

\Cref{thm:bhc-mixing-time} considers a special case when $\lambda=\beta$. 
The following theorem gives a mixing time bound for Glauber dynamics on the marginal distribution on the left side $V_L$ when $\lambda = \beta$.
\begin{theorem}\label{thm:bhc-mixing-time-marginal}
     Let $\delta > 0$ be a constant. If $\lambda\leq (1-10\delta)\lambda_c(\Delta)$, then the mixing time of Glauber dynamics on the marginal distribution $\mu_{L,\lambda,\lambda}^{\textnormal{HC}}$ starting from $\*1_{V_L}$ (the empty set on $V_L$) satisfies, for every $0<\epsilon\leq1/2$,
    \begin{align*}
    T_{\textnormal{mix-}\*1_{V_L}}(\mu_{L,\lambda,\lambda}^{\textnormal{HC}},\epsilon)\leq n \cdot \left(\frac{C}{\lambda}\right)^{O(1/\delta)} \cdot ({\Delta\log n})^{O(1)}\cdot \log^2\frac{1}{\epsilon},
    \end{align*}
    with the same absolute constant $C$ as in~\Cref{thm:bhc-mixing-time} and absolute constants implicit in $O(\cdot)$.
    %where $T_{\textnormal{mix}}(\mu_{V_L,\lambda,\beta}^{\textnormal{HC}},\epsilon)$ is the mixing time of Glauber dynamics on $\mu_{V_L,\lambda,\beta}^{\textnormal{HC}}$ starting from all-$1$ state.
    \end{theorem}

We first prove \Cref{thm:bhc-mixing-time-marginal} and then we use \Cref{thm:bhc-mixing-time-marginal} to prove \Cref{thm:bhc-mixing-time}.

\paragraph{Mixing time for the marginal distribution}
\Cref{thm:bhc-mixing-time-marginal} is proved by the general result in \Cref{thm:main-comparison}. As in the analysis of the random cluster model, we need to show that $\mu_{L,\lambda,\lambda}^{\textnormal{HC}}$ is a monotone system and choose parameters $\theta$ and $\alpha(\cdot)$. 

We start with the proof of the monotonicity. 
\begin{lemma}\label{lem:bhc-monotone}
    Both $\mu_{G,\lambda,\beta}^{\textnormal{HC}}$ and $\mu_{L,\lambda,\beta}^{\textnormal{HC}}$ are monotone systems.  
    \end{lemma}
    
    \begin{proof}[Proof of \Cref{lem:bhc-monotone}]
    We first consider $\mu_{G,\lambda,\beta}^{\textnormal{HC}}$.
    For simplicity, let $\mu = \mu_{G,\lambda,\beta}^{\textnormal{HC}}$.
    Fix any $v\in V_L$ and any feasible and comparable pinnings $\sigma,\tau\in \{0,1\}^{(V_L\cup V_R)\setminus \{v\}}$ with $\sigma\preceq\tau$. 
    Recall that the value $1$ at $v$ indicates that $v$ is not in the independent set.
    If all neighbors of $v$ in $\sigma$ are fixed to be 0, $\mu_{v}^{\sigma}(0)=\frac{\lambda}{1+\lambda}$. Otherwise, $\mu_{v}^{\sigma}(0)=0$. By $\sigma \preceq \tau$, $\sigma_w \leq \tau_w$ for all $w\in V_R$. We have
    \begin{align*}
        \mu_{v}^{\sigma}(0) = \frac{\lambda}{1+\lambda}\cdot\*1\{\sigma_w=0,\forall (v,w)\in E\} \geq \frac{\lambda}{1+\lambda}\cdot\*1\{\tau_w=0,\forall (v,w)\in E\}=\mu_{v}^{\tau}(0).
    \end{align*}
    Hence, $\mu_v^\sigma$ is stochastically dominated by $\mu_v^\tau$.
    If $v\in V_R$, we can show stochastic dominance by a similar argument. This proves the result for $ \mu_{G,\lambda,\beta}^{\textnormal{HC}}$.
    
    Let $\nu = \mu_{L,\lambda,\beta}^{\textnormal{HC}}$ be the marginal distribution on $V_L$.
    Fix a vertex $v\in V_L$. For any feasible and comparable pinnings $\sigma,\tau\in \{0,1\}^{V_L\setminus \{v\}}$ with $\sigma\preceq \tau$, we only need to prove that $\nu_{v}^{\sigma}(0)\geq \nu_{v}^{\tau}(0)$. Let $R_v$ be the set of neighbors of $v$. Let $R_\sigma:=\{w\in V_R:\exists u\in V_L\setminus\{v\}, \sigma_u=0\land(u,w)\in E\}$ and $R_\tau:=\{w\in V_R:\exists u\in V_L\setminus\{v\}, \tau_u=0\land(u,w)\in E\}$. 
    In words, $R_\sigma$ is the set of vertices in $V_R$ that are connected to a vertex $u$ that is fixed to be in the independent set by $\sigma$.
    Since $\sigma\preceq \tau$, it is easy to verify $R_\tau \subseteq R_\sigma$. The probability $\nu_{v}^{\sigma}(0)$ ($v$ is in the independent set conditional on $\sigma$) is
    \begin{align*}
        \nu_{v}^{\sigma}(0)&=\frac{\mu_{G,\lambda,\beta}^{\textnormal{HC}}(\sigma\land v\gets 0)}{\mu_{G,\lambda,\beta}^{\textnormal{HC}}(\sigma\land v\gets 0)+\mu_{G,\lambda,\beta}^{\textnormal{HC}}(\sigma\land v\gets 1)}=\frac{\lambda(1+\beta)^{|V_R\setminus (R_\sigma \cup R_v)|}}{\lambda(1+\beta)^{|V_R\setminus (R_\sigma \cup R_v)|}+(1+\beta)^{|V_R\setminus R_\sigma |}}\\
        &=\frac{\lambda}{\lambda+(1+\beta)^{|R_v\setminus R_{\sigma}|}}.
    \end{align*}
    The reason for the above equation is that 
    if $v\gets 0$, then $v$ is in the independent set and the vertices in $R_\sigma \cup R_v$ must lie outside the independent set, and the vertices in $V_R\setminus (R_\sigma \cup R_v)$ are free. The case $v\gets 1$ can be analyzed similarly. The same argument can be applied to $\nu_{v}^{\tau}(0)$. Since $R_\tau \subseteq R_\sigma$, $|R_v\setminus R_\sigma |\leq |R_v\setminus R_\tau|$. We have
    \begin{align*}
        \nu_{v}^{\sigma}(0)=\frac{\lambda}{\lambda+(1+\beta)^{|R_v\setminus R_{\sigma}|}}\geq \frac{\lambda}{\lambda+(1+\beta)^{|R_v\setminus R_{\tau}|}}=\nu_{v}^{\tau}(0).
    \end{align*}
    Hence, $\nu_v^\sigma$ is stochastically dominated by $\nu_v^\tau$. This proves the result for $\mu_{L,\lambda,\beta}^{\textnormal{HC}}$.
    \end{proof}

We now set the parameters $\theta$ and $\alpha(\cdot)$ for \Cref{thm:main-comparison} to prove \Cref{thm:bhc-mixing-time-marginal}. Define
\begin{align*}
    \theta=\frac{\lambda}{\e^9 (1+\lambda)^\Delta \Delta\log n}.
\end{align*}
Define the function 
$\alpha:[0,\log\frac{1}{\theta}]\to \mathbb{R}_+$ by
\begin{align}\label{eq:alpha-t-hc}
   \alpha(t)=\begin{cases}
       \frac{10^4(1+\lambda)^{5\Delta}}{\delta
       }  &\textnormal{if } t\in[0,\log \frac{1}{\theta_0});\\ 
       2\cdot 10^4 (1+\lambda)^{5\Delta}& \textnormal{if } t\in[\log\frac{1}{\theta_0},\log \frac{1}{\theta}],
   \end{cases} \quad\text{where } \theta_0 = \frac{\lambda}{e^{\e^9}}.
\end{align}

Consider the marginal distribution $\mu^{\textnormal{HC}}_{L,\lambda,\lambda}$ in \Cref{thm:bhc-mixing-time-marginal}. To apply \Cref{thm:main-comparison}, we need to bound $\mixtl(\mu^{\textnormal{HC}}_{L,\lambda,\lambda}, \theta,\epsilon)$, which requires us to bound the mixing time of Glauber dynamics on $(\theta * \mu^{\textnormal{HC}}_{L,\lambda,\lambda})^{\*1_\Lambda}$ for any $\Lambda \subseteq V_L$. By the mapping in \eqref{eq:bhc-mapping}, we have that $(\theta * \mu^{\textnormal{HC}}_{L,\lambda,\lambda})^{\*1_\Lambda}$ is the distribution of $(\mu^{\textnormal{HC}}_{G,\theta^{-1}\lambda,\lambda})^{\*1_\Lambda}$  projected onto $V_L$. In other words, we first increase the external field on $V_L$ by a factor of $\theta^{-1} > 1$, then fix all vertices in $\Lambda$ to be \emph{outside} the independent set, and finally take the marginal distribution by projecting onto $V_L$. Compared to the original distribution $\mu^{\textnormal{HC}}_{G,\lambda,\lambda}$, the tilted distribution has a larger external field on $V_L$.
For $\ell=|V_L\setminus\Lambda|\geq2$, $O(\ell\log(\ell/\epsilon))$ updates at unpinned vertices suffice for TV error $\epsilon/2$ by~\cite[Lem.~25]{Chen0Y23} (one update suffices for $\ell=1$, none for $\ell=0$).
For $\ell>0$, the number of such updates in $t$ steps has law $\mathrm{Bin}(t,\ell/|V_L|)$, so a binomial tail bound gives
\begin{align}
 \forall\,0<\epsilon<1, \quad   \mixtl(\mu_{L,\lambda,\lambda}^{\textnormal{HC}} ,\theta, \epsilon)= O\tp{n \log \frac{n}{\epsilon}}. \label{eq:bhc-mixing-tilted-marginal}
\end{align}

We prove the following entropic independence result for $\mu_{L,\lambda,\lambda}^{\textnormal{HC}}$.
\begin{lemma}\label{lem:HC-EI}
    $\forall\, t\in [0,\log \frac{1}{\theta}]$, for any $\Lambda \subseteq V_L$,  $(\e^{-t} *\mu_{L,\lambda,\lambda}^{\textnormal{HC}})^{\*1_\Lambda}$ is $\alpha(t)$-entropically independent for $\alpha$ defined in \eqref{eq:alpha-t-hc}.
    \end{lemma}

    \Cref{lem:HC-EI} improves the entropic independence result in \cite{Chen0Y23}, which is achieved by a more refined analysis of spectral independence. The detailed proof is deferred to Appendix \ref{sec:HC-EI-proof}.

Let us assume \Cref{lem:HC-EI} holds and prove the result for the bipartite hardcore model.

\begin{proof}[Proof of \Cref{thm:bhc-mixing-time-marginal}]
If $V_L=\varnothing$, the marginal is a point mass and the mixing time is zero; assume otherwise. Fix the absolute constant $C=\e^{\e^9}$.
Define a parameter $S$ as follows:
\begin{align*}
    S=S(\lambda, \Delta)=(1+\lambda)^\Delta\leq\ \left(1+\frac{(\Delta-1)^{\Delta-1}}{(\Delta-2)^\Delta}\right)^\Delta\leq \left(1+\frac{\e(\Delta-1)}{(\Delta-2)^2}\right)^\Delta\leq \left(1+\frac{6\e}{\Delta}\right)^\Delta\leq \e^{6\e}.
\end{align*}
In \Cref{thm:main-comparison}, the parameter $T$ is 
\begin{align*}
    T=O\left(\exp\left(\int_0^{\log\frac{1}{\theta}}4\alpha(t) \d t \right)\cdot\left(\log\log\frac{1}{({\mu}_{L,\lambda,\lambda}^{\textnormal{HC}})_{\min}}+\log\frac{1}{\epsilon}\right)\right).
\end{align*}
Since $\delta<1/10$, the first value in~\eqref{eq:alpha-t-hc} is larger, so the following bound holds regardless of the order of $\theta$ and $\theta_0$:
\begin{align*}
    &\exp\left ( \int_0^{\log \frac{1}{\theta}}4\alpha(t)\d t \right ) \leq\exp\left (4 \log \frac{1}{\theta_0}\cdot \frac{10^4S^5}{\delta} +4\log \frac{\theta_0}{\theta} \cdot 2\cdot10^4S^5 \right )\\
    &=\left(\frac{\e^{\e^9}}{\lambda}\right)^{\frac{4\cdot10^4S^5}{\delta}}\cdot \left( \frac{S\e^9\Delta\log n}{\e^{\e^9}}\right )^{8\cdot10^4S^5}\leq\left(\frac{C}{\lambda} \right)^{O(1/\delta)} (\Delta\log n)^{O(1)}.
\end{align*}
Since $\lambda\leq\lambda_c(\Delta)\leq4$, the minimum probability of the marginal satisfies
\[
(\mu_{L,\lambda,\lambda}^{\textnormal{HC}})_{\min}
\geq\left(\frac{\min\{1,\lambda\}}{1+\lambda}\right)^n
\geq\left(\frac{\lambda}{C}\right)^n.
\]
Thus $\log\log\bigl(1/(\mu_{L,\lambda,\lambda}^{\textnormal{HC}})_{\min}\bigr)\leq\log n+\log\log(C/\lambda)$.
As $C/\lambda>\e$, factors polynomial in $\delta^{-1}$ and $\log(C/\lambda)$ can be absorbed into $(C/\lambda)^{O(1/\delta)}$.
Hence, the parameter $T$ is bounded by
\begin{align*}
    %O\left(\left(\frac{C}{\lambda} \right)^{O(1/\delta)} (\Delta\log n)^{O(1)}\cdot\left(\log\log\frac{1}{(\bar{\mu}_{L,G,\lambda}^{\textnormal{HC}})_{\min}}+\log\frac{1}{\epsilon}\right)\right)
    T=\left(\frac{C}{\lambda} \right)^{O(1/\delta)}\cdot (\Delta\log n)^{O(1)}\cdot \log \frac{1}{\epsilon}.
\end{align*}
Using \Cref{thm:main-comparison} and \eqref{eq:bhc-mixing-tilted-marginal}, the mixing time of the Glauber dynamics on $\mu_{L,\lambda,\lambda}^{\textnormal{HC}}$ from the all-1 state $\*1_{V_L}$ (the empty set) satisfies the following bound
\begin{align*}
    T^{\textnormal{GD}}_{\textnormal{mix-}\*1_{V_L}}(\mu_{L,\lambda,\lambda}^{\textnormal{HC}},\eps)&\leq T\cdot \mixtl(\mu_{L,\lambda,\lambda}^{\textnormal{HC}},\theta,\frac{\epsilon}{2T})
    %&=\left(\frac{C}{\lambda} \right)^{O(1/\delta)}\cdot (\Delta\log n)^{O(1)}\cdot \log \frac{1}{\epsilon}\cdot 21n(\log n +\log \frac{T}{\epsilon})\\
    \leq n \cdot \left(\frac{C}{\lambda} \right)^{O(1/\delta)}(\Delta\log n)^{O(1)}\log^2 \frac{1}{\eps}. \qedhere
\end{align*}
\end{proof}

We have proved a $\tilde{O}_{\Delta,\lambda}(n)$ bound on the mixing time of Glauber dynamics on the marginal distribution. Using the classical censoring inequality~\cite{peres2013can}, we can prove that Glauber dynamics for the full distribution $\mu^{\textnormal{HC}}_{G,\lambda,\lambda}$ has mixing time $\tilde{O}_{\Delta,\lambda}(n^2)$. %We use following classical censoring inequality.

\begin{proof}[Proof of \Cref{thm:bhc-mixing-time}]
Let $\mu=\mu_{G,\lambda,\lambda}^{\textnormal{HC}}$ for simplicity. If $V_L=\varnothing$, $\mu$ is a product measure and Glauber dynamics mixes in $O(n\log(n/\epsilon))$ steps, which satisfies the claimed bound; assume otherwise.
%\begin{align*}
%X_i=\begin{cases}
%   1 & \textnormal{if } i \in R;\\
%   0 & \textnormal{if } i \in L.
%\end{cases}    
%\end{align*}
Using \Cref{thm:bhc-mixing-time-marginal} with $\delta/10$ in place of $\delta$,
let 
\[N=n\cdot \left(\frac{C}{\lambda} \right)^{O(1/\delta)}(\Delta\log n)^{O(1)}\cdot \log^2 (1/\eps)\]
be a positive integer upper bound on the mixing time to within total variation distance $\frac{\epsilon}{2}$ of stationarity for Glauber dynamics on the marginal distribution, started from $X^L = \*1_{V_L}$.
 In each update step, given the current state $X^L$, the Glauber dynamics picks a vertex $v \in V_L$ uniformly at random and updates the value of $v$ conditional on $X^L_{V_L \setminus \{v\}}$. This Glauber dynamics is equivalent to the following Markov chain on the entire graph. Starting from $X = \*1_{V_L \cup V_R}$ (by the mapping defined in~\eqref{eq:bhc-mapping}, the maximum state $\*1_{V_L \cup V_R}$ is the independent set with all vertices in $V_R$), in each step, the Markov chain picks $v \in V_L$ uniformly at random and then resamples $X_{\{v\} \cup V_R}$ from $\mu$ conditional on $X_{V_L \setminus \{v\}}$. After each such update, the conditional law on $V_R$ given $X_{V_L}$ is the target conditional law, so the full and marginal TV distances coincide. Given the condition $X_{V_L \setminus \{v\}}$, by self-reducibility, we can remove all vertices in $V_L \setminus \{v\}$ and all vertices $w \in V_R$ such that $w$ has a neighbor $w' \in V_L\setminus\{v\}$ with $X_{w'} = 0$. The remaining graph has maximum degree (which is the degree of $v$) at most $\Delta = \Delta_L$ on both sides. Since $\lambda \leq (1-\delta)\lambda_c(\Delta)$, the conditional distribution is a hardcore model in the uniqueness regime.  By the mixing result in~\cite{CFYZ22,ChenE22}, we can sample from $\mu_{\{v\} \cup V_R}^{X_{V_L \setminus \{v\}}}$ within TV-distance $\frac{\epsilon}{2N}$ by running the Glauber dynamics on $\mu^{X_{V_L \setminus \{v\}}}$\footnote{$\mu^{X_{V_L \setminus \{v\}}}$
 is a distribution over $\{0,1\}^{V_L \cup V_R}$ such that the configuration on $V_L \setminus \{v\}$ is fixed to  $X_{V_L \setminus \{v\}}$ and the configuration on $V_R \cup \{v\}$ follows the conditional distribution.} 
 for $N' = \left\lceil\exp(O(1/\delta))\,n \log \frac{Nn}{\epsilon}\right\rceil$ steps, with the explicit dependence on $\delta$ given in~\cite{CFYZ22}. Hence, we have the following process.
 Starting from the maximum state $X = \*1_{V_L \cup V_R}$,
 \begin{itemize}
    \item for each $i$ from $1$ to $N$;
\begin{itemize}
    \item pick a vertex $v \in V_L$ uniformly at random;
    \item for each $j$ from $1$ to $N'$
\begin{itemize}
    \item pick a vertex $w \in  V_L \cup V_R$ uniformly at random;
    \item if $w \in \{v\} \cup V_R$, resample $X_w \sim \mu_w^{X_{V \setminus \{w\}}}$.
\end{itemize}
\end{itemize} 
 \end{itemize}
A simple coupling argument shows that the output $X$ satisfies $\DTV{X}{\mu}\leq \frac{\epsilon}{2} + N \frac{\epsilon}{2N} = \epsilon$.

%Consider the standard Glauber dynamics on $\mu$ starting from the maximal state $X = \*1_{V_L \cup V_R}$ and updating for $M =NN'$ steps. Suppose the update sequence of vertices is $w_1,w_2,\ldots,w_M$, where every $w_i$ is picked uniformly at random from $V_L \cup V_R$. Let the output after $M$ steps is $X'$. 
Consider the following standard censored Glauber dynamics. 
Starting from the maximum state $X = \*1_{V_L \cup V_R}$,
 \begin{itemize}
    \item for each $i$ from $1$ to $N$;
\begin{itemize}
    \item pick a vertex $v \in V_L$ uniformly at random;
    \item for each $j$ from $1$ to $N'$
\begin{itemize}
    \item pick a vertex $w \in  V_L \cup V_R$ uniformly at random;
    \item if $w \in \{v\} \cup V_R$, resample $X_w \sim \mu_w^{X_{V \setminus \{w\}}}$; otherwise, ignore the update at $w$.
\end{itemize}
\end{itemize} 
\end{itemize}

The above censored Glauber dynamics is the same as the Markov chain defined above. 
The update sites and censoring decisions can be sampled in advance, independently of the randomness used to resample spins.
Consider the standard Glauber dynamics on $\mu$ starting from the maximal all-$1$ state. Let $X'$ be the output of the standard Glauber dynamics after $N N'$ updates. By \Cref{lem:bhc-monotone} and the randomized version of the censoring inequality in~\cite[Thm.~1.1]{peres2013can}, we have
\begin{align*}
    \DTV{X'}{\mu} \leq \DTV{X}{\mu} \leq \epsilon.
\end{align*}
By~\eqref{eq:bhc-mapping}, $\*1_{V_L \cup V_R}$ is the independent set $V_R$.
Since $C/\lambda>\e$, the factor $\exp(O(1/\delta))$ and the parameter dependence in $\log N$ are absorbed into $(C/\lambda)^{O(1/\delta)}$. Hence, the mixing time of Glauber dynamics starting from the independent set $V_R$ is
\begin{align*}
    T^{\textnormal{GD}}_{\textnormal{mix-}V_R}(\mu,\eps)\leq NN' \leq n^2(\Delta\log n)^{O(1)}(C/\lambda)^{O(1/\delta)}\cdot \log^3 \frac{1}{\epsilon}. &\qedhere
\end{align*}
\end{proof}

\ifthenelse{\boolean{conf}}{}{\paragraph{Acknowledgment}
We thank Xiaoyu Chen for helpful discussions about the bipartite hardcore model and the localization scheme.}

\begin{remark}[Corrections to the SODA~2026 version]\label{rem:conference-version}
The proofs in the SODA~2026 version of this paper (the original arXiv version) had two gaps.
First, the original proof for the random cluster model gave a mixing time bound involving $\log\log(1/\mu_{\min})$, which can be arbitrarily large as an edge parameter $p_e$ approaches~$1$.
Since we can choose the initial state, we start from the full edge set $E$.
The initial relative entropy is then $\log(1/\mu_{p,\lambda}^{\textnormal{RC}}(E))\leq n\log 2 + m\log(1/p_{\min})$.
Under the spin flip used in the comparison, $E$ becomes the all-$0$ state, so the initial law has a non-increasing density with respect to the flipped stationary law.
We therefore extended the original comparison argument, which assumed an increasing density, to cover this case.
Second, even if $P$ and $Q$ are both reversible with respect to $\mu$, their product $PQ$ need not be reversible.
In the original proof, we incorrectly assumed that $PQ$ was reversible whenever $P,Q\in\+{MC}_\mu$.
We have revised the related proofs so that they no longer rely on the reversibility of $PQ$.
All main results of the paper remain unchanged. The issues above were found using ChatGPT, and the revisions were made with AI assistance.
\end{remark}

\bibliographystyle{alpha}
\bibliography{refs}

\appendix

\section{Proofs for the comparison inequality relation}\label{sec:proofs-for-comparison-inequality-relation}

The following two observations follow directly from \Cref{def:stochastic-dominance}.
\begin{observation}\label{obs:stochastic-dominance-properties}
The stochastic dominance relation is preserved under convex combinations. If $\nu_0 \sd \pi_0$ and $\nu_1 \sd \pi_1$, then for any $\lambda \in [0,1]$, $\lambda \nu_0 + (1-\lambda) \nu_1 \sd \lambda \pi_0 + (1-\lambda) \pi_1$.

The stochastic dominance relation is transitive. If $\pi_0 \sd \pi_1$ and $\pi_1 \sd \pi_2$, then $\pi_0 \sd \pi_2$.
\end{observation}

The first result follows from the coupling definition of stochastic dominance and the second result follows from the expectation-based definition of stochastic dominance.

\subsection{Proof of \texorpdfstring{\Cref{lem:comparison-inequality-properties}}{}}

Fix a distribution $\nu$ such that $\frac{\nu}{\mu}$ is increasing. To prove the first result, we need to show that $\lambda \nu P_0 + (1-\lambda) \nu P_1 \sd \lambda \nu Q_0 + (1-\lambda) \nu Q_1$, which follows from the first item in \Cref{obs:stochastic-dominance-properties}.

For the second result, let
$\rho_j=\nu P_1\cdots P_j$ and $\eta_j=\nu Q_1\cdots Q_j$
for $0\leq j\leq\ell$, with $\rho_0=\eta_0=\nu$.
By \Cref{lem:reversible-transition},
\[
\frac{\rho_{j+1}}{\mu}=P_{j+1}\left(\frac{\rho_j}{\mu}\right).
\]
Thus $\rho_j/\mu$ remains increasing by the stochastic monotonicity of each $P_i$.
Starting from $\rho_0=\eta_0$, we obtain inductively, for $0\leq j<\ell$,
\[
\rho_{j+1}=\rho_jP_{j+1}
\sd\rho_jQ_{j+1}
\sd\eta_jQ_{j+1}=\eta_{j+1},
\]
where the first inequality follows from $P_{j+1}\mc Q_{j+1}$,
and the second from the induction hypothesis and the stochastic monotonicity of $Q_{j+1}$.
Hence $\rho_\ell\sd\eta_\ell$, as required.

\subsection{Proof of \texorpdfstring{\Cref{lem:gen-compare}}{}}
Let $A=\{\sigma \in \Omega:\pi(\sigma)>\mu(\sigma)\}$. The indicator function $\*1_A:\Omega \to \{0,1\}$ where $\*1_A(\sigma) = \*1_{\sigma \in A}$ is  increasing, because $\frac{\pi}{\mu}$ is increasing. If not, then there exist $\sigma \preceq \tau$, where $\sigma,\tau \in \Omega$, such that $\sigma \in A$ and $\tau \notin A$, which implies that $\frac{\pi(\sigma)}{\mu(\sigma)}> 1 \geq \frac{\pi(\tau)}{\mu(\tau)}$. This contradicts the fact that $\frac{\pi}{\mu}$ is increasing. Therefore, $\*1_A$ is increasing. We have   
\begin{align*}
\DTV{\pi}{\mu}=\sum_{\sigma\in A}(\pi(\sigma)-\mu(\sigma))&=\E[X\sim \pi]{\*1_A(X)}- \E[X\sim \mu]{\*1_A(X)}   \\
\text{(since~$\pi \sd \pi'$ and $\*1_A$ is increasing)}\quad &\leq \E[X\sim \pi']{\*1_A(X)}- \E[X\sim \mu]{\*1_A(X)}\\
&\leq \DTV{\pi'}{\mu},
\end{align*}
where the last inequality holds because $\DTV{\pi'}{\mu}=\max_{A'\subseteq \Omega}(\pi'(A')-\mu(A'))$.

The second result is a consequence of the first result, \Cref{def:comparison-inequality-relation}, \Cref{def:stochastically-monotone}, and \Cref{lem:reversible-transition}.

\subsection{The equivalence of two definitions of the comparison inequality relation}
We finally prove the equivalence of the two definitions of the comparison inequality relation. 
This equivalence allows us to apply the results in \cite{FillK13}, providing alternative proofs of \Cref{lem:comparison-inequality-properties} and \Cref{lem:gen-compare}.
For $P,Q\in\+{MC}_\mu$, consider the following inner-product formulation: $P\mc Q$ if and only if, for all nonnegative increasing functions $f,g$,
\begin{align}\label{eq:fill-kahn-def-comparison-inequality-relation}  
\langle Pg, f \rangle_\mu \leq \langle Qg, f \rangle_\mu.
\end{align}

For any distribution $\nu$ and function $g$,
\[
\left\langle Pg,\frac{\nu}{\mu}\right\rangle_\mu
=\mathbb E_{\nu P}[g].
\]
Taking $f=\nu/\mu$ for increasing $\nu/\mu$ shows that
\eqref{eq:fill-kahn-def-comparison-inequality-relation} implies
\Cref{def:comparison-inequality-relation}.
Conversely, for any nonzero nonnegative increasing $f$,
apply \Cref{def:comparison-inequality-relation} to
$\nu=f\mu/\mathbb E_\mu[f]$ and multiply by
$\mathbb E_\mu[f]$. The case $f=0$ is immediate.

The increasing and decreasing conventions agree since,
for sufficiently large constants $a,b$,
\[
\langle(P-Q)(b-g),a-f\rangle_\mu
=\langle(P-Q)g,f\rangle_\mu,
\]
using $P\mathbf1=Q\mathbf1=\mathbf1$ and $\mu P=\mu Q=\mu$.

\section{Counterexample for direct comparison}\label{sec:counterexample}
In this section, we show that
$\nu P_{\pi\textnormal{-GD}}\sd\nu P_{\textnormal{alg}}^{(t)}$
need not hold for every $t\geq 1$ and every distribution $\nu$
such that $\nu/\pi$ is increasing.
Consider the case where $(t-1) \text{ mod } T_2 = 0$. By~\eqref{eq:recurrence-relation-simulation-algorithm}, $P^{(t)}_{\textnormal{alg}} = P^{}_{\textnormal{cl}} \psim $. Let $\sigma$ denote the configuration in which all vertices take the value $\star$. Hence, $\sigma$ is the maximum configuration of $\Omega(\pi)$. Suppose $\nu$ takes the value $\sigma$ with probability 1. It is easy to see that $\frac{\nu}{\pi}$ is an increasing function. For any $i \in V$, let $p_i$ denote the probability that $i$ takes the value $\star$ in $\pi$ conditional on all other vertices taking the value $\star$. By~\eqref{eq:pi-1-i}, $p_i = (1-\theta)\mu_i^{\*1_{V \setminus \{i\}}}(1)$.
We have 
\begin{align*}
    (\nu P_{\pi\textnormal{-GD}})(\sigma) = \frac{1}{|V|}\sum_{i \in V} p_i.  
\end{align*}
Consider the distribution $\nu P^{}_{\textnormal{cl}} \psim$. Suppose $X \sim \nu$ and we apply $P^{}_{\textnormal{cl}}$ to $X$ to obtain $Y$ and we apply $\psim$ to $Y$ to obtain $Z$. Note that $Z \sim \nu P^{}_{\textnormal{cl}} \psim$. By the definition of $\psim$, $Z = \sigma$ if and only if $Y = \sigma$. By the definition of $P^{}_{\textnormal{cl}}$, we have 
\begin{align*}
    (\nu P^{}_{\textnormal{cl}} \psim)(\sigma) = (1-\theta)^{|V|}. 
\end{align*}
If $\nu P_{\pi\textnormal{-GD}}\sd\nu P_{\textnormal{cl}}\psim$,
then, since $\sigma$ is the maximum configuration of $\Omega(\pi)$, we must have
\begin{align*}
    (\nu P^{}_{\textnormal{cl}} \psim)(\sigma) = (1-\theta)^{|V|} \geq \frac{1}{|V|}\sum_{i \in V} p_i = (\nu P_{\pi\textnormal{-GD}})(\sigma).
\end{align*}
However, the LHS is exponentially small but the RHS is the average of $p_i$. It is easy to find a monotone system $\mu$ and a constant $\theta \in (0,1)$ to make the above inequality fail.

\begin{remark}
The counterexample shows that requiring stochastic dominance for every distribution
$\nu$ with increasing density $\nu/\pi$ is a strong condition.
A direct comparison restricted to the distributions actually generated by the chains might still be possible.
Our modification of the Glauber dynamics yields the comparison for all such $\nu$ in \Cref{lem:phase-comparison}.
\end{remark}

\section{Mixing of field dynamics from entropic independence}\label{sec:mixing-of-field-dynamics-from-entropic-independence}

In this section, we recall some key definitions and proof steps in \cite{Chen0Y23} and explain the modification needed to prove \Cref{thm:SIandMS-KL-Decay}.

The theorem is proved by the following ``negative-field'' localization process in \cite{ChenE22}.
Let $\mu$ be a distribution over $\Omega \subseteq \{0,1\}^V$. A \emph{negative-field localization process} $(\mu_t)_{t\geq 0}$ for $\mu$ is a continuous-time stochastic process defined by $\mu_t:=(\e^{-t}*\mu)^{\*1_{R_t}}$ for all $t\geq 0$, where $\mu_t$ is a random distribution. The pinning $\*1_{R_t}$ is the all-one pinning on a random set $R_t\subseteq V$, which is defined by another stochastic process.
Let $R_0=\emptyset$. Suppose that $(R_t)_{0\leq t\leq t_\star}$ has been generated for some time $t_\star\geq 0$.
If $R_{t_\star}=V$, we set $R_t=R_{t_\star}$ for all $t>t_\star$.
Otherwise, we iteratively generate the next stopping time $\tau>t_\star$ and define the process $(R_t)_{t_\star<t\leq \tau}$ as follows:  
\begin{itemize}
\item Let $(T_i)_{i\in V\setminus R_{t_\star}}$ be mutually independent random variables taking values in $\mathbb{R} \cup \{+\infty\}$ and defined by 
\begin{align*}
    \Pr[]{T_i> s}:=\exp\left(-\int_{t_\star}^s (\e^{-r} * \mu)_i^{\*1_{R_{t_\star}}}(1) \d r\right),\quad \forall s\geq t_\star;
\end{align*}
\item
Let $\tau:=\min_{i\in V\setminus R_{t_\star}}T_i$ and $J \in V\setminus R_{t_\star}$ such that $T_J = \tau$. If $\tau = +\infty$, let $R_t=R_{t_\star}$ for all $t>t_\star$. Otherwise, add $J$ to $R_{t_\star}$ at time $\tau$. Equivalently, 
\begin{align*}
         R_t = \begin{cases}
            R_{t_\star} & \text{ if } t_\star<t<\tau;\\
            R_{t_\star}\cup \{J\} & \text{ if } t = \tau.
        \end{cases}
    \end{align*}
\end{itemize}

For any distribution $\nu$ over $\Omega$ and any $f:\Omega\to \mathbb{R}_{\geq 0}$, the entropy of $f$ with respect to $\nu$ is defined by $\textnormal{Ent}_\nu[f]:=\E[\nu]{f\log f}-\E[\nu]{f}\log \E[\nu]{f}$, where $\E[\nu]{f\log f} = \sum_{\sigma\in \Omega} \nu(\sigma) f(\sigma)\log f(\sigma)$.

\begin{lemma}[\text{\cite[Prop. 41]{ChenE22}}]\label{lem:Ent-t-t+h}
Let $\mu$ be a distribution over $\{0,1\}^V$. Let $(\mu_t)_{t\geq 0}$ be a negative-field localization process for $\mu$. If $\mu_t$ is $\alpha$-entropically independent for some $\alpha>0$ and $t \geq 0$, then for any function $f:\Omega \to \mathbb{R}_{\geq 0}$, as $h\downarrow0$, we have $\E[]{\textnormal{Ent}_{\mu_{t+h}}[f]\mid \mu_t} \geq \textnormal{Ent}_{\mu_t}[f](1-4h\alpha)+o(h)$.
\end{lemma}

The following proof follows from the analysis in \cite[Sec. 9]{Chen0Y23}.

\begin{proof}[Proof of \Cref{thm:SIandMS-KL-Decay}]
Let $(\mu_t)_{t\geq 0}$ be a negative-field localization of $\mu$. By assumption, for any possible $R\subseteq V$ and $t\in [0,-\log \theta]$, $\mu_t=(\e^{-t}*\mu)^{\*1_R}$ is $\alpha(t)$-entropically independent. Then, by \Cref{lem:Ent-t-t+h}, for any function $f:\Omega\to \mathbb{R}_{\geq 0}$, as $h\downarrow0$ with $t+h\leq-\log\theta$,
\begin{align}\label{eq:ent-t-t+h}
    \E[]{\textnormal{Ent}_{\mu_{t+h}}[f]\mid \mu_t}\geq \textnormal{Ent}_{\mu_{t}}[f](1-4\alpha(t)h) +o(h).
\end{align}
The function $F(t):=\E[]{\textnormal{Ent}_{\mu_t}[f]}$ is absolutely continuous. Taking expectations in~\eqref{eq:ent-t-t+h} gives
\[
F'(t)\geq-4\alpha(t)F(t)\qquad\text{for almost every }t\in[0,-\log\theta].
\]
By Gr\"onwall's inequality,
\[
\E[]{\textnormal{Ent}_{\mu_{-\log\theta}}[f]}
\geq\exp\left(-\int_0^{-\log\theta}4\alpha(t)\d t\right)\textnormal{Ent}_\mu[f]
=\kappa\,\textnormal{Ent}_\mu[f].
\]
By the analysis in \cite{Chen0Y23}, the above inequality implies
$\mathrm{D}_{\textnormal{KL}}(\nu P_{\theta,\mu}\parallel \mu P_{\theta,\mu})\leq (1-\kappa)\mathrm{D}_{\textnormal{KL}}(\nu \parallel \mu)$.
Using Pinsker's inequality, we can bound the TV distance by the KL-divergence. See \cite[Sec. 9]{Chen0Y23} for more details of how to bound the TV distance.
\end{proof}

\section{Proof of entropic independence}

\subsection{Preliminaries}
The following spectral independence condition was introduced by Anari, Liu, and Oveis Gharan~\cite{ALO20}, which is closely related to entropic independence.
\begin{definition}[$\ell_\infty$-spectral independence \cite{ALO20}]
Let $\eta \geq 1$ be a constant.
A distribution $\mu$ over $\{0,1\}^V$ is said to be $\eta$-$\ell_\infty$-spectrally independent if for any $\Lambda \subseteq V$ with $|\Lambda| \leq |V| - 2$ and any feasible pinning $\sigma \in \{0,1\}^{\Lambda}$, the $\ell_\infty$-norm of the influence matrix $\Psi_\mu^\sigma$ is at most $\eta$:
\begin{align*}
  \Vert \Psi_\mu^\sigma \Vert_\infty =  \max_{u \in V} \sum_{v \in V} |\Psi_\mu^\sigma(u,v)| \leq \eta,
\end{align*} 
where the influence matrix $\Psi_\mu^\sigma \in \mathbb{R}^{V \times V}$ is defined by 
\begin{align*}
\forall u,v \in V, \quad \Psi_\mu^\sigma(u,v) \defeq \begin{cases}
\mu^{\sigma \land (u \gets 1)}_v(1) - \mu^{\sigma \land (u \gets 0)}_v(1) &\text{ if } \mu_u^\sigma(0),\mu_u^\sigma(1)>0;\\
0 & \text{ otherwise},
\end{cases}
\end{align*}
where $\mu^{\sigma \land (u \gets x)}_v(\cdot)$ is the marginal distribution at $v$ conditional on $\sigma$ and $u$ being set to $x$.
Thus the diagonal entry is $1$ for every nonconstant coordinate.
\end{definition}

By \cite{AJKPV22} and \cite{AASV21}, a uniform bound on spectral independence under arbitrary external fields implies a bound on entropic independence under arbitrary external fields.

\begin{lemma}[\text{\cite[Lem. 69]{AJKPV22} and \cite[Thm. 4]{AASV21}}]\label{lem:equiv-EI-SI}
    Let $\alpha \geq 1$.
Let $\mu$ be a distribution over $\{0,1\}^V$ with $V = [n]$. If for any $\lambda=(\lambda_1,\cdots,\lambda_n)\in \mathbb{R}_{> 0}^{V}$, $(\lambda*\mu)$ is $\alpha$-$\ell_\infty$-spectrally independent, then for any $\lambda'=(\lambda_1',\cdots,\lambda_n')\in \mathbb{R}_{> 0}^{V}$, $(\lambda'*\mu)$ is $\alpha+1$-entropically independent under all pinnings.
\end{lemma}

Next, we derive entropic independence from spectral independence and marginal stability.
We need the following marginally stable condition introduced in~\cite{CFYZ22,ChenE22} and the definition of tilted distributions.

\begin{definition}[Marginally stable \cite{CFYZ22,ChenE22}]
Let $K \geq 1$. A distribution $\mu$ over $\{0,1\}^V$ is said to be $K$-marginally stable if for any $\Lambda \subseteq V$, any feasible pinning $\tau \in \{0,1\}^{\Lambda}$, any $S \subseteq \Lambda$, and any vertex $v \in V \setminus \Lambda$, the following holds:
\begin{align*}
 R^\tau_v \leq K \cdot R^{\tau_S}_v \text{ and } \mu^\tau_v(0) \geq K^{-1},
\end{align*}
where $R^{\tau}_v = {\mu^\tau_v(1)}/{\mu^\tau_v(0)}$ is the marginal ratio and $R^{\tau_S}_v$ is defined analogously by replacing the pinning $\tau$ with a smaller pinning $\tau_S$ on a subset $S \subseteq \Lambda$.
\end{definition}

\begin{lemma}[\text{\cite[Thm. 67]{ChenE22}}]\label{lem:SIandMStoEI}
Let $\mu$ be a distribution on $\{0,1\}^V$, $\eta\geq1$, and $K\geq2$. If $\mu$ is $\eta$-$\ell_\infty$-spectrally independent and $K$-marginally stable, then $\mu$ is $384\eta K^4$-entropically independent under all pinnings.
\end{lemma}
This follows from \cite[Thm. 67]{ChenE22} with its parameter $C=K/2\geq1$ after encoding $0,1$ as $-1,+1$.

Finally, for applications, sometimes it is easier to obtain the following coupling independence condition, which is a sufficient condition for spectral independence.

\begin{definition}[Coupling independence]
    A distribution $\mu$ over $\{0,1\}^V$ is $C$-coupling independent if for any $\Lambda\subseteq V$ with $|\Lambda|\leq |V|-2$, $i\in V$ and any feasible pinning $\sigma\in \{0,1\}^\Lambda$ such that $\mu^{\sigma}_i(0),\mu^{\sigma}_i(1)>0$, there exists a coupling $(X,Y)$ between $\mu^{\sigma\land (i\gets 0)}$ and $\mu^{\sigma\land (i\gets 1)}$ such that 
    \begin{align*}
        \E[]{\abs{X\oplus Y}}\leq C.
    \end{align*}
\end{definition}

The following folklore lemma follows from the standard coupling inequality.

\begin{lemma}[\text{\cite[Prop. 4.3]{CZ23}}]\label{lem:CI-SI}
If a distribution $\mu$ over $\{0,1\}^V$ is $C$-coupling independent for some $C\geq1$, then $\mu$ is $C$-$\ell_\infty$-spectrally independent.
\end{lemma}

In the proofs, we often consider the flipped distribution $\bar{\mu}$ of $\mu$ in \Cref{def:flipped-distribution}. The following observation is easy to verify from the definition of the flipped distribution.
\begin{observation}\label{ob:CI}
The following holds for any distribution $\mu$ over $\{0,1\}^V$:
\begin{itemize}
    \item $\mu$ is $\eta$-$\ell_\infty$-spectrally independent if and only if $\bar{\mu}$ is $\eta$-$\ell_\infty$-spectrally independent.
    \item $\mu$ is $C$-coupling independent if and only if $\bar{\mu}$ is $C$-coupling independent.
    \item $\mu$ is $\alpha$-entropically independent if and only if $\bar{\mu}$ is $\alpha$-entropically independent.
\end{itemize}
\end{observation}
The first two items are trivial from the definitions. For the third item, note that~\eqref{eq:entropic-independence} holds for $\mu$ and $\nu$ if and only if it holds for $\bar{\mu}$ and $\bar{\nu}$.

\subsection{Entropic independence of the random cluster model}\label{sec:RC-EI-proof}

We now prove \Cref{lem:RC-EI}.
%y the definitions of entropic independence and flipped distribution, a distribution $\pi$ is entropically independent with parameter $\alpha$ if and only if the flipped distribution $\bar{\pi}$ is $\alpha$-entropically independent. Hence, it is equivalent to show that $(\e^{t} *{\mu}_{p,\lambda}^{\textnormal{RC}})$ is $\alpha(t)$-entropically independent, where we change $\e^{-t}$ to $e^t$ because flipping changes the roles of $0$ and 1. 
For $0 \leq t \leq \log\frac{1}{\theta_0}$, we can use the following result by Chen and Zhang \cite{CZ23}.

\begin{lemma}[\text{\cite[Lem. 4.4 \& 4.5]{CZ23}}]\label{lem:RC-CI}
For any vector $r \in (0,\infty)^E$,
$(r * \mu_{p,\lambda}^{\textnormal{RC}})$ is $2(1-\lambda_{\max})^{-2}$-coupling independent.
\end{lemma}
The bound includes the conditioned edge, since $1+(1+\lambda_{\max})^2/(2(1-\lambda_{\max})^2)\leq2(1-\lambda_{\max})^{-2}$.

Note that $r * \bar{\mu}_{p,\lambda}^{\textnormal{RC}} = \overline{r^{-1} * \mu_{p,\lambda}^{\textnormal{RC}}}$, where $r^{-1}_e  = \frac{1}{r_e}$ for all $e \in E$.
By \Cref{lem:RC-CI}, \Cref{ob:CI}, and \Cref{lem:CI-SI}, $(r * \bar{\mu}_{p,\lambda}^{\textnormal{RC}})$ is $2(1-\lambda_{\max})^{-2}$-$\ell_\infty$-spectrally independent for all $r \in (0,\infty)^E$.
By \Cref{lem:equiv-EI-SI}, $(r * \bar{\mu}_{p,\lambda}^{\textnormal{RC}})$ is $(2(1-\lambda_{\max})^{-2}+1)$-entropically independent for all $r \in (0,\infty)^E$, and $2(1-\lambda_{\max})^{-2}+1\leq3(1-\lambda_{\max})^{-2}$.
In particular, this implies the $\alpha(t)$-entropic independence result for $0 \leq t \leq \log\frac{1}{\theta_0}$ in \Cref{lem:RC-EI}, where the pinning on $\Lambda$ can be achieved by letting $r_e \to \infty$ for all $e \in \Lambda$.

Next, we give an improved bound for entropic independence for $(\e^{-t}*\bar{\mu}_{p,\lambda}^{\textnormal{RC}})^{\*1_\Lambda}$ when $\log\frac{1}{\theta_0} < t \leq \log \frac{1}{\theta}$.
We first prove the entropic independence for $(\e^{-t}*\bar{\mu}_{p,\lambda}^{\textnormal{RC}})$ without pinning. At the end of the proof, we show how to generalize the result to the case with pinning.

We need to define another model, the subgraph-world model, to help us analyze the entropic independence of the random cluster model.

\begin{definition}[Subgraph-world model]
Let $G=(V,E)$ be a graph, and let $p\in [0,1)^E$ and $\eta \in [0,1]^V$ be parameters. The subgraph-world model $\mu_{p,\eta}^{\textnormal{SW}}$ specified by $G,p$, and $\eta$ is defined by
\begin{align*}
    \forall S\subseteq E,\quad \mu_{p,\eta}^{\textnormal{SW}}(S)\propto \prod_{e\in S} \frac{p_e}{1-p_e} \prod_{v\in V:\abs{S\cap E_v}\equiv 1 (\textnormal{mod }2)}\eta_v,
\end{align*}
where $E_v$ is the set of edges incident to vertex $v$.
\end{definition}

Previous works \cite{GJ18,feng2023swendsen} show the following relation between the subgraph-world model and the random cluster model.

\begin{lemma}[\text{\cite[Lem. 4.5]{feng2023swendsen}}]\label{lem:random-var-rc-sw}
Let $\mu_{p,\lambda}^{\textnormal{RC}}$ be the random cluster model with $p\in [0,1]^E$ and $\lambda\in [0,1)^V$. Let $\mu_{p',\eta}^{\textnormal{SW}}$ be the subgraph-world model with $p'=\frac{p}{2}$ and $\eta_v =\frac{1-\lambda_v}{1+\lambda_v}$ for all $v \in V$.
Let $X\sim \mu_{p',\eta}^{\textnormal{SW}}$. 
Let $Z \subseteq E$, independent of $X$, be a random subset that includes each $i \in E$ independently with probability $q_i$, where $q_i = \frac{p_i}{2-p_i}$.
Then the random variable $Y:=X\cup Z \sim \mu_{p,\lambda}^{\textnormal{RC}}$.
\end{lemma} %Equivalently, we have 
%begin{align*}
%    \forall Y\subseteq E,\quad \mu_{p,\lambda}^{\textnormal{RC}}(Y)=\sum_{X\subseteq Y}\mu_{p',\eta}^{\textnormal{SW}}\prod_{E\in Y\setminus X} q_e \prod_{f\in E\setminus Y}(1-q_f).
%\end{align*}

Chen and Zhang \cite{CZ23} show the following coupling independence result.
\begin{lemma}[\text{\cite[Lem. 4.4]{CZ23}}]\label{lem:sw-CI}
For $p\in[0,1/2]^E$ and $\eta\in(0,1]^V$, the subgraph-world model $\mu_{p,\eta}^{\textnormal{SW}}$ is $(1+(2\eta_{\min}^2)^{-1})$-coupling independent, where $\eta_{\min}=\min_{v\in V}\eta_v$.
\end{lemma}
The additional $1$ counts the conditioned edge; common pinnings only change the parity constraints in the proof of~\cite[Lem. 4.6]{CZ23}.

Using the subgraph model, we prove the following new coupling-independence result for the random cluster model. Compared to \Cref{lem:RC-CI}, the following result improves the coupling independence bound. However, this result does not hold for all tilted distributions.
\begin{lemma}\label{lem:CI-rc-large-lambda}
    For the random cluster model $\mu_{p,\lambda}^{\textnormal{RC}}$ specified by $G,p\in (0,1]^E$, and $\lambda\in [0,1)^V$, $\mu_{p,\lambda}^{\textnormal{RC}}$ is $(1 +4(1-p_{\min})(1-\lambda_{\max})^{-2})$-coupling independent.    
\end{lemma}
\begin{proof}
    Following the definition in \Cref{lem:random-var-rc-sw},
    let $p'=\frac{p}{2}, \eta_v = \frac{1-\lambda_v}{1+\lambda_v}$ for all $v \in V$, and $q_e=\frac{p_e}{2-p_e}$ for all $e \in E$.
    Fix an edge $e\in E$ for which both pinnings $e\gets0$ and $e\gets1$ are feasible under $\mu_{p,\lambda}^{\textnormal{RC}}$.
    Let $(X_0,X_1)$ be the coupling of $(\mu_{p',\eta}^{\textnormal{SW}})^{e\gets 0}$ and $(\mu_{p',\eta}^{\textnormal{SW}})^{e\gets 1}$ that achieves the coupling independence bound in \Cref{lem:sw-CI}:
    \begin{align}\label{eq:sw-CI-bound}
        \E[]{|X_0 \oplus X_1|}-1\leq \frac{1}{2\eta_{\min}^2}=\frac{(1+\lambda_{\max})^2}{2(1-\lambda_{\max})^2}\leq 2(1-\lambda_{\max})^{-2}.
    \end{align}
    Here, we use $X_0 \oplus X_1$ to denote the symmetric difference of $X_0$ and $X_1$. Note that $|X_0 \oplus X_1|$ is the Hamming distance if we view both $X_0$ and $X_1$ as binary vectors in $\{0,1\}^E$.

    We construct a coupling $(Y_0,Y_1)$ of $(\mu_{p,\lambda}^{\textnormal{RC}})^{e \gets 0}$ and $(\mu_{p,\lambda}^{\textnormal{RC}})^{e \gets 1}$ based on $(X_0,X_1)$.
    The construction of the coupling is guided by the process in \Cref{lem:random-var-rc-sw}. 
    To generate a sample $Y \sim \mu_{p,\lambda}^{\textnormal{RC}}$ from the random cluster model, one can first sample $X \sim \mu_{p',\eta}^{\textnormal{SW}}$ from the subgraph-world model, then sample $Z \sim \otimes_{i\in E}\textnormal{Ber}(q_i)$ from independent Bernoulli trials, and take $Y = X \cup Z$.
    The event $e \notin Y$ occurs if and only if $e \notin X$ and $e \notin Z$. 
    Hence, $Y_0 \sim (\mu_{p,\lambda}^{\textnormal{RC}})^{e \gets 0}$ could be generated by letting $Y_0 = (X_0\cup Z)\setminus\{e\}$, where $X_0 \sim (\mu_{p',\eta}^{\textnormal{SW}})^{e \gets 0}$ is from the coupling in~\eqref{eq:sw-CI-bound}. 
    Next, consider the event $e \in Y$, which happens if and only if either $e \in X$ or $e \notin X \land e \in Z$. 
    Conditional on $e \in Y$, suppose that $e \notin X \land e \in Z$ happens with probability $t_e \in [0,1]$.
    Then $Y_1 \sim (\mu_{p,\lambda}^{\textnormal{RC}})^{e \gets 1}$ could be generated as follows: (1) with probability $t_e$, let $Y_1 = (X_0\cup Z)\cup \{e\}$, where $X_0 \sim (\mu_{p',\eta}^{\textnormal{SW}})^{e\gets 0}$ is from the coupling in~\eqref{eq:sw-CI-bound}; (2) with probability $1-t_e$, let $Y_1 = (X_1\cup Z)$, where $X_1 \sim (\mu_{p',\eta}^{\textnormal{SW}})^{e\gets 1}$ is from the coupling in~\eqref{eq:sw-CI-bound}.
    Formally, we have the following coupling $(Y_0,Y_1)$, where $Z$ and the branch choice are mutually independent and independent of $(X_0,X_1)$:
    \begin{enumerate}
       \item Sample $Z\sim \otimes_{i\in E}\textnormal{Ber}(q_i)$ from independent Bernoulli trials;
       \item with probability $t_e$, let $Y_0=(X_0\cup Z)\setminus \{e\}$ and $Y_1=(X_0\cup Z)\cup \{e\}$;
       \item with probability $1-t_e$, let $Y_0=(X_0\cup Z)\setminus \{e\}$ and $Y_1=(X_1\cup Z)$.  
    \end{enumerate}
    Let us view $X_0,X_1,Y_0,Y_1,Z$ as binary vectors in $\{0,1\}^E$.
    Note that for any edge $f \in E$ with $f \neq e$, $Y_0(f) \neq Y_1(f)$ only if $X_0(f) \neq X_1(f)$ and $Z(f) = 0$ in the above coupling.
    Let $q_{\min} = \min_{e\in E} q_e$.
    Now we can bound $\E[]{|Y_0\oplus Y_1|}$ as follows:
    \begin{align*}
        \E[]{|Y_0\oplus Y_1|}&=1+(1-t_e)\sum_{f\ne e}(1-q_f)\Pr[]{X_0(f)\ne X_1(f)}\\
        &\leq1+(1-t_e)(1-q_{\min})\bigl(\E[]{|X_0\oplus X_1|}-1\bigr)\\
        &\leq1+4(1-p_{\min})(1-\lambda_{\max})^{-2}.
     \end{align*}
      Any feasible pinning is obtained by deleting edges pinned to $0$ and setting $p_f=1$ on edges pinned to $1$. This does not decrease $p_{\min}$ or change $\lambda_{\max}$, so the same bound holds under all pinnings; if no unpinned edges remain, the claim is trivial.
     Then $\mu_{p,\lambda}^{\textnormal{RC}}$ is $(1 +4(1-p_{\min})(1-\lambda_{\max})^{-2})$-coupling independent.
    \end{proof}

Recall that our goal is to show that $\e^{-t}*\bar{\mu}_{p,\lambda}^{\textnormal{RC}}$ is $\alpha(t)$-entropically independent when $\log\frac{1}{\theta_0} < t \leq \log \frac{1}{\theta}$. By the definition of flipping, we have that the distribution $\e^{-t}*\bar{\mu}_{p,\lambda}^{\textnormal{RC}}$ is the same as the distribution $\bar{\mu}_{p_t,\lambda}^{\textnormal{RC}}$, where for any edge $f \in E$, the parameter $p_t(f)$ satisfies
\begin{align}\label{eq:para-pt}
    \frac{p_t(f)}{1-p_t(f)} = \frac{p(f)}{1-p(f)} \e^t.
\end{align}
Note that $e^t > \frac{1}{\theta_0} = \frac{2}{p_{\min}(1-\lambda_{\max})^2}$. Hence, for any $f \in E$, $\frac{p_t(f)}{1 - p_t(f)} > \frac{p_{\min}}{1-p_{\min}} \cdot \frac{1}{p_{\min}(1-\lambda_{\max})^2}$. Note that $0 < p_t(f) < 1$. We have the following bound 
\begin{align*}
\forall f \in E,\quad  1 - p_t(f) \leq \frac{1-p_t(f)}{p_t(f)} \leq (1 - p_{\min})(1-\lambda_{\max})^2 \leq (1-\lambda_{\max})^2.
\end{align*}
Hence, $1 - \min_{f \in E}p_t(f) < (1-\lambda_{\max})^2$.
Note that $1 + 4(1-\lambda_{\max})^2(1-\lambda_{\max})^{-2} = 5$.
By \Cref{lem:CI-rc-large-lambda}, $\mu_{p_t,\lambda}^{\textnormal{RC}}$ is 5-coupling independent. 
By \Cref{ob:CI}, $\bar{\mu}_{p_t,\lambda}^{\textnormal{RC}}$ is also 5-coupling independent. Since $\e^{-t}*\bar{\mu}_{p,\lambda}^{\textnormal{RC}} \equiv \bar{\mu}_{p_t,\lambda}^{\textnormal{RC}}$, we have the following corollary by \Cref{lem:CI-SI}.
\begin{corollary}\label{cor:EI-rc-large-lambda}
For any $t > \log\frac{1}{\theta_0}$, $\e^{-t}*\bar{\mu}_{p,\lambda}^{\textnormal{RC}}$ is 5-$\ell_\infty$-spectrally independent.
\end{corollary}

To obtain the entropic independence for $\e^{-t}*\bar{\mu}_{p,\lambda}^{\textnormal{RC}}$ when $t > \log\frac{1}{\theta_0}$, we need to use \Cref{lem:SIandMStoEI}.
By \Cref{cor:EI-rc-large-lambda}, we already established the spectral independence. 
We only need to verify the following marginal stability for the flipped distribution $\e^{-t}*\bar{\mu}_{p,\lambda}^{\textnormal{RC}} \equiv \bar{\mu}_{p_t,\lambda}^{\textnormal{RC}}$.
We give the following result for general random cluster models.

\begin{lemma}\label{lem:rc-ms}
For any $p \in [\frac{2}{3},1)^E$ and $\lambda \in [0,1)^V$,
$\bar{\mu}_{p,\lambda}^{\textnormal{RC}}$ is $2$-marginally stable.   
\end{lemma}

\begin{proof}%[Proof of \Cref{lem:rc-ms}]
We view $\bar{\mu}_{p,\lambda}^{\textnormal{RC}}$ as a distribution over $2^E$ and $\{0,1\}^E$ interchangeably.
Fix arbitrary $\Lambda\subseteq E$, feasible pinning $\tau\in \{0,1\}^{\Lambda}$, subset $S\subseteq \Lambda$ and edge $e=\{u,v\}\in E\setminus \Lambda$. For any $A\subseteq E\setminus\{e\}$, let $C_u, C_v$ denote the sets of vertices of the connected components in $G_A=(V,A)$ containing $u$ and $v$, respectively. Note that $C_u=C_v$ if $u$ and $v$ are in the same connected component. Let $D_u=\prod_{w\in C_u}\lambda_w<1$ and $D_v=\prod_{w\in C_v}\lambda_w<1$.
Let $A'=A\cup\{e\}$.
There are two cases: (1) $u$ and $v$ are in the same connected component in $G_A$; (2) they are in different connected components in $G_A$ and the same component in $G_{A'}$. For the first case,
\begin{align*}
    \mu_{p,\lambda}^{\textnormal{RC}}(A)/ \mu_{p,\lambda}^{\textnormal{RC}}(A') = \frac{1-p_e}{p_e}\frac{1+D_u}{1+D_u} = \frac{1-p_e}{p_e}.
\end{align*}
For the second case, by $0\leq D_u,D_v<1$ we have
\begin{align*}
 \mu_{p,\lambda}^{\textnormal{RC}}(A)/ \mu_{p,\lambda}^{\textnormal{RC}}(A') = \frac{1-p_e}{p_e} \frac{(1+D_u)(1+D_v)}{1+D_u D_v} \in \left[\frac{1-p_e}{p_e},\frac{2(1-p_e)}{p_e}\right).
\end{align*}
In conclusion, $\frac{1-p_e}{p_e}\leq\mu_{p,\lambda}^{\textnormal{RC}}(A)/ \mu_{p,\lambda}^{\textnormal{RC}}(A')<\frac{2(1-p_e)}{p_e}$ in any case.
%For any configuration $\sigma$, $S_\sigma = \{f \in E \mid \sigma_f = 1\}$.
Summing over pairs $A,A'$ compatible with $\bar\tau$ gives
\begin{align*}
    (\bar{\mu}_{p,\lambda}^{\textnormal{RC}})_e^{\tau}(0)&=(\mu_{p,\lambda}^{\textnormal{RC}})_e^{\bar{\tau}}(1) = \frac{\sum_{\sigma \in \{0,1\}^E: \sigma_e = 1 \land \sigma_\Lambda = \bar{\tau} }\mu_{p,\lambda}^{\textnormal{RC}}(\sigma)}{\sum_{\sigma \in \{0,1\}^E: \sigma_e = 0 \land \sigma_\Lambda = \bar{\tau} }\mu_{p,\lambda}^{\textnormal{RC}}(\sigma)+\sum_{\sigma \in \{0,1\}^E: \sigma_e = 1 \land \sigma_\Lambda = \bar{\tau} }\mu_{p,\lambda}^{\textnormal{RC}}(\sigma)}\\
     &> \frac{1}{\frac{2(1-p_e)}{p_e}+1}= \frac{p_e}{2-p_e}\geq \frac{1}{2},
\end{align*}
The last inequality holds because $2/3\leq p_e\leq1$. By a similar argument, we have
\begin{align*}
    \frac{(\bar{\mu}_{p,\lambda}^{\textnormal{RC}})_e^{\tau}(1)/(\bar{\mu}_{p,\lambda}^{\textnormal{RC}})_e^{\tau}(0)}{(\bar{\mu}_{p,\lambda}^{\textnormal{RC}})_e^{\tau_S}(1)/(\bar{\mu}_{p,\lambda}^{\textnormal{RC}})_e^{\tau_S}(0)}=\frac{(\mu_{p,\lambda}^{\textnormal{RC}})_e^{\bar{\tau}}(0)/(\mu_{p,\lambda}^{\textnormal{RC}})_e^{\bar{\tau}}(1)}{(\mu_{p,\lambda}^{\textnormal{RC}})_e^{\bar{\tau_S}}(0)/(\mu_{p,\lambda}^{\textnormal{RC}})_e^{\bar{\tau_S}}(1)}&<\frac{2(1-p_e)}{p_e}/\frac{1-p_e}{p_e} = 2.
\end{align*}
Hence, $\bar{\mu}_{p,\lambda}^{\textnormal{RC}}$ is $2$-marginally stable.
\end{proof}

%As we need, the monotonicity of the random cluster model and the rapid mixing of ``good case'' of the random cluster model are both shown in previous results.
%We can further prove $\bar{\mu}_{p,\lambda}^{\textnormal{RC}}$ is also monotone.

By~\eqref{eq:para-pt}, the parameter $p_t(f)$ of $\bar{\mu}_{p_t,\lambda}^{\textnormal{RC}}\equiv(\e^{-t}*\bar{\mu}_{p,\lambda}^{\textnormal{RC}})$ satisfies $\frac{p_t(f)}{1-p_t(f)} = \frac{p(f)}{1-p(f)} e^t$. We have $p_t(f)\geq \frac{2}{3}$ if and only if $\e^t\geq \frac{2(1-p(f))}{p(f)}$. Therefore, $(\e^{-t} * \bar{\mu}_{p,\lambda}^{\textnormal{RC}})$ is $2$-marginally stable when $t\geq \log \frac{1}{\theta_0}\geq\log \frac{2}{p_{\min}}$. Combining \Cref{lem:SIandMStoEI}, \Cref{cor:EI-rc-large-lambda}, and \Cref{lem:rc-ms}, for any $t \in (\log\frac{1}{\theta_0},\log \frac{1}{\theta}]$, the distribution $(\e^{-t} * \bar{\mu}_{p,\lambda}^{\textnormal{RC}})$ is $384 * 5 * 2^4\leq 5*10^4$-entropically independent. This finishes the proof of entropic independence for $(\e^{-t}*\bar{\mu}_{p,\lambda}^{\textnormal{RC}})$ without pinning.

For $\Lambda=E$, the pinned distribution is a point mass, so the claim is trivial. Otherwise, note that $(\e^{-t} * \bar{\mu}_{p,\lambda}^{\textnormal{RC}})^{\*1_\Lambda}$ is a flipped random cluster model defined on the subgraph $G=(V,E\setminus \Lambda)$. We remark that we consider the flipped distribution so that fixing the value to be 1 is removing the edge from the graph. Note that $\min_{e \in E \setminus \Lambda} p_e \geq p_{\min}$. Going through the above analysis in the subgraph finishes the proof of the case with pinning.

\subsection{Entropic independence of the bipartite hardcore model}\label{sec:HC-EI-proof}

In this section, we prove \Cref{lem:HC-EI}. We prove the entropic independence by verifying the uniqueness condition. 
Note that $(\e^{-t} *\mu_{L,\lambda,\lambda}^{\textnormal{HC}})^{\*1_\Lambda}$ can be viewed as the left marginal of a bipartite hardcore model on a subgraph. Formally, let $G'=G[(V_L \setminus \Lambda) \cup V_R]$ be the induced subgraph. Let $\mu_{G',\lambda,\lambda}^{\textnormal{HC}}$ be the bipartite hardcore model specified by $G',\lambda,\lambda$.
Note that $\mu_{G',\lambda,\lambda}^{\textnormal{HC}}$  also satisfies the left-side uniqueness condition.
Let $\mu_{L',\lambda,\lambda}^{\textnormal{HC}}$ be the marginal distribution on $V_L \setminus \Lambda$ projected from $\mu_{G',\lambda,\lambda}^{\textnormal{HC}}$. Note that in distribution, $(\e^{-t} *\mu_{L,\lambda,\lambda}^{\textnormal{HC}})^{\*1_\Lambda}$ fixes all variables in $\Lambda$ to be $1$ and all other variables in $V_L \setminus \Lambda$ follow the distribution $\e^{-t} *\mu_{L',\lambda,\lambda}^{\textnormal{HC}}$. Hence, we only need to verify the entropic independence for $\e^{-t}*\mu_{L,\lambda,\lambda}^{\textnormal{HC}}$ without pinning. The case with pinning can be reduced to the same analysis on the subgraph.

%Because for a bipartite hardcore model with some feasible pinning, it is easy to verify that the distribution is equivalent to another bipartite hardcore model with smaller dimension and the same parameter constraints. Therefore, Our above analysis that only consider the case of original distribution without pinning is enough.

Compared to the classical uniqueness threshold $\lambda_c(\Delta)\defeq \frac{(\Delta-1)^{\Delta-1}}{(\Delta-2)^\Delta}$, we need a refined definition of the uniqueness threshold of the bipartite hardcore model in~\cite{Chen0Y23}.

\begin{definition}[\cite{Chen0Y23}]
Let $F_{\lambda,d,\beta,w}(x):=\lambda(1+\beta(1+x)^{-w})^{-d}$, where $\lambda,d,\beta,w>0$.
For any $\delta\in [0,1)$, a tuple $(\lambda,d,\beta,w)\in \mathbb{R}_{+}^4$ is said to be $\delta$-unique, if for all fixed points $\hat{x}\geq 0$ satisfying $F(\hat{x})=\hat{x}$, it holds that $F'(\hat{x})\leq 1-\delta$.      

Furthermore, for any $\delta\in [0,1)$, a tuple $(\lambda,d,\beta)\in \mathbb{R}_{+}^3$ is said to be $\delta$-unique,   if $(\lambda,d,\beta,w)$ is $\delta$-unique for all $w>0$.
\end{definition}

The above definition has two parameters $d$ and $w$. Intuitively, one can think of $d+1$ as the maximum degree on $V_L$ and $w+1$ as the maximum degree on $V_R$. The following uniqueness condition is proved in \cite{Chen0Y23}.

\begin{lemma}[\text{\cite[Thm. 2 \& Thm. 12]{Chen0Y23}}]\label{lem:unique-simple}
In the case of $\lambda = \beta$, for any $\Delta=d+1\geq 3$ and $\delta\in [0,1)$, $(\lambda,d,\lambda)$ is $\frac{\delta}{10}$-unique if $\lambda\leq (1-\delta)\lambda_c(\Delta) = (1-\delta)\frac{(\Delta-1)^{\Delta-1}}{(\Delta-2)^\Delta}$.

More generally, for any $\delta\in [0,1)$, $\beta>0$ and $d\geq 1-\delta$,
\begin{enumerate}
    \item if $\beta \leq \frac{1-\delta}{d}\e ^{1+\frac{1-\delta}{d}}$, then for all $\lambda>0$, $(\lambda,d,\beta)$ is $\delta$-unique;
    \item  if $\beta > \frac{1-\delta}{d}\e ^{1+\frac{1-\delta}{d}}$, there exists $\lambda_c>0$ such that $(\lambda,d,\beta)$ is $\delta$-unique if and only if $\lambda \geq \lambda_c$, where $\lambda_c:=x_c(1+\beta(1+x_c)^{-w_c})^d$ and $(x_c,w_c)$ is the unique positive solution of
    \begin{align}\label{eq:unique-condition}
            \begin{cases}
                (1-\delta)(1+x)(\beta+(1+x)^w)-\beta d w x=0,\\
                w \log (1+x) (\beta d-(1+x)^{w+1}) + \delta(1+x)(\beta+(1+x)^w)=0.
            \end{cases}
    \end{align}
    \end{enumerate}
\end{lemma}

The following relation between the uniqueness condition and the spectral independence for the bipartite hardcore model is known.

\begin{lemma}[\text{\cite[App.~B, Eq.~(83)]{Chen0Y23}}]\label{lem:bhc-unique-si}
Let $\mu_{G,\lambda,\beta}^{\textnormal{HC}}$ be a bipartite hardcore model specified by $G=((V_L,V_R),E),\lambda,\beta$, where $\lambda>0$ is the external field on $V_L$ and $\beta>0$ is the external field on $V_R$. Let $\Delta = d+1\geq 3$ be an upper bound on the degrees in $V_L$. For any $\delta\in (0,1)$, if $(\lambda,d,\beta)$ is $\delta$-unique, the marginal distribution $\mu^\textnormal{HC}_{L,\lambda,\beta}$ on $V_L$ is $\frac{\Delta(1+\beta)^\Delta}{d \cdot \delta}$-$\infty$-spectrally independent.
\end{lemma}

Consider the model $\mu_{G,\lambda,\lambda}^{\textnormal{HC}}$ (where $\lambda=\beta$) with $\lambda\leq (1-10\delta)\lambda_c(\Delta)$ and its marginal distribution $\mu_{L,\lambda,\lambda}^{\textnormal{HC}}$ on the left side $V_L$. By \Cref{lem:unique-simple}, for any $0 < \lambda_0 \leq 1$, note that $(\lambda_0 * \mu_{L,\lambda,\lambda}^{\textnormal{HC}})\equiv \mu_{L,\lambda_0^{-1}\cdot\lambda,\lambda}^{\textnormal{HC}}$. The parameter tuple $(\lambda_0^{-1}\lambda, \Delta - 1, \lambda)$ is $\delta$-unique, where the left external field changes to $\lambda_0^{-1}\cdot\lambda$ by the mapping in~\eqref{eq:bhc-mapping}. We have 
\begin{corollary}\label{cor:bhc-unique-si}
    For any $0 < \lambda_0 \leq 1$, $(\lambda_0 * \mu_{L,\lambda,\lambda}^{\textnormal{HC}})$ is   $\frac{\Delta(1+\lambda)^\Delta}{(\Delta - 1) \cdot \delta}$-$\infty$-spectrally independent.
\end{corollary}

We next prove the improved spectral independence in~\eqref{eq:alpha-t-hc} when $t \geq \log\frac{1}{\theta_0}$.

%by \Cref{lem:bhc-unique-si}, we need a constant-unique property. Our below analysis of~\eqref{eq:unique-condition} focuses on $\mu_{L,G,\lambda_0\cdot\lambda,\lambda}^{\textnormal{HC}}$ when $\lambda_0$ is large enough.

\begin{lemma}\label{lem:bhc-large-lambda-unique}
Fix a bipartite hardcore model $\mu_{G,\lambda',\lambda'}^{\textnormal{HC}}$, where $G=((V_L,V_R),E)$, the degrees in $V_L$ are at most $\Delta=d+1\geq 3$, and $\lambda'< \lambda_c(\Delta)$. For any $\lambda_0\geq \frac{\e^{\e^9}}{\lambda'}$, the distribution $(\lambda_0^{-1} * \mu_{L,\lambda',\lambda'}^{\textnormal{HC}})\equiv \mu_{L,\lambda_0\cdot\lambda',\lambda'}^{\textnormal{HC}}$ is $\frac{2\Delta(1+\lambda')^\Delta}{\Delta-1   }$-$\infty$-spectrally independent. 
\end{lemma}
\begin{proof}
We first simplify the uniqueness condition in \Cref{lem:unique-simple} for $\mu_{G,\lambda,\beta}^{\textnormal{HC}}$. 
Note that in \Cref{lem:bhc-large-lambda-unique}, we consider the case when $\lambda = \lambda_0 \cdot \lambda'$ and $\beta = \lambda'$. 
We show that $\mu_{G,\lambda,\beta}^{\textnormal{HC}}$ is $\frac{1}{2}$-unique if $\lambda \geq \e^{\e^9}$ and $\beta \leq (1-\delta_0)\lambda_c(\Delta)$ for any $\delta_0 > 0$. By the assumption in \Cref{lem:bhc-large-lambda-unique}, the distribution $(\lambda_0^{-1} * \mu_{L,\lambda',\lambda'}^{\textnormal{HC}})\equiv \mu_{L,\lambda_0\cdot\lambda',\lambda'}^{\textnormal{HC}}$ is $\frac{1}{2}$-unique. Then by \Cref{lem:bhc-unique-si}, the distribution $(\lambda_0^{-1} * \mu_{L,\lambda',\lambda'}^{\textnormal{HC}})$ is $\frac{2\Delta(1+\lambda')^\Delta}{\Delta-1   }$-$\infty$-spectrally independent.

To verify the uniqueness condition in \Cref{lem:unique-simple}, we only need to focus on the second case.
We calculate the value of $\lambda_c$ in the uniqueness condition when $\delta=\frac{1}{2}$.
We show that $\lambda_c \leq \e^{\e^9}$ if $\beta < \lambda_c(\Delta)$.
Hence, if $\beta < \lambda_c(\Delta)$ and $\lambda \geq \e^{\e^9}$, the distribution $\mu_{G,\lambda,\beta}^{\textnormal{HC}}$ is $\frac{1}{2}$-unique.

The first equation of~\eqref{eq:unique-condition} implies
\begin{align}\label{eq:unique-eq1}
    \frac{\beta d w x}{1-\delta}=(1+x)(\beta+(1+x)^w)=(1+x)^{w+1}+\beta(1+x).
\end{align}
Substituting~\eqref{eq:unique-eq1} into the second equation gives
\begin{align}\label{eq:unique-w}
    &w\log (1+x)\left(\frac{\beta d w x}{1-\delta}-\beta(1+x)-\beta d\right) = \frac{\delta\beta d w x}{1-\delta};\notag\\
    \textnormal{divided by }\beta dw\Rightarrow\quad &\log (1+x)\left(\frac{w x}{1-\delta}-\frac{1+x}{d}-1\right) = \frac{\delta x}{1-\delta};\notag\\
    \textnormal{divided by }\log (1+x)\Rightarrow\quad &\frac{wx}{1-\delta} = \frac{\delta x}{(1-\delta)\log (1+x)}+\frac{1+x}{d}+1;\notag\\
    \textnormal{multiplying }\frac{1-\delta}{x}\Rightarrow\quad &w=\frac{\delta }{\log (1+x)}+\frac{(1+x)(1-\delta)}{d x} + \frac{1-\delta}{x};\notag\\
        \overset{\text{set }\delta=\frac{1}{2}}{\Rightarrow} \quad&w=\frac{1}{2}\left(\frac{1}{\log (1+x)}+\frac{1+x}{d x} + \frac{1}{x}\right). 
\end{align}    
Substituting the value of $w$ and $\delta=\frac{1}{2}$ into the first equation of~\eqref{eq:unique-condition} gives
\begin{align*}
    &\frac{1}{2}(1+x)\left(\beta+(1+x)^{\frac{1}{2}\left(\frac{1}{\log (1+x)}+\frac{1+x}{d x} + \frac{1}{x}\right)}\right)=\beta dx\frac{1}{2}\left(\frac{1}{\log (1+x)}+\frac{1+x}{d x} + \frac{1}{x}\right)\\
    \Rightarrow\quad&(1+x)^{1+\frac{1}{2}\left(\frac{1}{\log (1+x)}+\frac{1+x}{d x} + \frac{1}{x}\right)}=\beta d\left(1+\frac{x}{\log (1+x)}\right).
\end{align*}
Since $\beta < \lambda_c(\Delta)$, we have $\beta d\leq \frac{d^d}{(d-1)^{d+1}}d\leq 8$. We have
\begin{align}\label{eq:x-c-ineq}
    1+x\leq (1+x)^{1+\frac{1}{2}\left(\frac{1}{\log (1+x)}+\frac{1+x}{d x} + \frac{1}{x}\right)}=\beta d\left(1+\frac{x}{\log (1+x)}\right) \leq 8\left(1+\frac{x}{\log(1+x)}\right).
\end{align}
We claim $\log(1+x)\leq 9$. Otherwise, $8(1+\frac{x}{\log(1+x)})<8 + \frac{8x}{9}\leq 1+x$, which contradicts~\eqref{eq:x-c-ineq}. Therefore,
\begin{align}\label{eq:x-c-bound}
    x_c\leq \e^{9}.
\end{align}
Note $\lambda_c = x_c(1+\beta(1+x_c)^{-w_c})^d$.
Consider another factor of $\lambda_c$. By~\eqref{eq:unique-eq1} and $\delta=\frac{1}{2}$,
\begin{align}\label{eq:lambda-bound}
\begin{split}
    (1+\beta(1+x_c)^{-w_c})^d&=(1+\frac{\beta}{\frac{\beta dw_c x_c}{\frac{1}{2}(1+x_c)}-\beta})^d\\
    \text{by~\eqref{eq:unique-w}}\quad&= (1+\frac{1}{\frac{dx_c}{(1+x_c)\log (1+x_c)}+\frac{d}{1+x_c}})^d\\
    &\leq \e^{\frac{1}{\frac{x_c}{(1+x_c)\log (1+x_c)}+\frac{1}{1+x_c}}}\leq \e^{\frac{1}{\frac{1}{1+x_c}+\frac{1}{1+x_c}}}\\
    &=\e^{\frac{1+x_c}{2}}\leq \e^{\frac{1+\e^9}{2}}.
\end{split}
\end{align}

We use~\eqref{eq:x-c-bound} and $x_c\geq\log(1+x_c)$ in the above inequalities. 
Combining~\eqref{eq:x-c-bound} and~\eqref{eq:lambda-bound}, we have 
\begin{align*}
    \lambda_c = x_c(1+\beta(1+x_c)^{-w_c})^d\leq \e^9\cdot \e^{\frac{1+\e^9}{2}} \leq \e^{\e^9}. &\qedhere
\end{align*}
\end{proof}

To establish entropic independence using \Cref{lem:SIandMStoEI}, we need the following lemma.
%The above lemma gives a slightly weaker result of marginally bounded, where at least one of the original distribution and the flipped distribution is $2S$-marginally bounded. The weak marginally bound result combining with spectral independence could still imply the entropic independence. 

\begin{lemma}[\text{\cite[Lem. 23]{Chen0Y23}}]\label{lem:bhc-marginal-bound}
For any $\lambda_0\in (0,1]$, let $\nu = \lambda_0 * \mu_{L,\lambda,\lambda}^{\textnormal{HC}}$. At least one of $\nu$ and $\bar{\nu}$ is $2(1+\lambda)^\Delta$-marginally stable.   
\end{lemma}

By \Cref{ob:CI}, a distribution $\nu$ is entropically (or spectrally) independent if and only if its flipped distribution $\bar{\nu}$ is entropically (or spectrally) independent. Therefore, the entropic independence of $(\e^{-t} * \mu_{L,\lambda,\lambda}^{\textnormal{HC}})$ can be proved by combining \Cref{lem:SIandMStoEI}, \Cref{cor:bhc-unique-si}, \Cref{lem:bhc-large-lambda-unique}, and \Cref{lem:bhc-marginal-bound}.
For the case $(\e^{-t} * \mu_{L,\lambda,\lambda}^{\textnormal{HC}})^{\*1_\Lambda}$ with pinning, as discussed before, we can go through the same proof on a subgraph. We retain the same degree bound $\Delta$ for these subgraphs.

\end{document}